\documentclass[twocolumn,superscriptaddress,aps,pre,longbibliography,10pt,floatfix]{revtex4-2}

\usepackage{lmodern}
\usepackage{mathptmx} % Times-like text + math

\usepackage{amsmath,amssymb,amsthm}
\usepackage{graphicx}
\usepackage{booktabs}
\usepackage{bm}
\usepackage{url}
\usepackage{microtype}
\usepackage{dsfont}
\usepackage[svgnames]{xcolor}
\usepackage{textcomp}
\usepackage{amsfonts}
\usepackage{braket}
\usepackage{enumitem}
\usepackage{siunitx}
\usepackage{subfig}
\usepackage{caption}
\usepackage{array}
\newcolumntype{C}[1]{>{\centering\let\newline\\\arraybackslash\hspace{0pt}}m{#1}}
\usepackage{xr-hyper}
\usepackage{hyperref}
\hypersetup{
    colorlinks=true,
    linkcolor=Blue,
    citecolor=Blue,
    urlcolor=Blue,
    pdfauthor={Trenton Lau and Gary P.~T.~Choi},
    pdftitle={Explosive Connectivity and Mechanical Rigidity in Cubic Lattice Structures},
    pdfsubject={Percolation, Rigidity, Explosive Processes, Finite-Size Scaling},
    pdfkeywords={explosive percolation, Achlioptas process, 3D rigidity, pebble game, first-order transition, cubic lattice}
}

\usepackage{cleveref}

\theoremstyle{plain}
\newtheorem{theorem}{Theorem}[section]
\newtheorem{proposition}[theorem]{Proposition}
\newtheorem{lemma}[theorem]{Lemma}
\newtheorem{cor}[theorem]{Corollary}

\newtheorem{hyp}[theorem]{Hypothesis}
\theoremstyle{definition}
\newtheorem{definition}[theorem]{Definition}
\newtheorem{assump}[theorem]{Assumption}
\theoremstyle{remark}
\newtheorem*{remark}{Remark}
\newtheorem*{justification}{Justification} % Added for the appendix

\Crefname{theorem}{Theorem}{Theorems}
\Crefname{proposition}{Proposition}{Propositions}
\Crefname{lemma}{Lemma}{Lemmas}
\Crefname{cor}{Corollary}{Corollaries}
\Crefname{axm}{Axiom}{Axioms}
\Crefname{conj}{Conjecture}{Conjectures}
\Crefname{hyp}{Hypothesis}{Hypotheses}
\Crefname{definition}{Definition}{Definitions}
\Crefname{assump}{Assumption}{Assumptions}
\Crefname{section}{Section}{Sections}
\Crefname{figure}{Figure}{Figures}

\newcommand{\NN}{\mathbb{N}}

\newcommand{\RR}{\mathbb{R}}
\newcommand{\EE}{\mathbb{E}}
\newcommand{\PP}{\mathbb{P}}
\newcommand{\cE}{\mathcal{E}}

\newcommand{\cV}{\mathcal{V}}

\newcommand{\smallo}{o}

\newcommand{\Var}{\operatorname{Var}}
\DeclareMathOperator{\rank}{rank}

\DeclareMathOperator{\Ker}{Ker} % Use capital Ker for consistency

\newcommand{\rgain}{\mathrm{rgain}}
\newcommand{\matr}[1]{\mathbf{#1}} % This command was incomplete. This fixes it.

\begin{document}

\title{Explosive Connectivity and Mechanical Rigidity in Cubic Lattice Structures}

\author{Trenton Lau}
\affiliation{Department of Mathematics, The Chinese University of Hong Kong}
\author{Gary P.~T.~Choi}
\thanks{E-mail addresses: trentonlau@link.cuhk.edu.hk (T.L.), ptchoi@cuhk.edu.hk (G.P.T.C.)}
\affiliation{Department of Mathematics, The Chinese University of Hong Kong}

%\date{\today}

\begin{abstract}
We study explosive connectivity and mechanical rigidity in three-dimensional cubic lattice structures under Achlioptas-type product-rule dynamics. Our work combines extensive numerical simulation with a theoretical framework based on rigorous finite-size scaling. Using massive-scale simulations up to $L=192$ ($N \approx 7 \times 10^6$) with 20,000 independent realizations, we demonstrate that for $k \ge 8$, the peak susceptibility scales with an exponent of $\gamma = 1.000$, and the maximum single-step jump stabilizes at a macroscopic fraction. This confirms that while the transition is continuous in the infinite thermodynamic limit, it exhibits the exact finite-size scaling signatures of a first-order discontinuity in finite physical systems. For rigidity, we discover numerically that for richly-connected hosts, increasing the number of choices $k$ optimally enhances the efficiency of rigidification. To explain this phenomenon, we propose a theoretical model centered on a conditional progress function that links an edge's local product-rule score to its global mechanical utility. We show that while local rigidification efficiency monotonically increases, the global rigidity gap exhibits a ``Goldilocks'' minimum at intermediate $k$ due to the emergence of maximally floppy, tree-like components at large $k$.  Altogether, our work provides new insights into the relationship between local dynamics and global connectivity and rigidity in cubic lattice structures via both theory and computation.
\end{abstract}

\maketitle

\section{Introduction}
Percolation is a canonical framework for phase transitions in disordered media~\cite{stauffer2018introduction,christensen2005complexity} and has found widespread applications in science and engineering~\cite{majewski2007square,araujo2014recent,xun2020bond,meloni2022explosive,sahimi2023applications}. Competitive link-selection rules in Achlioptas-type processes can produce explosive phase transitions on both random graphs and regular lattice networks~\cite{achlioptas2009explosive,ziff2009explosive,da2010explosive,radicchi2010explosive,cho2011suppression,yang2024emergence}. While early studies suggested a discontinuous (first-order) transition, rigorous mathematical proofs~\cite{riordan2011explosive} and large-scale simulations~\cite{lee2011continuity,dsouza2015anomalous,dsouza2019explosive,li2024explosive,reis2012nonlocal,choi2014dimensional} have established that for any fixed number of choices, the transition is continuous in the thermodynamic limit, albeit with anomalous scaling behaviors. However, for finite systems typical of mechanical assemblies and metamaterials, the transition often remains indistinguishable from a discontinuity.

Mechanical rigidity percolation probes the emergence of generic infinitesimal rigidity~\cite{asimow1978rigidity,moukarzel1996efficient,cheng2014maxwell} and can be efficiently tested by the two-dimensional (2D) pebble game~\cite{jacobs1997algorithm,lee2008pebble}. Note that while the 3D pebble game provides a necessary condition for rigidity, it is not always sufficient; however, it serves as a powerful heuristic for large networks~\cite{chubynsky2007algorithms}. In recent years, there has been increasing interest in the explosive rigidity percolation in various two-dimensional (2D) systems, including origami structures~\cite{li2025explosive,li2025rigidity} and kirigami structures~\cite{choi2023explosive}. Some works have also studied the explosive percolation of physical systems in different dimensions~\cite{li2023explosive,li2024explosive}. In particular, deterministic and stochastic approaches have been developed for controlling the connectivity and rigidity of 3D cubic and prismatic assemblies~\cite{choi2020control}. 

In physical and biological systems, the emergence of global mechanical stability is rarely a purely random process. Instead, it is governed by competitive, history-dependent assembly pathways. For instance, in the self-assembly of mechanical metamaterials or the aggregation of attractive particulate gels, the system can undergo rapid, random aggregation (akin to fast quenching), which typically yields tenuous, ``floppy'' dendritic networks. Conversely, slower, energetically favorable pathways (akin to annealing) allow the system to explore local configurations, avoiding premature, inefficient connections and ultimately achieving global rigidity with fewer redundant bonds. To model this competition between local selection rules and global mechanical stability, we employ the Achlioptas process with a choice parameter, $k$. Physically, $k$ acts as a \textit{kinetic dial} or a proxy for the thermodynamic assembly protocol: $k=1$ corresponds to standard random percolation (fast quenching), while $k > 1$ grants the system a ``choice'' to suppress the growth of the giant component, favoring the merging of smaller, isolated clusters. However, establishing a direct relationship between such a local choice mechanism and the \emph{global} connectivity and rigidity remains highly challenging. In particular, the mechanical utility of a potential bond depends on the complex, non-local structure of the entire existing network, and it is not guaranteed that a simple local rule can lead to monotonic improvements in global stability. Motivated by these challenges, this paper combines theory and computation to study the explosive connectivity and mechanical rigidity in 3D cubic lattice structures. Specifically, by performing comprehensive numerical simulations under an Achlioptas process, we analyze connectivity and rigidity in 3D cubic lattice structures. Using the Nearest-Neighbor (NN) model as a control and the Intra-cube (Intra) model as the target system, we uncover several key phenomena. We then develop a series of theoretical models to explain the phenomena. Our main contributions include:
\begin{itemize}[leftmargin=1.25em]
    \item We numerically characterize the explosive connectivity transition. While the onset of anomalous finite-size scaling (a merger cascade) begins at $k=2$, we demonstrate that for $k \ge 8$, the transition width collapses inversely with system volume, and the susceptibility scaling exactly matches a first-order discontinuity ($\gamma=1.000$).
    \item For the richly-connected Intra model, we observe a significant rigidity-connectivity gap that shrinks dramatically to a global minimum at an intermediate $k$ before widening slightly as $k$ becomes very large.
    \item We develop a theoretical framework to explain these observations. Specifically, we derive the existence of merger-cascade windows that drive the abrupt growth of the giant component in finite systems.
    \item We introduce a novel \textit{conditional progress function} to explain this optimal rigidification efficiency. We show that this local efficiency is a direct consequence of two physically-motivated assumptions, for which we provide strong supporting evidence using tractable random subgraph models.
\end{itemize}

\begin{figure}[t]
    \centering
    \includegraphics[width=\linewidth]{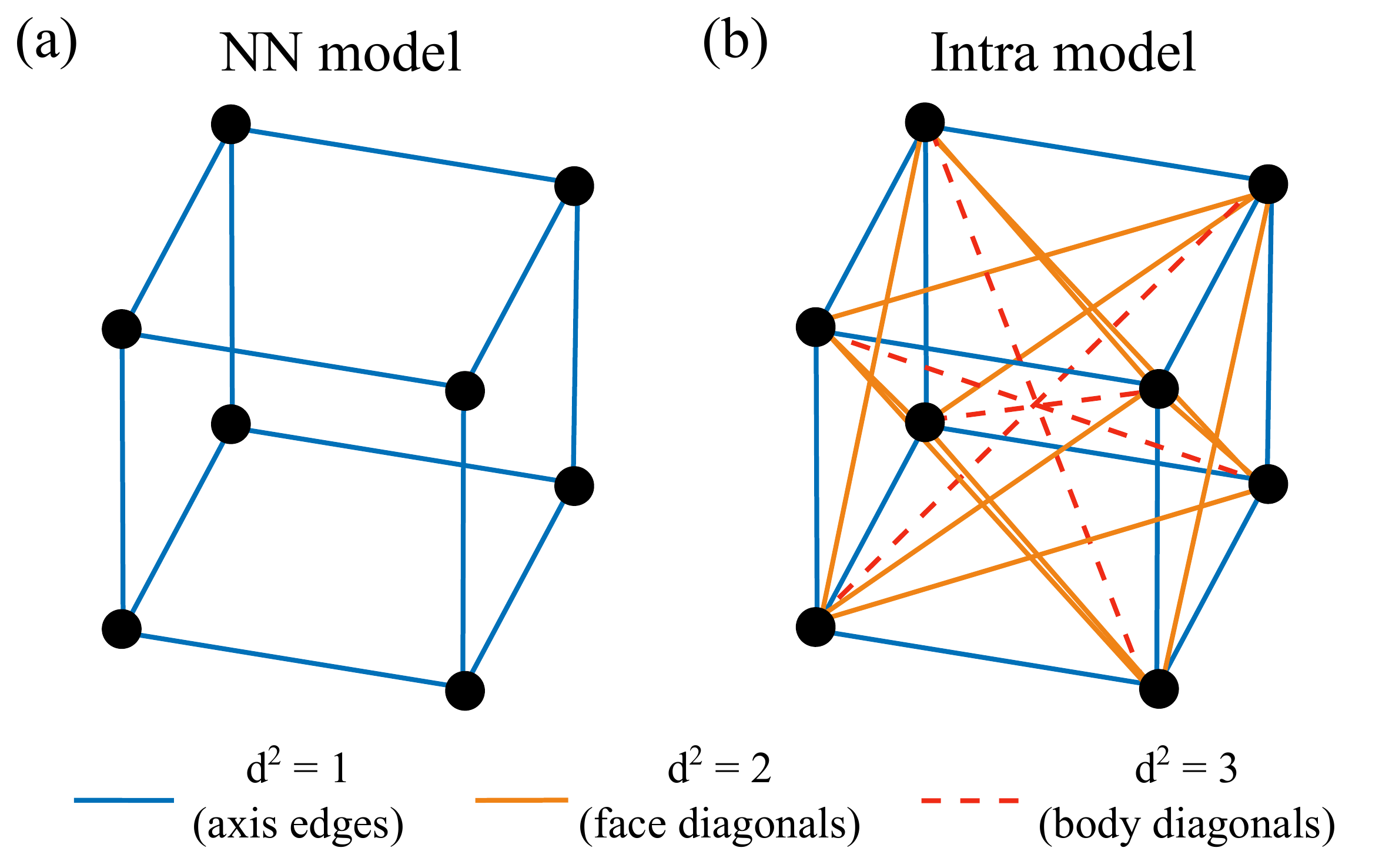}
    \caption{\textbf{The two 3D vertex models considered in this study.} (a) The Nearest-Neighbor (NN) model. (b) The Intra-cube (Intra) model.}
    \label{fig:models}
\end{figure}

\section{Preliminaries and System Setup}

In this work, we study explosive connectivity and mechanical rigidity in 3D cubic lattice structures. We build these structures progressively using a competitive link-selection scheme known as an Achlioptas process. This section introduces our core physical models and the essential theoretical concepts. For a comprehensive list of foundational definitions from graph theory, rigidity theory, and probability theory, the reader is referred to the Supplementary Information, SI Section S1.

\subsection{Host Geometries: The NN and Intra Models}

We focus on two host models of 3D cubic lattice structures, namely the \emph{Nearest-Neighbor, NN} (Shell 1) model and the \emph{Intra-cube, Intra} (S1-S3) model (see \Cref{fig:models} for visualizations of their respective configurations).

\begin{definition}[NN unit cell graph]
The \emph{Nearest-Neighbor (NN) unit cell graph} $G_{\mathrm{NN}}^{\mathrm{unit}}=(V_{\square},E_{\mathrm{NN}})$ has a vertex set $V_\square=\{0,1\}^3$ (the 8 cube corners) and an edge set $E_{\mathrm{NN}}$ consisting of all pairs $\{u,v\}$ with $d^2(u,v)=1$. Here, $d(u,v) = \sqrt{(x_1-x_2)^2+(y_1-y_2)^2+(z_1-z_2)^2}$ is the Euclidean distance between two points $u = (x_1, y_1, z_1)$ and $v = (x_2, y_2, z_2)$.

\end{definition}

As shown in \Cref{fig:models}(a), the NN unit cell graph contains exactly the 12 axis-aligned edges of the cube, and hence $|E_{\mathrm{NN}}|=12$.

\begin{definition}[Intra unit cell graph]
The \emph{Intra-cube (Intra) unit cell graph} $G_{\mathrm{Intra}}^{\mathrm{unit}}=(V_{\square},E_{\mathrm{Intra}})$ has the same vertex set $V_\square$, and the edge set $E_{\mathrm{Intra}}$ consisting of all pairs $\{u,v\}$ with $d^2(u,v)\in\{1,2,3\}$.
\end{definition}

As shown in \Cref{fig:models}(b), for the Intra unit cell graph, $E_{\mathrm{Intra}}$ contains the 12 axis edges ($d^2=1$), the 12 face diagonals ($d^2=2$), and the 4 body diagonals ($d^2=3$). Therefore, $|E_{\mathrm{Intra}}|=12+12+4=28$.

Using the above unit cell graph, we can consider a 3D cubic lattice structure with size $(L+1)\times (L+1) \times (L+1)$ for any $L \geq 1$. By simple counting, we can see that there are in total $L(L+1)(L+1)\cdot 3 = 3L^3+6L^2 + 3L$ edges, $2L^2(L+1)\cdot 3 = 6L^3 + 6L^2$ face diagonals, and $4L^3$ body diagonals. Therefore, denoting the total number of potential edges for the structure as $M$, for the NN model we have 
\[
M = 3L^3+6L^2 + 3L,
\]
while for the Intra model we have
\[
\begin{split}
M &= (3L^3+6L^2 + 3L) + (6L^3 + 6L^2) + (4L^3) \\
&= 13L^3+12L^2 +3L.
\end{split}
\]

\begin{figure}[t]
    \centering
    \includegraphics[width=\linewidth]{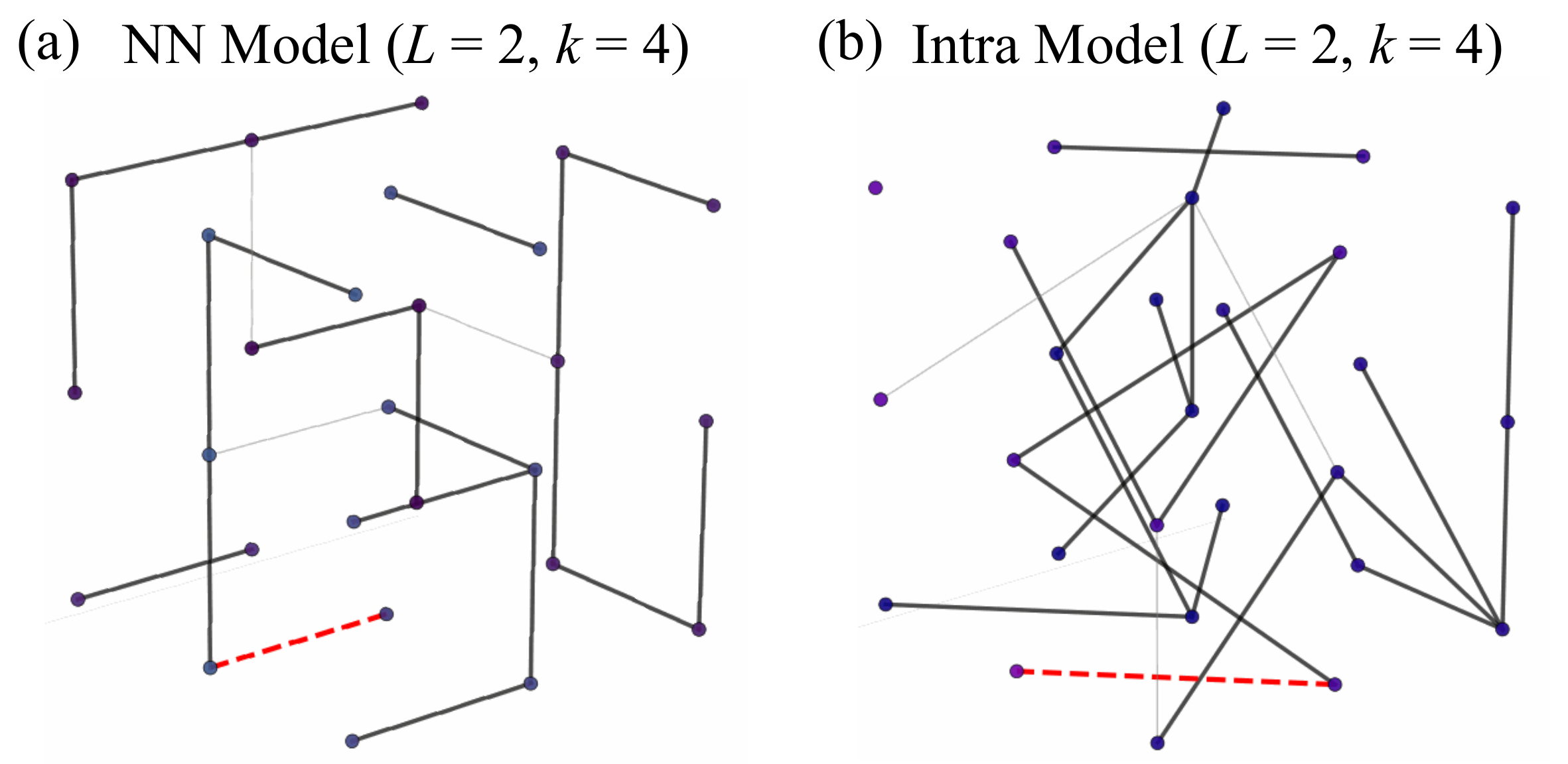}
    \caption{\textbf{An illustration of the \textit{k}-choice Achlioptas process on an intermediate state of a cubic lattice structure with $L=2$ and $k=4$.} (a)~The Nearest-Neighbor (NN) model. (b)~The Intra-cube (Intra) model. The existing edges from the previous states are in black. At the current state, four candidate edges (in grey and red) are sampled, and the one that minimizes the product score is selected (red).}
    \label{fig:process}
\end{figure}

\subsection{The \textit{k}-Choice Achlioptas Process}

Starting with the vertices in a 3D cubic lattice structure without any links, we study the change in connectivity and rigidity of the system if we progressively add the links based on either of the two host models. In particular, we are interested in how the choice of the links and the type of the host model affect the rigidity and connectivity throughout the entire process. In this study, we will consider both adding the links stochastically and via an Achlioptas process (see \Cref{fig:process} for an illustration). In each discrete time step, the parameter $k$ represents the number of randomly sampled options from which one optimal choice is made.

Below, we define the core selection mechanism based on component sizes.

\begin{definition}[Product score]
At any step $t$ of a graph evolution process, let $G^{t-1}$ be the current graph. For any candidate edge $e=\{u,v\}$ not yet in $G^{t-1}$, its \textit{product score} is defined as
\[ s_t(e) := |C_{G^{t-1}}(u)| \cdot |C_{G^{t-1}}(v)|, \]
where $|C_{G^{t-1}}(x)|$ is the size of the connected component containing vertex $x$ in the graph $G^{t-1}$.
\end{definition}

With this, we can now define the Achlioptas process used in this study.

\begin{definition}[Product-rule with $k$ choices]
Let $E$ be the set of all possible edges in a graph. Fix $k\in\{1,2,\dots\}$. At each step $t=1,2,\dots$, we perform the following:
\begin{enumerate}
  \item Sample $k$ distinct candidate edges $S_t=\{e_{t,1},\dots,e_{t,k}\}$ uniformly at random without replacement from the remaining edges $E\setminus E^{t-1}$.
  \item Compute the product score $s_t(e)$ for each $e\in S_t$ using the current graph $G^{t-1}$.
  \item Select one edge $e_t\in S_t$ that minimizes $s_t(e)$ (if several edges tie for the minimum, break ties uniformly at random).
  \item Set $E^t := E^{t-1}\cup\{e_t\}$.
\end{enumerate}
The sequence $(G^t)_{t\ge 0}$ obtained this way is the \emph{$k$-choice product-rule Achlioptas process}.
\end{definition}

\subsection{Fundamentals of Mechanical Rigidity}
To analyze the mechanical rigidity of the cubic lattice structures, we now formalize the key concepts from rigidity theory~\cite{graver1993combinatorial}. These definitions establish the mathematical basis for quantifying when a framework of vertices and edges is stable against infinitesimal deformations.

\begin{definition}[Rigidity matrix]
An \emph{infinitesimal flex} of a framework $(G,P)$, where $G=(V,E)$ is a graph and $P:V\to\RR^d$ assigns to each vertex $v\in V$ a position $p_v\in\RR^d$, is an assignment of velocities $\delta P$ to the vertices such that the length of any edge $\{u,v\}\in E$ does not change to first order. This is expressed by the linear constraints:
\[
 (p_u-p_v)\cdot(\delta p_u-\delta p_v)=0.
\]
These constraints can be written as a linear system $R(P)\,\delta P=0$, where $R(P)$ is the \emph{rigidity matrix} of size $|E|\times 3|V|$. The row corresponding to edge $\{u,v\}$ contains the vectors $(p_u-p_v)$ and $(p_v-p_u)$ in the columns associated with vertices $u$ and $v$, respectively, and zeros elsewhere. The space of all solutions, the \emph{infinitesimal flex space}, is the kernel $\ker R(P)\subseteq \RR^{3|V|}$.
\end{definition}

\begin{definition}[Trivial motions and floppy modes]
In 3D, any framework admits a $6$-dimensional space of \emph{trivial} infinitesimal motions corresponding to three global translations and three global rotations of the entire structure. A motion is \emph{non-trivial} if it deforms the framework. We define the number of independent non-trivial motions, or \emph{floppy modes}, as
\[
 f(G,P) \;=\; \dim\ker R(P) - 6.
\]
A framework $(G,P)$ is \emph{infinitesimally rigid} if it has no floppy modes, i.e., if $f(G,P)=0$. This property is often referred to as local rigidity, and it is distinct from the stronger condition of \emph{global rigidity}, where a generic framework is uniquely determined by its edge lengths up to isometries. The study of global rigidity has deep connections to graph connectivity and unique reconstruction from distance data~\cite{Garamvolgyi2022}.
\end{definition}

\begin{definition}[Generic placements and generic rigidity]
A property holds for \emph{generic placements} if it holds for all vertex positions $P$ outside a specific lower-dimensional algebraic variety (a set of measure zero). A graph $G$ is \emph{generically rigid} if the framework $(G,P)$ is infinitesimally rigid for all generic placements $P$.
\end{definition}

For the NN and Intra host models introduced previously, note that the \emph{Maxwell count} for generic 3D rigidity states that a necessary condition for a graph $G=(V,E)$ to be generically infinitesimally rigid in 3D is
\[
|E|\;\ge\; 3|V|-6.
\]
In particular, applying the Maxwell count to the unit cell graph, we have
\[
3|V|-6 \;=\; 3(8)-6 \;=\; 24-6 \;=\; 18.
\]
It follows that the NN unit cell graph (12 edges) can never be rigid.

While the Maxwell count is necessary, it is not sufficient for rigidity in three dimensions. This stands in contrast to the two-dimensional case, where a complete combinatorial characterization is given by Laman's Theorem~\cite{laman1970graphs}: a graph is minimally rigid in 2D if and only if $|E|=2|V|-3$ and for every subgraph with $|V'|\ge 2$ vertices, $|E'| \le 2|V'|-3$. It provides a purely combinatorial condition and is a cornerstone of the modern matroid-theoretic approach to rigidity~\cite{Cruickshank2026}. The absence of such a complete characterization for $d \ge 3$ implies that combinatorial algorithms (like the (3,6) Pebble Game~\cite{jacobs1997algorithm, lee2008pebble}) provide a \textit{necessary} but not strictly \textit{sufficient} condition for generic rigidity. Recently, there have been significant algorithmic breakthroughs in exact rigidity analysis~\cite{Javerzat2026, Lu2026PRE}. For instance, novel integrations of the Pebble Game with Newman--Ziff approaches and exact cluster-merging theorems have achieved near-linear time complexity for exact 2D rigidity percolation~\cite{Javerzat2026}. Furthermore, dynamic event-based ensemble approaches have recently been formulated to efficiently track rigid cluster evolutions and cascade mergers in 2D systems~\cite{Lu2026PRE}. However, as is well-known in rigidity theory~\cite{chubynsky2007algorithms, cheng2014maxwell}, the exact determination of generic rigid clusters in three (or more) dimensions remains a crucial computational and theoretical bottleneck. Furthermore, on finite 3D cubic lattices with open boundaries, the strict algorithmic enforcement of the $3|V|-6$ condition encounters ``first-edge'' initialization paradoxes and can become trapped by local over-constrained subgraphs or boundary anomalies. 

To circumvent these combinatorial traps while capturing the macroscopic mechanical transition, we measure \textit{relaxed rigidity}. In this relaxed framework, an edge is deemed mechanically independent if it bridges two topologically disconnected components (thereby removing a relative translational degree of freedom), or if it successfully consumes a pebble under the relaxed intra-component pebble conservation rules. This heuristic successfully tracks the bulk rigidification process without being artificially halted by localized boundary anomalies.

\begin{assump}[Generic placements and tie-avoidance]\label{assump:generic_in_sec2}
All statements regarding rigidity in this paper are made for generic placements $P$ of the vertices $V_L$ in $\RR^3$. This measure-one assumption ensures that any lack of rigidity is due to the combinatorial structure of the graph rather than a non-generic, degenerate alignment of vertices.
\end{assump}

\subsection{Quantifying Connectivity and Mechanical Rigidity}

To analyze the change in connectivity, we consider the connected components in the cubic lattice structures. For a system of a given linear size $L$, we denote $N=(L+1)^3$ as the total number of vertices and $M$ as the total number of potential edges for the host graph.

\begin{definition}[Connected components and sizes]
For a graph $H=(V,F)$, its (vertex) connected components are the maximal subsets of vertices within which every pair is connected by a path using edges from $F$. For a vertex $v$, write $C_H(v)$ for the component of $v$ in $H$, and $|C_H(v)|$ for its size (number of vertices).
\end{definition}
In particular, we denote $S_{\max}$ as the number of vertices in the largest connected component.

A key quantity for characterizing the connectivity transition is the susceptibility. It is defined as the mean size of the component to which a randomly chosen vertex belongs. Let the system have components $C_i$ of size (number of vertices) $s_i$. A vertex $\rho$ chosen uniformly at random belongs to a specific component $C_i$ with probability $s_i/N$. The expected component size is therefore equivalent to the second moment of the cluster size distribution:
\[
\EE[|C(\rho)|] = \sum_{\text{all clusters } i} s_i \cdot \frac{s_i}{N} = \frac{1}{N}\sum_{\text{all clusters } i} s_i^2.
\]
Based on this, we define two distinct susceptibility measures for theoretical and numerical purposes:

\begin{definition}[Susceptibility Measures]\label{def:susceptibility_measures}
Let $s_i$ be the size of the $i$-th connected component and $S_{\max}$ be the size of the largest component.
\begin{enumerate}
    \item The \textit{inclusive susceptibility}, denoted $\chi_L$, is given by:
    \[ \chi_L(p) = \frac{1}{N} \sum_{\text{all clusters } i} s_i^2. \]
    This quantity is mathematically non-decreasing with link density $p$ and is used in our theoretical proofs.
    
    \item The \textit{exclusive susceptibility} (of finite clusters), denoted $\chi'_L$, is defined by summing only over clusters that are not the largest one:
    \[ \chi'_L(p) = \frac{1}{N} \sum_{i \neq \max} s_i^2 = \chi_L(p) - \frac{S_{\max}^2}{N}. \]
    This quantity exhibits a sharp peak at the percolation threshold $p_c$ and is therefore the standard tool for numerically locating the transition point~\cite{stauffer2018introduction}.
    
    \item The \textit{peak susceptibility}, denoted $\chi'_{L,\max}$, is the maximum value of the exclusive susceptibility over the entire link density range:
    \[ \chi'_{L,\max} := \max_p \chi'_L(p). \]
    This quantity is utilized in our finite-size scaling (FSS) analysis to locate transition points.
\end{enumerate}
\end{definition}

The detailed analysis of the full order parameter distribution, from which such moments are derived, provides a powerful and refined method for characterizing critical phenomena and improving numerical estimates in finite-size scaling, as pioneered for the Ising model~\cite{binder1981finite}.

To analyze the efficiency of the rigidification process, we must quantify the contribution of each added edge. The rigidity of a framework $(G,P)$ is determined by the rank of its rigidity matrix, $R(P)$. We adopt the following definitions:

\begin{definition}[Rank Gain and Edge Redundancy]
Let $(G,P)$ be a framework. For an edge $e \notin E(G)$, its \textit{rank gain} is defined as
\[ \mathrm{rgain}(e) := \rank(R_P(G+e)) - \rank(R_P(G)). \]
For a generic placement, the rank gain is either 0 or 1. An edge $e$ is \textit{non-redundant} if $\mathrm{rgain}(e)=1$ and \textit{redundant} if $\mathrm{rgain}(e)=0$.
\end{definition}

\begin{definition}[Single-Step Progress]
We define the single-step progress towards rigidity at step $t$, $\Delta \Phi_t$, as the rank gain of the edge $e_t$ added at that step: $\Delta \Phi_t = \mathrm{rgain}(e_t)$. The total number of redundant edges added by time $T$ is therefore $N_{\mathrm{red}}(T) = T - \sum_{t=1}^T \Delta \Phi_t$.
\end{definition}

\begin{remark}
In numerical simulations of finite open lattices, the strict exact calculation of generic rank gain is heavily obscured by boundary floppy modes. Therefore, the single-step progress $\Delta \Phi_t$ is operationally measured using the standard \textit{(3,6) pebble game}. We utilize this as a ``relaxed'' heuristic: because the algorithm successfully identifies inter-component degrees of freedom while gracefully handling internal redundancy, it serves as a highly robust, scale-invariant combinatorial proxy for the bulk theoretical rank gain, bypassing the severe boundary anomalies of open finite lattices.
\end{remark}

To formalize the critical transition points, we first define the order parameter for rigidity, analogous to the one for connectivity, and then define the thresholds in the thermodynamic limit.

\begin{definition}[Critical Thresholds]\label{def:critical_thresholds}
Let $G=(V,E)$ be a graph.
\begin{enumerate}[label=(\alph*),leftmargin=1.5em]
    \item The \textit{connectivity order parameter} for a system of size $N$ at density $p$ is the expected fraction of vertices in the largest connected component:
    \[
    P_N(p) := \EE\left[\frac{S_{\max}(G^{\lfloor pM \rfloor})}{N}\right].
    \]
    
    \item The \textit{largest rigid cluster} of $G$, denoted $R_{\max}(G)$, is a maximal generically rigid subgraph of $G$ with the largest number of vertices. We define its size as $S_{\max}^{\mathrm{rigid}}(G) := |V(R_{\max}(G))|$.

    \item The \textit{rigidity order parameter} for a system of size $N$ at density $p$ is the expected fraction of vertices in the largest rigid cluster:
    \[
    P_N^{\mathrm{rigid}}(p) := \EE\left[\frac{S_{\max}^{\mathrm{rigid}}(G^{\lfloor pM \rfloor})}{N}\right].
    \]

    \item The \textit{connectivity threshold}, $p_c^{\mathrm{conn}}$, is the infimum of densities where the connectivity order parameter $P_N(p)$ is non-zero in the thermodynamic limit ($N\to\infty$):
    \[
    p_c^{\mathrm{conn}} := \inf\Big\{ p \in [0,1] \;\Big|\; \lim_{N\to\infty} P_N(p) > 0 \Big\}.
    \]

    \item The \textit{rigidity threshold}, $p_c^{\mathrm{rigidity}}$, is the infimum of densities where the rigidity order parameter is non-zero in the thermodynamic limit:
    \[
    p_c^{\mathrm{rigidity}} := \inf\Big\{ p \in [0,1] \;\Big|\; \lim_{N\to\infty} P_N^{\mathrm{rigid}}(p) > 0 \Big\}.
    \]

\end{enumerate}
In numerical simulations on finite systems, these thresholds are estimated by the location of the peaks of their corresponding susceptibilities.
\end{definition}

It is well-known that rigidity implies connectivity~\cite{graver1993combinatorial}. Therefore, the rigidity threshold can be no lower than the connectivity threshold. This allows us to define the \textit{rigidity--connectivity gap}:

\begin{definition}[Rigidity--connectivity gap]
The rigidity--connectivity gap $\Delta p_c$ is defined as the difference between the rigidity threshold and the connectivity threshold:
\[ \Delta p_c = p_c^{\mathrm{rigidity}} - p_c^{\mathrm{conn}} \ge 0. \]
\end{definition}
This gap quantifies the additional link density required to rigidify the system after a giant connected component has already formed, and will be a central quantity measured in our numerical results.

\begin{figure*}[t!]
    \centering
    \includegraphics[width=1\textwidth]{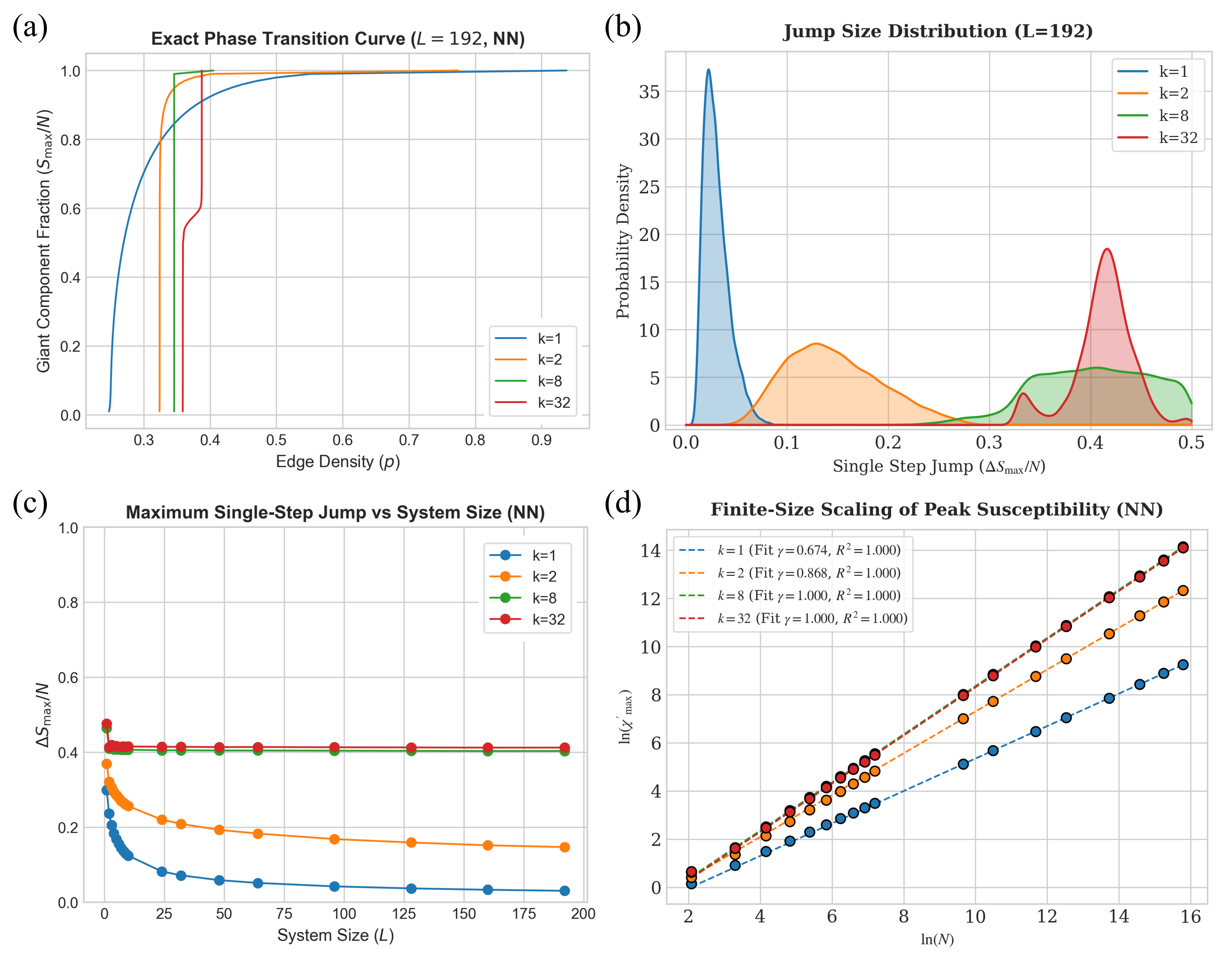}
    \caption{\textbf{Thermodynamic Scaling Analysis of the NN Model.} (a) High-resolution phase curves at $L=192$ reveal a multiple-discontinuity cascade for $k=32$. (b) The probability density of the single-step jump size at $L=192$ demonstrates the transition from a continuous regime ($k=1$) to an intermediate crossover plateau ($k=8$) and finally a highly quantized, multimodal explosive regime ($k=32$). (c) The maximum single-step jump $\Delta S_{\max}/N$ (where $N = (L+1)^3$) stabilizes at a macroscopic fraction ($\sim 0.4$) for large $k$, proving the discontinuity does not decay with system size. (d) Finite-size scaling of the peak susceptibility $\chi_{\max}$ demonstrates that for $k \ge 8$, the critical exponent mathematically converges to $\gamma = 1.000$, the signature of a first-order transition. Here, the dashed best-fit lines are computed strictly using data points for large system sizes ($L>10$) to capture the true asymptotic scaling behavior.}
    \label{fig:nn_thermo}
\end{figure*}

\begin{figure*}[t!]
    \centering
    \includegraphics[width=1\textwidth]{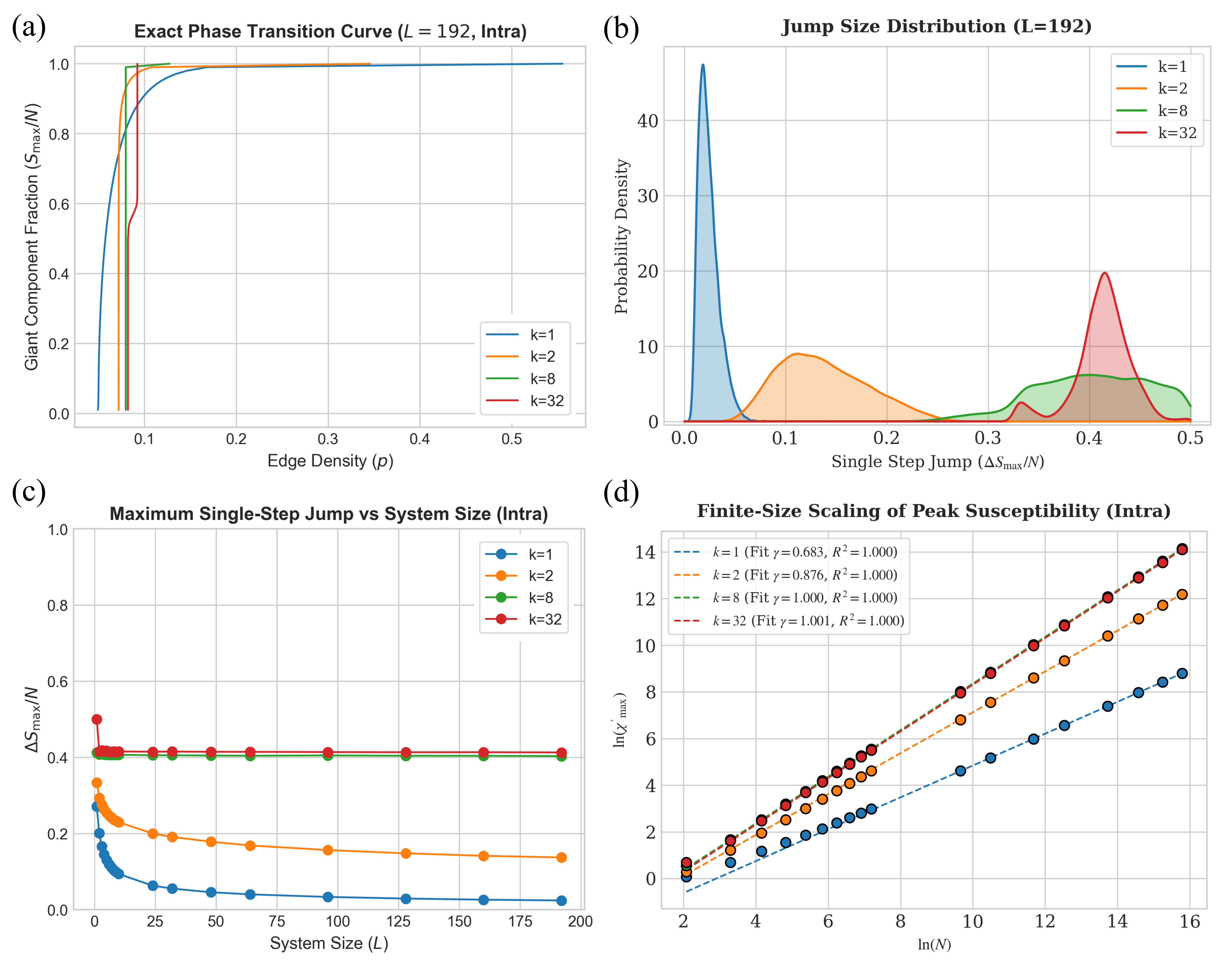}
    \caption{\textbf{Thermodynamic Scaling Analysis of the Intra Model.} Consistent with the NN model, the highly coordinated Intra lattice exhibits an explosive connectivity transition for $k \ge 8$.  (a) The phase transition curves at $L=192$ illustrate the severely delayed, abrupt onset of the giant component. (b) The probability density of the single-step jump size at $L=192$ highlights the continuous regime ($k=1$), the intermediate crossover plateau ($k=8$), and the highly quantized, multimodal discontinuous regime ($k=32$). (c) The maximum single-step jump establishes a persistent macroscopic fraction across all system sizes.  (d) Finite-size scaling of the peak susceptibility confirms the critical exponent mathematically converges to $\gamma = 1.000$. Here, the dashed best-fit lines are computed strictly using data points for large system sizes ($L>10$) to capture the true asymptotic scaling behavior.}
    \label{fig:intra_thermo}
\end{figure*}

\section{Numerical Results and Analysis}\label{sec:num_results}

To investigate the emergence of connectivity and rigidity, we conducted massive-scale numerical simulations on 3D cubic lattices of size $(L+1)\times(L+1)\times(L+1)$ for $L$ up to $192$, analyzing both the NN (Shell 1) and Intra (S1--S3) models with choice parameter $k$ from 1 to 32. Here, to balance physical accuracy with computational feasibility, we adopt a two-tiered numerical strategy. Sweeping 32 distinct values of the choice parameter $k$ while evaluating the 3D relaxed pebble game over $20,000$ independent realizations is exceptionally demanding computationally. Therefore, we first deploy our massive-scale simulation resources ($L$ up to $192$) for representative values of $k$ ($k=1, 2, 8, 32$) to rigorously characterize the phase transition behaviors in the thermodynamic limit. Specifically, for each parameter set $(L,k, \text{host model})$, we performed $20,000$ independent simulations, tracking the exact algorithmic step of the giant component formation to achieve $O(1/M)$ resolution across over 1.2 million thermodynamic trajectories. We then utilize the smaller system sizes ($L = 1 \text{ to } 10$) to perform a high-density parameter sweep across all $k = 1, 2, 3, \dots, 31, 32$, again with $20,000$ independent simulations for each setup. This allows us to systematically locate the non-monotonic behavior of the rigidity gap and identify the optimal choice parameter. Statistical analysis, including bootstrapped t-tests~\cite{efron1992bootstrap} for the rigidity gap, is detailed in the Supplementary Information (SI Section S6--S7 and Tables~S1--S3). The results presented below reveal the key physical behaviors and motivate the theoretical framework developed in the subsequent section.

In our analysis, we track several key quantities to characterize the transitions. The order parameter for connectivity is the relative size of the largest component, $S_{\max}/N$, where $S_{\max}$ is the number of vertices in the largest connected component and $N=(L+1)^3$ is the total number of vertices in the system. We monitor the peak susceptibility $\chi'_{L,\max}$ (as formally defined in Definition~\ref{def:susceptibility_measures}), which serves as the primary quantity for our FSS analysis. To locate the transition point, we measure the exclusive susceptibility $\chi'_L$, whose peaks serve as reliable estimators for the critical thresholds.

\subsection{Explosive Connectivity Transition for \textit{k} \textgreater{} 1} \label{sect:crossover}

We first present numerical results with $L$ up to 192 for the connectivity transition in the 3D cubic lattices. In \Cref{fig:nn_thermo}(a) and \Cref{fig:intra_thermo}(a), we plot the order parameter $S_{\max}/N$ as a function of link density $p$ for $k = 1, 2, 8, 32$ for a large system size of $L = 192$ (with $N = (L+1)^3  \approx 7 \times 10^6$) and compare the transitions. 

For $k=1$, both the NN and Intra models exhibit a gentle, continuous-looking curve characteristic of standard percolation. Note that the standard Achlioptas process with a choice parameter of $k=1$ is equivalent to classical random percolation. For our NN host model, this corresponds to standard bond percolation on the simple cubic lattice. The critical threshold for this transition is a well-established numerical value, $p_c^{\text{cubic}} \approx 0.24881$~\cite{xun2020bond,newman2001fast}. Our numerical results for the $k=1$ case are in excellent agreement with this benchmark, validating our simulation framework. It is worth noting that this lattice threshold is higher than the mean-field prediction of $p_c(d) = 1/(d-1)$ for random $d$-regular graphs, which for $d=6$ yields $p_c=0.2$ \cite{Goerdt2001}, illustrating the role of the fixed lattice geometry in delaying percolation. This classical, continuous transition serves as a baseline against which we analyze the behavior for $k>1$.

As $k$ increases to 2 and 8, the transition becomes dramatically sharper. This occurs because the product rule, which minimizes the product of merging component sizes, becomes more effective with a larger pool of choices. It preferentially selects edges that connect small, isolated components, thereby suppressing the growth of a single dominant cluster and delaying the phase transition to a higher density. This suppression creates a ``powder keg'' of medium-sized components that eventually merge rapidly, leading to an explosive growth of the order parameter. While it has been established that this transition is formally continuous in the thermodynamic limit~\cite{riordan2011explosive, li2024explosive, reis2012nonlocal, choi2014dimensional}, our theoretical framework (detailed in SI Sections S3.B and S3.C) demonstrates that for finite systems, a merger-cascade window exists where the order parameter exhibits a macroscopic jump. This phenomenon, often termed explosive percolation, manifests as an effective discontinuity for the system sizes relevant to mechanical metamaterials. The emergence of such abrupt jumps is a key feature of significant current interest in percolation theory; similar behaviors have been observed in standard bond percolation on dense random graphs with prescribed degree sequences~\cite{Lichev2024}. 

\begin{figure*}[t!]
    \centering
    \includegraphics[width=\textwidth]{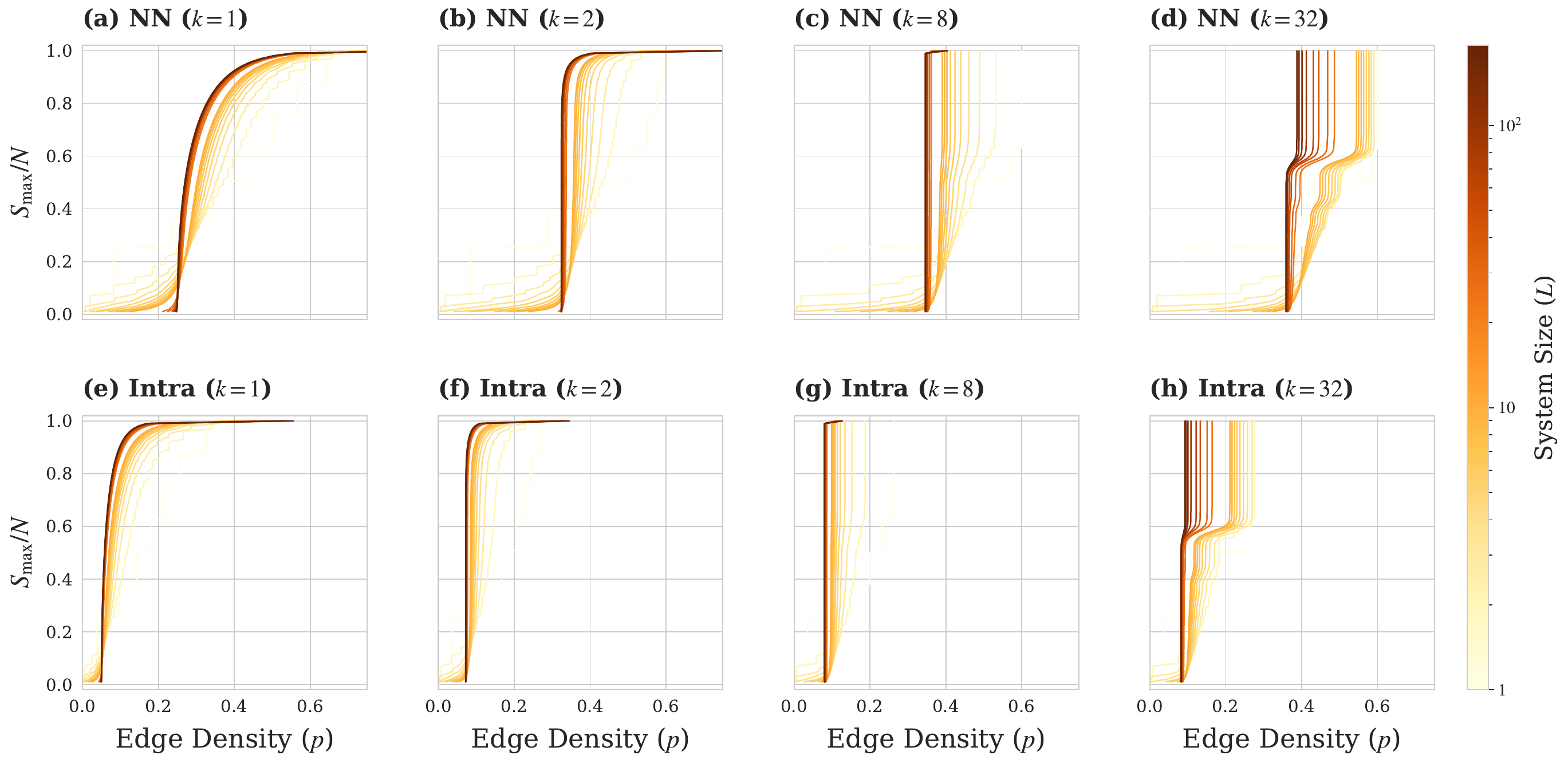}
\caption{\textbf{Explosive connectivity transition phase curves for the NN and Intra models for system sizes up to $L = 192$.} Here, we plot the normalized size $S_{\max}/N$ as a function of the link density $p$ for the NN model ((a)--(d)) and the Intra model ((e)--(h)), evaluated across all system sizes up to $L=192$. Each curve represents the average transition curve over the 20,000 simulations for a specific $L$. }
\label{fig:phasecurves}
\end{figure*}

Interestingly, while the product rule enhances the efficiency of bond placement, one can see that it does not monotonically sharpen the transition in finite systems. As seen in \Cref{fig:phasecurves}, the transition for $k=32$ is visually less abrupt than for $k=8$ for all system sizes from $L = 1$ to $L = 192$. It contrasts sharply with processes that use globally-informed ``oracle'' rules, such as the ``most efficient'' rule in the studies of origami and kirigami percolation~\cite{choi2023explosive, li2025explosive}, where a larger choice set always provides more options to a globally optimal decision process, leading to a monotonically sharpening transition. To explain this phenomenon, note that the $k$-choice product-rule Achlioptas process can be viewed as a deterministic process if $k$ is sufficiently large. Now, if we consider the bond placement problem deterministically, then the optimal strategy under the product rule is to start by picking two isolated vertices (with product score $s = 1\cdot 1 = 1$) at every step, until all $N$ vertices in the system form $N/2$ pairs of vertices (connected components of size $2$). After that, the optimal strategy will be to choose edges that connect two such connected components of size $2$ as much as possible (with product score $s = 2\cdot 2= 4$), which can be repeated for $N/4$ steps. One can then continue the process until unavoidably getting one large cluster with size $N$. Therefore, performing the product rule deterministically (i.e., assuming that the maximum $k$) should yield a jump in $S_{\max}/N$ (from $2/N$ to $4/N$) at the link density $p = \frac{1}{2}\cdot \frac{N}{M}$, followed by another jump (from $4/N$ to $8/N$) at the link density $p = \left(\frac{1}{2} + \frac{1}{4}\right)\cdot \frac{N}{M}$ and so on. The structure will then become one large cluster at 
\begin{equation}\label{eq:deterministic_p}
p = \left(\frac{1}{2} + \frac{1}{4}+\frac{1}{8}+\cdots\right)\cdot \frac{N}{M} =  \frac{N}{M}.
\end{equation}
This qualitatively matches the stepwise increase in $S_{\max}/N$ observed in the $k=32$ plots for both NN and Intra models. However, note that in practice we will need $k \gg32$ to match the above theoretical result quantitatively, as the total number of potential edges grows rapidly with $L$ (e.g., even for only $L = 3$, we already have $M = 144$ for the NN model and $M = 468$ for the Intra model). Also, from the above argument, we can see that the theoretical transition width from $S_{\max} = 2$ to $S_{\max} = N$ will be 
\begin{equation}\label{eq:transition_width}
\frac{N}{2M} = \left\{\begin{array}{cl}
    \frac{(L+1)^3}{2(3L^3+6L^2+3L)} &  \text{ for the NN model},\\
    \frac{(L+1)^3}{2(13L^3+12L^2+3L)} &  \text{ for the Intra model},
\end{array}\right.
\end{equation}
and hence the transition for large $k$ in practice may not be as sharp as the case of $k = 2$. A more detailed derivation of the required scale for $k$ is included in Section \ref{Sec:sel-opt}. See also SI Videos 3--4, in which we focus on smaller system sizes $L$ and perform more high-density sweeps covering all $k = 1,2,3\dots,31,32$ to demonstrate the change in the transition behavior.

Simultaneously, the order parameter jump distribution ($\Delta S_{\max}/N$) at the critical density provides a profound statistical signature of this transition. As observed in the jump size distribution for $L = 192$ in \Cref{fig:nn_thermo}(b) and \Cref{fig:intra_thermo}(b), the dynamics exhibit three distinct physical regimes. For $k=1$, the jump size is predictably microscopic, forming a sharp unimodal peak near zero. For a larger number of choices (e.g., $k=32$), the process avoids the giant component so aggressively that it induces a multiple-discontinuity cascade. This forces the jump distribution to exhibit quantized, multimodal peaks as a discrete number of macroscopic ``mega-clusters'' are finally forced to merge. Crucially, the $k=8$ case serves as the transition intermediate: the algorithm is strong enough to suppress standard continuous percolation but lacks the absolute control required to enforce macroscopic symmetry. Consequently, the jump distribution is pushed down and smeared into a broad crossover plateau, perfectly capturing the onset of the explosive regime. In \Cref{fig:nn_thermo}(c) and \Cref{fig:intra_thermo}(c), we further plot the maximum single-step jump for different $L$ from $1$ to $192$, from which we observe a persistent macroscopic fraction across different system sizes.

Further validation comes from the finite-size scaling of the peak susceptibility. While high-density parameter sweeps for small $L \le 10$ (detailed in SI Section S7) exhibit transient finite-size overshoots due to the onset of the merger cascade, our massive-scale simulation for large $L$ (with $L$ up to $192$) definitively demonstrates that for $k \ge 8$, the scaling stabilizes and the critical exponent mathematically converges to $\gamma = 1.000$, confirming the strict theoretical maximum for a first-order jump (see \Cref{fig:nn_thermo}(d), \Cref{fig:intra_thermo}(d), and SI Videos 1--2).

It is natural to compare our 3D lattice findings with the recent study of explosive percolation in finite dimensions by Li, Wang, and Deng~\cite{li2024explosive}. Investigating hypercubic lattices across dimensions $d = 2$ to $6$ using an elegant event-based ensemble, Ref.~\cite{li2024explosive} demonstrated that explosive percolation governed by the product rule remains formally continuous in the infinite thermodynamic limit, governed by standard finite-size scaling theory once dynamic fluctuations of the pseudocritical point are accounted for. Our results on the 3D simple cubic lattice (NN model) are fully consistent with this underlying thermodynamic continuity. At the same time, our work uncovers several distinct physical phenomena that extend beyond the scope of Ref.~\cite{li2024explosive}. First, while Ref.~\cite{li2024explosive} primarily characterized small choice parameters (principally $k = 2$), our broad exploration up to $k = 32$ reveals that the sharpening of the transition is non-monotonic with choice number: intermediate choices ($k \approx 8$) produce a maximally abrupt crossover in finite systems, whereas very large choices ($k = 32$) approach a deterministic limit that spreads the merger cascade into discrete, quantized doubling steps. Second, beyond the nearest-neighbor cubic lattice, we examine the richly-connected Intra-cube geometry. Most importantly, while Ref.~\cite{li2024explosive} focused exclusively on pure connectivity percolation, our work also studies the interplay between explosive connectivity and 3D mechanical rigidity, demonstrating how competitive link selection governs the closing of the rigidity--connectivity gap $\Delta p_c$. More details will be provided in the following sections.

\begin{figure*}[t!]
    \centering
    \includegraphics[width=\textwidth]{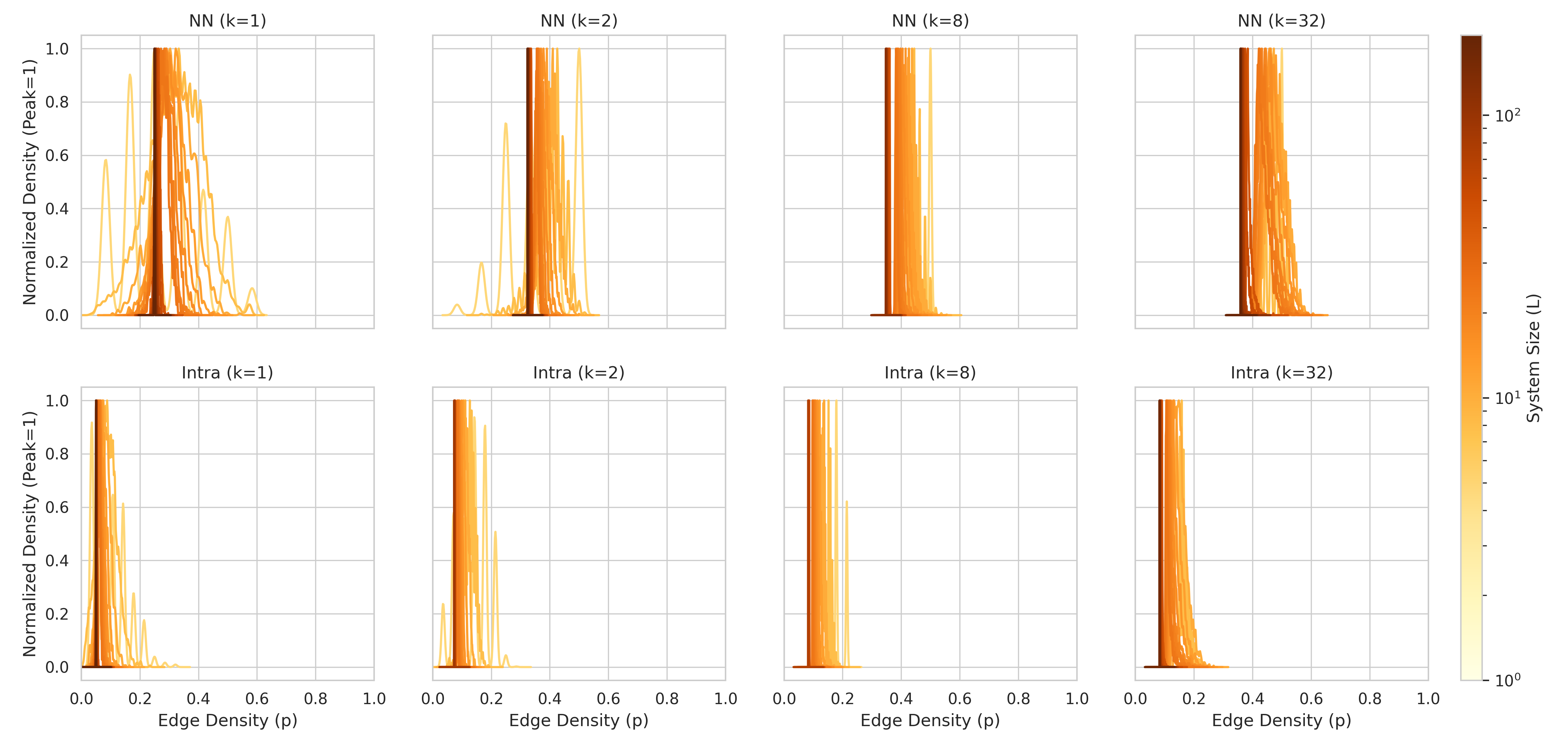}
\caption{\textbf{Normalized Exclusive Susceptibility Spectrum.} The NN model (top row) and Intra model (bottom row) evaluated across all system sizes up to $L=192$. To visualize the evolution across vastly different scales, the amplitude of each susceptibility curve, $\chi'_L(p)$, is normalized to a peak height of 1.0. As $k$ increases, the peak systematically shifts to higher densities and sharpens, graphically illustrating the crossover to a first-order discontinuity. The color gradient maps to the system size $L$.}
\label{fig:susceptibility_comp}
\end{figure*}

\subsection{Monotonic Delay of Connectivity with Number of Choices}

In \Cref{fig:susceptibility_comp}, we show the susceptibility plots for the NN and Intra models with different $k$. As $k$ increases, the location of the susceptibility peak, which serves as a reliable estimator for $p_c^{\mathrm{conn}}$, systematically shifts to higher densities for both the NN and Intra models across all investigated scales, from local system sizes up to our largest macroscopic limit of $L = 192$. Crucially, this delaying effect is not a transient finite-size artifact; the susceptibility peaks shift consistently to higher densities for every single system size $L$ as $k$ is increased, confirming that the delay persists in the thermodynamic limit. In our high-density parameter sweep for all $k = 1, 2, \dots, 31, 32$ for smaller system sizes $L = 1, 2, \dots, 10$ (SI Section~S7 and SI Videos 5--8), we further confirm that this observation is statistically robust: for the Intra model, a Spearman's rank correlation test between $k$ and the measured $p_c^{\mathrm{conn}}$ yields a correlation coefficient of $\rho = 1.0$ ($p \ll 0.01$), signifying a perfect positive monotonic relationship.

This numerical observation provides powerful quantitative evidence for our further theoretical analysis. Our theoretical result in the Supplementary Information (Theorem S3.40) proves that the connectivity threshold, $p_c^{\mathrm{conn}}$, is a monotonically non-decreasing function of the number of choices $k$. In other words, greater choice consistently and predictably delays the onset of global connectivity.

\begin{figure*}[t!]
    \centering
    \includegraphics[width=\textwidth]{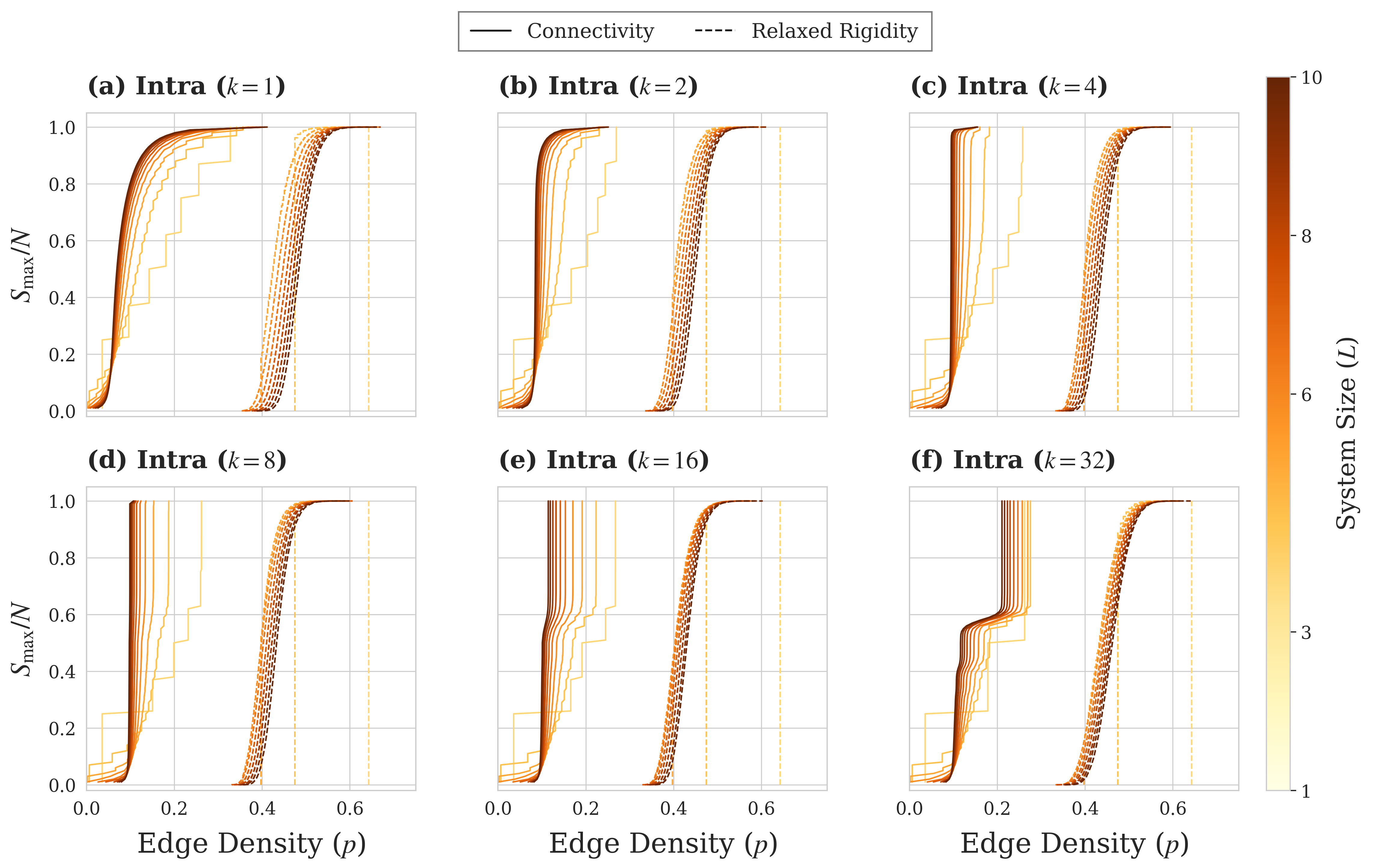}
    \caption{\textbf{Rigidity-connectivity gap up to macroscopic limits.} We compare the connectivity transitions (solid lines) and relaxed rigidity transitions (dashed lines) via the normalized size $S_{\max}/N$ as a function of the link density $p$ for the Intra model for $L = 1, 2, \dots, 10$. All plotted curves represent exact ensemble averages over 20,000 independent realizations. (a)--(f) correspond to $k = 1, 2, 4, 8, 16, 32$ respectively. A positive gap exists which systematically shrinks as $k$ increases, proving that local choice monotonically enhances global rigidification efficiency. Notably, the relaxed rigidity thresholds exhibit clear inverse finite-size scaling: at small $L$, a high surface-to-volume ratio introduces numerous floppy surface modes requiring high bond density to stabilize, whereas for larger $L$, the bulk rigidifies earlier, shifting the dashed curves leftward. The color gradient maps logarithmically to the system size $L$.}
    \label{fig:rigidity_gap_comp}
\end{figure*}

\subsection{Rigidity--Connectivity Gap and Optimal Efficiency}

After considering the connectivity features of the NN and Intra models, we study the rigidity properties. As the NN lattice (coordination number $z=6$) falls well below the Maxwell isostatic threshold for 3D rigidity, the NN model can never be rigid regardless of system sizes and number of choices. Therefore, here we focus on the rigidification process of the richly-connected Intra model. 

However, here we note that evaluation of 3D central-force rigidity for large system sizes poses a significant challenge computationally. While simple connectivity can be evaluated in almost-linear time, assessing 3D central-force generic rigidity is constrained by the non-linear scaling of Pebble Game algorithms. The 3D relaxed pebble game requires extensive recursive searches to find and rearrange free pebbles. Near the rigidity transition threshold, non-local dependencies and redundant constraints trigger high rates of search failures and backtracking, causing the algorithmic complexity per trajectory to scale non-linearly as $O(N^2)$ to $O(N^3)$, where $N = (L+1)^3$ is the total number of vertices in the system. Because the number of vertices $N$ scales cubically with the linear dimension $L$, the real-world execution time scales as $O(L^6)$ to $O(L^9)$, making the computation for large $L$ infeasible. 

Due to this prohibitive computational complexity of evaluating 3D relaxed rigidity on massive structures, here we focus on the smaller system sizes up to $L = 10$ ($N = (L+1)^3 = 1331$) in the rigidity analysis, with a high-density parameter sweep for all $k \in [1, 32]$ and 20,000 simulations for each setup. To analyze the efficiency of rigidification, we consider the rigidity--connectivity gap $\Delta p_c$. This quantity measures the additional link density required to achieve a globally rigid state after a giant connected component has already formed. A smaller gap implies a more efficient rigidification process. Our numerical simulations reveal a dramatic initial increase in rigidification efficiency (a shrinking gap) that reaches an optimal minimum at intermediate $k$ (see SI Table S3 for the complete $k$-sweep from $k = 1$ to $k = 32$). 

As seen in \Cref{fig:rigidity_gap_comp}, for all values of $k$, a large and statistically significant positive gap ($\Delta p_c > 0$) exists between the connectivity transition (solid lines) and the rigidity transition (dashed lines) (see also SI Video 9). The geometric basis for why this host supports efficient rigidity is explored later in Proposition~S1.26 in SI. More importantly, the dramatic initial shrinking of the gap provides strong motivation for our main theoretical result on local rigidification efficiency (\Cref{thm:monotonic_eff_main}). Furthermore, the gap's eventual plateau and slight widening at very large $k$ ($k > 16$) beautifully align with our deterministic limit analysis (\Cref{Sec:sel-opt}). Here, our massive-scale simulations with $20,000$ independent realizations for each $L$ and each $k$ offer a striking confirmation: As $k$ increases, the rigidity-connectivity gap narrows rapidly to a global minimum at intermediate $k$ before slightly widening at very large $k$. This sweep reveals that excessive choice leads to tree-like, maximally floppy components, thereby yielding an optimal efficiency ``Goldilocks'' zone at intermediate $k$ ($k \approx 16$). This provides compelling numerical evidence that the local product-rule acts as a highly effective proxy for achieving global mechanical stability efficiently.   We remark that the underlying principle is that the Intra model, with its dense connectivity including face and body diagonals, serves as a strong ``rigidity expander''~\cite{lew2025rigidity}, a class of graphs whose mechanical stability is governed by $d$-dimensional algebraic connectivity.

Analogous to the connectivity transition phenomenon discussed earlier in Section~\ref{sect:crossover}, we note that the sharpness of the rigidity transition in \Cref{fig:rigidity_gap_comp} appears maximal at an intermediate $k$ (e.g., $k=8$) instead of a much larger $k$ (e.g., $k=32$). This phenomenon can be explained by a similar deterministic argument. Specifically, as explained previously, the product rule with a sufficiently large $k$ will preferentially form a large connected cluster at $p = N/M$, but the cluster will be floppy. As additional links are added at the subsequent steps, the size of the largest rigid cluster will increase steadily. This suggests that the rigidity transition for a large $k$ may not be as sharp as that for an intermediate $k$. 

It is important to contextualize our findings within rigorous mathematical results. Riordan and Warnke~\cite{riordan2011explosive} proved that for any fixed $k$, the Achlioptas process undergoes a continuous phase transition in the strict thermodynamic limit ($N \to \infty$). However, this continuity can be physically elusive. Our extended massive-scale simulations reveal a clear physical bifurcation. While the onset of an ``explosive'' merger cascade begins at $k=2$ (driving an initial anomalous finite-size scaling), the transition for low choice parameters ultimately remains continuous. Conversely, for $k \ge 8$, the finite-size corrections become so severe that the susceptibility scales with an exact exponent of $\gamma = 1.000$, which is the theoretical signature of a discontinuous jump. Furthermore, the single-step order parameter jump stabilizes at a macroscopic fraction regardless of system size. Thus, we classify the transition for $k \ge 8$ as ``explosive connectivity'', a mathematically continuous transition with such extreme finite-size corrections that it physically behaves as a first-order discontinuity in any realizable finite material structure.

\section{Theoretical Framework and Analysis} 

Motivated by the observations from our numerical simulations, here we develop a series of theoretical results on explosive connectivity and rigidity.

\subsection{Motivation: The Rigidity-Connectivity Gap in Mean-Field Models}

While our primary focus is the Achlioptas process on a structured 3D lattice, we can gain valuable insights from simpler, analytically tractable mean-field random graph models, whose theory is detailed in modern treatments of the subject~\cite{frieze2015introduction}. These models offer a powerful conceptual baseline for understanding the fundamental relationship between network connectivity, coordination, and mechanical stability. While the theorems for these mean-field models do not directly translate to our lattice system, they provide a strong foundation for the hypotheses we develop and test in this work. Specifically, a detailed analysis of such models (provided in SI Section~S2) reveals a key principle: the gap between the connectivity and rigidity thresholds is a monotonically decreasing function of the average network coordination. This insight motivates our central hypothesis for how the choice parameter $k$ influences rigidification in our lattice-based system.

By analogy with the mean-field results, we hypothesize that this effective coordination number, not $k$ itself, is the fundamental parameter controlling the efficiency of rigidification.

\begin{hyp}[Effective Coordination as the Unifying Principle]
The primary role of the choice parameter $k$ in the rigidity percolation of our 3D lattice models is to control the effective coordination at criticality, $d_{\text{eff}}(k)$. We hypothesize that the rigidity-connectivity gap, $\Delta p_c(k)$, is a monotonically decreasing function of this effective coordination number.
\end{hyp}

To make this connection explicit, we define the effective coordination at criticality as the average degree of the graph at the connectivity threshold: $d_{\text{eff}}(k) = 2|E(G^{\lfloor p_c^{\text{conn}}(k)M \rfloor})|/N = 2 M p_c^{\text{conn}}(k) / N$. Since for our lattice models the total number of potential edges $M$ is proportional to the number of vertices $N$, $d_{\text{eff}}(k)$ is directly proportional to the critical density $p_c^{\text{conn}}(k)$. Our theoretical result in the Supplementary Information (Theorem S3.40) proves that $p_c^{\mathrm{conn}}(k)$ is non-decreasing with $k$, establishing that $d_{\mathrm{eff}}(k)$ is also non-decreasing.

This framework provides a powerful conceptual bridge: it suggests that even in a complex, history-dependent process on a fixed lattice, the fundamental mean-field principle—that higher coordination yields a smaller rigidity gap—still holds, with the Achlioptas rule acting as a kinetic dial to tune this coordination. This unifying principle is strongly validated by our numerical simulations for the Intra model (see \Cref{fig:rigidity_gap_comp} and SI Table S3), where the gap shrinks dramatically from $0.4281$ at $k=1$ down to $0.3326$ at $k=16$ as the effective coordination is systematically increased.

\begin{definition}[First-order scaling shift]
The standard FSS relation for the shift of the pseudo-critical point in a $d$-dimensional system is~\cite{binder1984finite}:
\[ p_c(L) - p_c(\infty) = A L^{-d} + \smallo(L^{-d}), \]
where $A$ is a non-universal constant.
\end{definition}

For the 3D cubic lattice structures studied here ($d=3$), this predicts a shift proportional to $L^{-3}$. While we focus on the susceptibility scaling in our numerical analysis, this scaling shift provides an alternative, powerful method for identifying and characterizing first-order transitions in future studies.

\subsection{Selection Optimality, Large-\textit{k} Effects, and Bounds}\label{Sec:sel-opt}

A key observation from our numerical results in \Cref{sect:crossover} is that the sharpening of the phase transition is non-monotonic with the choice parameter $k$; for instance, the transition appears more abrupt for $k=8$ than for $k=32$. This section develops a theoretical framework to explain this phenomenon by analyzing the conditions under which the stochastic Achlioptas process approaches a deterministic, globally-optimal limit.

A deterministic process would, at each step, select an edge with the globally optimal (minimal) product score from the set of all available edges. Our stochastic process deviates from this ideal if none of the $k$ sampled edges happen to be among this optimal set. We can formalize this ``failure probability'' and analyze how it depends on $k$.

\begin{proposition}[Failure probability formula]\label{prop:fail-formula}
Let $M$ be the number of edges available for selection at a given step, $m$ be the number of globally optimal edges among them, and $k$ be the number of edges sampled uniformly at random without replacement. The probability that \emph{none} of the $k$ sampled edges is globally optimal equals
\begin{equation}\label{eq:fail}
P(\text{fail} \mid k) \;=\; \frac{\binom{M-m}{k}}{\binom{M}{k}},
\end{equation}
for $0\le k\le M$ (and interpreted as $0$ if $k > M-m$).
\end{proposition}

\begin{proof} This follows from a standard combinatorial argument on hypergeometric sampling. See Proposition~S3.5 in SI. \end{proof}

Intuitively, increasing $k$ should decrease this failure probability; Lemma~S3.6 and Theorem~S3.7 in the Supplementary Information formalize this and prove that the improvement is strictly monotonic.

\vspace{-5mm}

\subsection{Theoretical Model for Local Rigidification Efficiency}
\label{sec:rigidity_formal_framework}

\vspace{-3mm}

We further develop a theoretical model to explain the central numerical observation that increasing the choice parameter $k$ locally enhances the efficiency of mechanical rigidification, driving the initial shrinkage of the rigidity gap. Our model replaces heuristic arguments with a structured framework conditional on two physically-motivated assumptions. The analytical support for these assumptions is provided in SI Section~S4 and S5. Below, we first establish that the mechanical utility of an edge is a non-increasing function of its local product-rule score. We then use this result within the synchronous coupling framework (Theorem S3.40 in SI) to prove the main theorem on monotonic efficiency.

Let $(\mathcal{F}_t)_{t\ge 0}$ be the natural filtration generated by the $k$-choice process. Let $e=\{u,v\}$ be a candidate edge not in $G_{t-1}$. Its product score is $s(e) = |C_{t-1}(u)| \cdot |C_{t-1}(v)|$, and its rank gain is $\rgain(e) \in \{0,1\}$. We define the \emph{conditional progress function} as follows:

\begin{definition}[Conditional Progress Function]
For any score value $s > 0$, we define the \emph{conditional progress function} $P(s)$ as the expected rank gain of a uniformly chosen available candidate edge with score $s$:
\[
P(s) := \mathbb{E}[\rgain(e) \mid s(e) = s].
\]
\end{definition}

With the above definition, we can establish a monotonic relationship between score and redundancy:

\begin{lemma}[Monotonicity of the Conditional Progress Function]\label{lem:progress_monotonicity_main}
The conditional progress function $P(s)$ is a non-increasing function of the score $s$.
\end{lemma}

\begin{proof}
See SI Section~S5. 
\end{proof}

With this lemma rigorously established, we can now state and prove the following theorem.

\begin{theorem}[Monotonic Efficiency with Choice]\label{thm:monotonic_eff_main}
For the $k$-choice product-rule process, let $\mathbb{E}[\Delta \Phi_t^{(k)}]$ be the expected single-step progress towards rigidity at step $t$. For any $1 \le k_1 < k_2$, we have:
\[ \mathbb{E}[\Delta \Phi_t^{(k_2)}] \ge \mathbb{E}[\Delta \Phi_t^{(k_1)}]. \]
Consequently, the expected number of redundant edges added by any time $T$, $\mathbb{E}[N_{\mathrm{red}}(T; k)]$, is a non-increasing function of $k$.
\end{theorem}
\begin{proof}
Conditional on the validity of the Monotonic Density assumption (Assumption S5.3 in the Supplementary Information), the proof follows by combining the synchronous coupling from Theorem S3.40 with the resulting non-increasing property of the conditional progress function $P(s)$.

For the score distribution, the synchronous coupling establishes that the minimal score selected by the $k_2$-process, $s_{k_2}^*$, is stochastically smaller than that selected by the $k_1$-process, $s_{k_1}^*$.
    
The expected single-step progress is $\mathbb{E}[\Delta \Phi_t^{(k)}] = \mathbb{E}_{s_k^*}[P(s_k^*)]$, where the expectation is over the distribution of the selected minimal score. From \Cref{lem:progress_monotonicity_main}, the function $P(s)$ is non-increasing. A standard result of probability theory is that for any non-increasing function $f$ and random variables $X, Y$ where $X$ is stochastically smaller than $Y$, it holds that $\mathbb{E}[f(X)] \ge \mathbb{E}[f(Y)]$.
    
Applying this principle with $f=P$, $X=s_{k_2}^*$, and $Y=s_{k_1}^*$, we immediately obtain:
\[ \mathbb{E}_{s_{k_2}^*}[P(s_{k_2}^*)] \ge \mathbb{E}_{s_{k_1}^*}[P(s_{k_1}^*)]. \]
This inequality holds at every step $t$, which proves $\mathbb{E}[\Delta \Phi_t^{(k_2)}] \ge \mathbb{E}[\Delta \Phi_t^{(k_1)}]$.

Summing the expected progress over time shows that the total expected rank gain for $k_2$ is at least that for $k_1$. Since the number of redundant edges is $N_{\mathrm{red}}(T) = T - \sum_{t=1}^T \rgain(e_t)$, taking the expectation shows that $\mathbb{E}[N_{\mathrm{red}}(T;k)]$ is a non-increasing function of $k$.
\end{proof}

This result formally explains the dramatic initial shrinking of the rigidity-connectivity gap, $\Delta p_c$, with increasing $k$. As previously noted, for extremely large $k$, the topological shift toward spanning trees counterbalances this local efficiency, resulting in the global minimum observed at $k \approx 16$ as shown in \Cref{fig:rigidity_gap_comp}.

\section{Conclusion and Discussion}

In this work, we have performed extensive numerical simulations and developed a rigorous framework to analyze explosive connectivity and rigidity in 3D cubic lattice structures. Our theoretical results establish a formal basis for the explosive signature of the connectivity transition. This is corroborated by massive-scale numerical simulations up to $L=192$, which reveal that the peak susceptibility scales with an exact first-order exponent of $\gamma = 1.000$, and the transition features a persistent, macroscopic phase discontinuity characteristic of explosive percolation. For rigidity, our analysis on the richly connected Intra-host shows the product rule optimally enhances the efficiency of rigidification pathways, finding a global minimum gap at intermediate $k$. Also, to explain the numerically observed initial enhancement of rigidification efficiency, we have introduced a novel conditional progress function. Our model formalizes the link between local selection rules and global mechanical stability. It shows that local efficiency follows from two physically-motivated assumptions, for which we provide strong supporting evidence via tractable proxy models.

\subsection{Broader Impact and Applications}

The principles established in this paper offer a new lens through which to interpret phenomena in physical systems where connectivity and rigidity emerge. Specifically, while the $k$-choice product rule explicitly computes component sizes and requires global information, it serves as an effective \textit{generative proxy} for non-local physical fields and history-dependent protocols in various scenarios discussed below.

\subsubsection{Jamming in Attractive Particulate Systems} 
Our model offers a mechanism to explain the distinct percolation behavior observed in jamming transitions. As shown by Lois et al.~\cite{lois2008jamming}, attractive particulate systems exhibit two distinct transitions: a connectivity percolation at low density, followed by a rigidity percolation at a higher density. The region between these transitions represents a ``floppy gel'' state. Crucially, the mechanical state is highly history-dependent; rapid quenching versus slow annealing leads to different structures. 

In this context, the choice parameter $k$ does not imply particles making calculations, but rather proxies the \textbf{thermodynamic protocol}. A standard random process ($k=1$) mimics fast aggregation, resulting in a tenuous, floppy network (a large rigidity-connectivity gap). Increasing $k$ mimics a slower, more energetically favorable pathway (annealing) that suppresses the growth of floppy dendrites. Our observation that increasing $k$ shrinks the rigidity-connectivity gap aligns with the physical reality that optimized assembly protocols can delay the onset of a giant component to achieve a more mechanically stable (rigid) configuration more efficiently.

\subsubsection{Biological Network Formation} 
The self-assembly of endothelial cells into functional vascular networks is a prime example of efficient percolation. While cells do not literally count the size of distant clusters, they interact via \textbf{long-range substrate-mediated elastic interactions}. Noerr et al.~\cite{noerr2023optimal} demonstrated that contractile cells exert traction forces that create deformation fields in the substrate; these fields decay slowly, allowing cells to sense neighbors far beyond contact. 

Crucially, the magnitude of the deformation field scales with the integrated contractility of the cell cluster. Therefore, a larger cluster generates a stronger mechanical signal. Our product rule, which favors merging smaller components to delay the formation of a giant cluster, phenomenologically captures the kinetic bias imposed by these long-range fields. The fields act as a physical proxy for ``component size,'' driving the system toward a distributed, space-spanning network rather than a single dense aggregate. Our proof of monotonic rigidification efficiency (Theorem~\ref{thm:monotonic_eff_main}) provides a formal basis for the observation that such mechanical guidance leads to more ``cost-efficient'' networks~\cite{noerr2023optimal}.

\subsubsection{Design and Fabrication of Mechanical Metamaterials} 
Finally, for mechanical metamaterials, our framework may serve as a \textbf{design heuristic} for top-down fabrication (e.g., 3D printing), rather than a model of self-assembly. In a manufacturing context, global component information is readily available.

Indeed, delaying the connectivity transition ($p_c^{\text{conn}}$) appears to generate a ``less stiff'' structure in the bulk sense. However, our objective is to optimize the \textbf{efficiency of rigidification}. In random assembly ($k=1$), a giant component forms early but remains non-rigid (floppy) for a large range of density. This gap represents wasted material that contributes to mass but not stiffness. By increasing $k$, we suppress this early floppy percolation. When the structure finally connects, it does so at a density much closer to the rigidity threshold (a smaller $\Delta p_c$). Thus, our rule provides a prescription for fabricating disordered, porous architectures that achieve global rigidity with a minimum number of redundant bonds, maximizing the stiffness-to-weight ratio for stochastic lattice materials.\\

%\vspace{-5mm}

\subsection{Future Directions}

This work opens several new avenues for research. The most immediate theoretical challenge is to provide a rigorous proof for the two central assumptions that underpin our rigidity framework: the monotonic average density of components and the emergent rigidity of large Intra-host subgraphs. Progress on this front would likely require new analytical techniques to handle the history-dependent nature of the Achlioptas process. 

Furthermore, our investigation highlights a major algorithmic hurdle in the study of finite 3D metamaterials. While the relaxed 3D pebble game effectively captures the macroscopic rigidification sequence, the strict generic $3N-6$ constraint mathematically fails to manifest on finite cubic lattices with open boundaries due to severe surface floppy modes and localized geometric traps. Formulating an exact, computationally tractable combinatorial algorithm for strict 3D generic rigidity on bounded lattice geometries remains an outstanding open problem for future research.

Another compelling direction arises from our discovery of a non-monotonic finite-size effect in the connectivity transition, where the apparent sharpness is maximized at an intermediate choice parameter (e.g., $k=8$) functioning as a crossover regime. A systematic study of this optimal choice, $k_{\text{opt}}(L)$, and its dependence on system size and host geometry could yield deeper insights into the interplay between local heuristics and global phase transitions.

Finally, the predictive power of this framework invites its application to other classes of disordered materials, from amorphous solids to biopolymer networks, where the mechanisms of local selection and emergent stability remain open questions. The framework developed here offers a complementary perspective to data-driven machine learning approaches increasingly being used to study phase transitions~\cite{mehta2019high}, as well as to methods from topological data analysis, which use tools like persistent homology to characterize the multiscale structure, or ``shape,'' of such complex systems~\cite{ghrist2008barcodes}.\\

\bibliographystyle{ieeetr}
\bibliography{references}

\vspace{1cm}

%%%%%%%%%%%%%%%%%%%%%%%

\centerline{\Large\textbf{Supplementary Information}}
\appendix
\renewcommand\thefigure{S\arabic{figure}}    
\setcounter{figure}{0}
\renewcommand\thetable{S\arabic{table}}    
\setcounter{table}{0}
\renewcommand{\thesection}{S\arabic{section}}

%\tableofcontents

\section{Theoretical Preliminaries}\label{sec:si_definitions}

For completeness and self-containedness, in this supplementary section, we provide the detailed descriptions of the concepts and preliminaries in graph theory and rigidity theory relevant to our work.

\subsection{Host Families, Processes, and Order Parameters}
We consider finite host graphs $G_L=(V_L,E_L)$ with bounded maximum degree $\Delta$ on 3D cubic lattice structures, where $N=|V_L|$ and $M=|E_L|$.

\begin{definition}[Graph]
A (simple) graph $G=(V,E)$ consists of a finite set $V$ of \emph{vertices} and a set $E\subseteq \{\{u,v\}: u,v\in V,\, u\neq v\}$ of \emph{edges}. We say $u$ and $v$ are \emph{adjacent} if $\{u,v\}\in E$.
\end{definition}

\begin{definition}[Host graph]
Let $G=(V,E)$ be a fixed finite simple graph (no loops, no multiple edges), called the \emph{host}. Think of $V$ as the vertex set and $E$ as the set of all edges that are \emph{available} to be added during the process.
\end{definition}

\begin{definition}[Evolving graph]
We construct a sequence of subgraphs $(G^t)_{t\ge 0}$, where $G^0=(V,\emptyset)$ has no edges, and each step $t\ge 1$ adds exactly one new edge from $E$ that was not already present. Thus $G^t=(V,E^t)$ with $E^t=E^{t-1}\cup\{e_t\}$ for some $e_t\in E\setminus E^{t-1}$.
\end{definition}

\begin{definition}[Master candidate sequence]
Fix an ordering of all edges in $E$ (say, label edges $1,2,\dots,|E|$). Consider a random permutation $\pi$ of $E$ chosen uniformly at random among all permutations; think of this as an i.i.d.-like source without replacement. We will read consecutive blocks from $\pi$ to form candidate sets. When some edges have already been added to the evolving graph, we skip them and keep reading forward until we have collected the required number of \emph{unused} candidates.
\end{definition}

To establish our theoretical results related to the number of choices $k$, we simplify the scenario and focus on the case where the master permutation is fixed. Then, for two choice parameters $k= k_1$ and $k=k_2$ (say, with $k_1 < k_2$), the set of candidate edges for $k_1$ at each step can be considered as a subset of that for $k_2$, which makes our analysis easier.

\begin{definition}[$k$-choice product-rule Achlioptas process with fixed master permutation]
Let $(G^{t}_{(k_1)})_{t\ge 0}$ and $(G^{t}_{(k_2)})_{t\ge 0}$ be two evolving graphs on the same vertex set $V$, starting with $G^{0}_{(k_i)}=(V,\emptyset)$ for $i\in\{1,2\}$.

At each step $t\ge 1$, do:
\begin{enumerate}
  \item From the master permutation $\pi$, scan forward and collect the first $k_2$ edges that are not yet present in \emph{either} process at step $t-1$. Call this $k_2$-set $\mathcal{S}_t=\{e_{t,1},\dots,e_{t,k_2}\}$.
  \item Define the $k_1$-set for the smaller-$k$ process as the first $k_1$ edges within $\mathcal{S}_t$, i.e., $S^{(k_1)}_t=\{e_{t,1},\dots,e_{t,k_1}\}\subset \mathcal{S}_t$, and define the $k_2$-set for the larger-$k$ process as $S^{(k_2)}_t=\mathcal{S}_t$.
  \item Compute product scores with respect to the \emph{current} graphs:
        \[
        \begin{split}
        s^{(k_i)}_t(e) \;=\; &\big| C_{G^{t-1}_{(k_i)}}(u)\big| \cdot \big| C_{G^{t-1}_{(k_i)}}(v)\big|\\
        &\quad\text{for } e=\{u,v\}\in S^{(k_i)}_t, \quad i\in\{1,2\}.
        \end{split}
        \]
        Then choose $e^{*}_{t,(k_i)}\in S^{(k_i)}_t$ minimizing $s^{(k_i)}_t(e)$ (break ties uniformly at random \emph{using the same tie-breaking randomness for both processes restricted to their own candidate sets}).
  \item Update $G^{t}_{(k_i)}$ by adding $e^{*}_{t,(k_i)}$ to $G^{t-1}_{(k_i)}$.
\end{enumerate}
\end{definition}

\begin{lemma}[Suppressive $k$-coupling]\label{lem:kcouple}
With the coupled construction above, for every step $t\ge 1$ we have
\[
\min_{e\in S^{(k_2)}_t} s^{(k_2)}_t(e) \;\le\; \min_{e\in S^{(k_1)}_t} s^{(k_1)}_t(e).
\]
In words: the product score of the edge actually chosen by the $k_2$-choice process at step $t$ is at most the product score of the edge chosen by the $k_1$-choice process at step $t$. Consequently, the $k_2$-choice product-rule process is at least as suppressive of large-component merges as the $k_1$-choice process, and hence connectivity (and giant-component growth) is stochastically delayed when $k$ increases.
\end{lemma}

\begin{proof}
We couple the $k_1$- and $k_2$-choice processes ($1\le k_1<k_2$) on the same probability space as specified earlier: at step $t$ we first form the common candidate pool $\mathcal{S}_t$ by scanning forward in the master permutation and collecting the first $k_2$ currently unused edges; then we take $S^{(k_1)}_t$ to be the first $k_1$ edges of $\mathcal{S}_t$ and $S^{(k_2)}_t=\mathcal{S}_t$, so that $S^{(k_1)}_t\subset S^{(k_2)}_t$. Each process then selects (with consistent, shared tie-breaking randomness) a minimizer of the product score computed with respect to its own current graph.

We will show, by induction on $t$, the following refinement property of component partitions:
for all vertices $v,w$ and all $t\ge 0$,
\begin{equation}\label{eq:refine}
v\text{ connected to }w\text{ in }G^{t}_{(k_2)}\ \ \Longrightarrow \ \  v\text{ connected to }w\text{ in }G^{t}_{(k_1)}.
\end{equation}
Equivalently, at every time $t$, the partition of $V$ into connected components under the $k_2$-process is a refinement of the partition under the $k_1$-process. As an immediate corollary, for every vertex $u$,
\begin{equation}\label{eq:size-monotone}
|C_{G^t_{(k_2)}}(u)| \;\le\; |C_{G^t_{(k_1)}}(u)|,
\end{equation}
and thus, for any edge $e=\{u,v\}$,
\begin{equation}\label{eq:product-monotone}
\begin{split}
s^{(k_2)}_{t+1}(e) \;&=\; |C_{G^{t}_{(k_2)}}(u)|\cdot |C_{G^{t}_{(k_2)}}(v)|\\
\;&\le\; |C_{G^{t}_{(k_1)}}(u)|\cdot |C_{G^{t}_{(k_1)}}(v)| \;=\; s^{(k_1)}_{t+1}(e).
\end{split}
\end{equation}

Base case $t=0$ is trivial: Both processes start from the empty graph and hence have identical component partitions (all singletons), so Eq.~\eqref{eq:refine} holds.

Inductive step: Assume Eq.~\eqref{eq:refine} holds at time $t-1$. Consider step $t$. Let $e^{*}_{t,(k_i)}\in S^{(k_i)}_t$ be the edge selected by the $k_i$-process, i.e.,
\[
e^{*}_{t,(k_i)} \in \arg\min_{e\in S^{(k_i)}_t} s^{(k_i)}_t(e), \qquad i\in\{1,2\}.
\]
We claim that after adding these edges, Eq.~\eqref{eq:refine} still holds at time $t$. There are two cases.

Case 1: $e^{*}_{t,(k_2)}$ connects two vertices that are already in the same component of $G^{t-1}_{(k_2)}$. Then the partition under $k_2$ does not coarsen at step $t$. Since adding any edge cannot split components, the refinement relation in Eq.~\eqref{eq:refine} is preserved.

Case 2: $e^{*}_{t,(k_2)}$ connects two distinct components $A$ and $B$ of $G^{t-1}_{(k_2)}$. By the induction hypothesis, every $G^{t-1}_{(k_2)}$-component is contained in a $G^{t-1}_{(k_1)}$-component. Let $\tilde A$ and $\tilde B$ be the (possibly equal) components of $G^{t-1}_{(k_1)}$ containing $A$ and $B$, respectively. If $\tilde A=\tilde B$, then the merge in the $k_2$-process remains within a single $k_1$-component, so the refinement is preserved. If $\tilde A\neq \tilde B$, then consider the candidate set inclusion $S^{(k_1)}_t\subset S^{(k_2)}_t$ and the product scores. Because of Eq.~\eqref{eq:product-monotone} applied at time $t-1$, for every $e\in S^{(k_1)}_t$,
\[
s^{(k_2)}_t(e)\;\le\; s^{(k_1)}_t(e).
\]
Let $\hat e_t\in \arg\min_{e\in S^{(k_1)}_t} s^{(k_2)}_t(e)$ be a minimizer over the smaller set, but scored in the $k_2$-graph. Then
\[
\min_{e\in S^{(k_2)}_t} s^{(k_2)}_t(e) \;\le\; s^{(k_2)}_t(\hat e_t)
\;\le\; s^{(k_1)}_t(\hat e_t)
\;\le\; \min_{e\in S^{(k_1)}_t} s^{(k_1)}_t(e),
\]
where the first inequality uses $S^{(k_1)}_t\subset S^{(k_2)}_t$, the second uses Eq.~\eqref{eq:product-monotone}, and the last uses the definition of the $k_1$-choice. In particular,
\begin{equation}\label{eq:min-score-ineq}
\min_{e\in S^{(k_2)}_t} s^{(k_2)}_t(e) \;\le\; \min_{e\in S^{(k_1)}_t} s^{(k_1)}_t(e).
\end{equation}
Now, any inter-component edge across $A$ and $B$ in $G^{t-1}_{(k_2)}$ induces an inter-component edge across $\tilde A$ and $\tilde B$ in $G^{t-1}_{(k_1)}$ (since $A\subseteq \tilde A$, $B\subseteq \tilde B$ and $\tilde A\neq \tilde B$). Thus, if the $k_2$-process executes a merge at step $t$, then either the $k_1$-process also merges the corresponding two $k_1$-components or it selects an edge with product score at least as large. In both subcases, adding edges cannot cause $k_2$ to identify vertices that $k_1$ does not identify; hence the refinement relation persists at time $t$.

This completes the induction and establishes Eq.~\eqref{eq:refine} and hence Eq.~\eqref{eq:product-monotone} for all relevant times.

Finally, to prove the statement of the lemma at the given step $t$, combine the set inclusion $S^{(k_1)}_t\subset S^{(k_2)}_t$ with Eq.~\eqref{eq:product-monotone} at time $t-1$ exactly as in Eq.~\eqref{eq:min-score-ineq}:
\[
\min_{e\in S^{(k_2)}_t} s^{(k_2)}_t(e)
\;\le\; \min_{e\in S^{(k_1)}_t} s^{(k_2)}_t(e)
\;\le\; \min_{e\in S^{(k_1)}_t} s^{(k_1)}_t(e).
\]
Thus, the minimum product score available (and hence selected) under $k_2$ is at most that under $k_1$ at step $t$. Since lower product scores systematically prefer merges of smaller components and disfavor merges that would markedly increase the largest component, the $k_2$-choice process is at least as suppressive of large-component growth as the $k_1$-choice process. Standard stochastic domination for increasing graph properties then implies that events such as “the largest component has size at least $m$ by time $t$” occur no earlier (and, in distribution, no more often at fixed time) under $k_2$ than under $k_1$. Hence, increasing $k$ stochastically delays connectivity and giant-component emergence.
\end{proof}

\begin{definition}[Time index and edge density]
We now describe the evolving random graph process on a fixed host graph $G_L=(V_L,E_L)$ with $N:=|V_L|$ and $M:=|E_L|$. We start from the empty subgraph and add host edges one at a time.

Time is discrete: $t=0,1,2,\dots,M$. At time $t$, we have added exactly $t$ edges. The edge density is $p:=t/M\in[0,1]$.
\end{definition}

\begin{definition}[Largest component and order parameter]
Let $C_{\max}(t)$ be a largest component of $G^t_L$ (break ties arbitrarily). The order parameter is
\[
P_{N}(p):=\EE\!\left[\frac{|C_{\max}(\lfloor pM\rfloor)|}{N}\right].
\]
\end{definition}

\subsection{Graphs, Frameworks, and Rigidity}

\begin{definition}[Embedding / Framework in $\RR^d$]
Fix a dimension $d\ge 1$. A \emph{bar-joint framework} (or simply \emph{framework}) in $\RR^d$ is a pair $(G,P)$ where $G=(V,E)$ is a graph and $p:V\to\RR^d$ assigns to each vertex $v\in V$ a \emph{position} $p(v)\in\RR^d$. We interpret each edge $\{u,v\}\in E$ as a rigid bar of fixed length between the points $p(u)$ and $p(v)$.
\end{definition}

\begin{definition}[Infinitesimal motion]
Let $(G,P)$ be a framework in $\RR^d$. An \emph{infinitesimal motion} (or \emph{infinitesimal velocity field}) is an assignment $u:V\to\RR^d$ of a velocity vector $u(v)$ to each vertex $v\in V$ such that for every edge $\{i,j\}\in E$ we have
\[
\big(p(i)-p(j)\big)\cdot \big(u(i)-u(j)\big) \;=\; 0.
\]
This is the first-order condition that the squared length $\|p(i)-p(j)\|^2$ does not change at time $0$ if the points move with velocities $u(i)$ and $u(j)$.
\end{definition}

\begin{definition}[Trivial infinitesimal motions]
A \emph{trivial} infinitesimal motion is one induced by an infinitesimal rigid motion of the entire space: a combination of a translation and a rotation. Concretely, there exists a vector $a\in\RR^d$ and a skew-symmetric $d\times d$ matrix $A$ (so $A^\top=-A$) such that
\[
u(v) \;=\; a + A\,p(v) \qquad \text{for all } v\in V.
\]
These velocities come from translating all points by $a$ and rotating them with instantaneous angular velocity encoded by $A$.
\end{definition}

\begin{definition}[Infinitesimal rigidity]
A framework $(G,P)$ in $\RR^d$ is \emph{infinitesimally rigid} if every infinitesimal motion is trivial. Equivalently, the only solutions $u:V\to\RR^d$ to the edge constraints
\[
(p(i)-p(j))\cdot(u(i)-u(j))=0 \quad \forall\,\{i,j\}\in E
\]
are the trivial ones of the form $u(v)=a+Ap(v)$.
\end{definition}

\begin{definition}[Generic framework (informal)]
A framework $(G,P)$ in $\RR^d$ is \emph{generic} if the coordinates of the points $p(v)$ satisfy no special algebraic relations other than those forced by the graph structure. In particular, genericity ensures that the space of infinitesimal motions has the smallest possible dimension given the graph. We will call a graph \emph{generically infinitesimally rigid in $\RR^d$} (or simply \emph{generically rigid}) if, for almost all (generic) placements $p$, the framework $(G,P)$ is infinitesimally rigid.
\end{definition}

\begin{lemma}[Independent component-wise rigid motions]\label{lem:component-motions}

Suppose the graph $G=(V,E)$ is disconnected, i.e., it has at least two nonempty connected components.
Let $(G,P)$ be a framework in $\RR^d$ and suppose $V$ decomposes as a disjoint union $V=V_1\cup V_2$ with no edges between $V_1$ and $V_2$. Then the following holds:
\begin{itemize}
    \item If $u_1:V_1\to\RR^d$ is any trivial infinitesimal motion on the points $\{p(v):v\in V_1\}$, and $u_2:V_2\to\RR^d$ is any trivial infinitesimal motion on the points $\{p(v):v\in V_2\}$, then the combined field $u:V\to\RR^d$ defined by
\[
u(v)=\begin{cases}
u_1(v), & v\in V_1,\\
u_2(v), & v\in V_2,
\end{cases}
\]
is an infinitesimal motion of $(G,P)$.
\item Moreover, if $u_1$ and $u_2$ are not restrictions of the \emph{same} global trivial motion (i.e.\ there do not exist $a\in\RR^d$ and a skew-symmetric $A$ with $u_i(v)=a+Ap(v)$ for all $v\in V_i$, simultaneously for $i=1,2$), then $u$ is a \emph{nontrivial} infinitesimal motion of $(G,P)$.
\end{itemize}
\end{lemma}
\begin{proof}[Proof]
    See e.g.,~\cite{graver1993combinatorial}.
\end{proof}

\begin{lemma}[Rigidity implies connectivity]\label{lem:rig-implies-conn-restated}
Let $G=(V,E)$ be a graph and $d\ge 2$. If $G$ is disconnected, then for every placement $p:V\to\RR^d$ the framework $(G,P)$ is \emph{not} infinitesimally rigid. Consequently, if there exists a placement $p$ for which $(G,P)$ is infinitesimally rigid (in particular, if $G$ is \emph{generically} infinitesimally rigid), then $G$ must be connected.
\end{lemma}
\begin{proof}[Proof]
    See e.g.,~\cite{graver1993combinatorial}.
\end{proof}

The central question of rigidity theory is whether a given framework is ``floppy'' or ``rigid''. A rigid structure is one that does not deform under pressure, while a non-rigid structure has internal wobbly motions, or \emph{infinitesimal flexes}, that allow it to change shape without instantly altering the length of any of its constituent bars. A framework is formally defined as infinitesimally rigid if the only such motions it can undergo are global, trivial motions of the entire structure in space, namely translations and rotations.

The infinitesimal rigidity of a framework $(G,P)$ is governed by the rank of its rigidity matrix, $\mathcal{R}(G,P)$. This is an $|E| \times d|V|$ matrix where each row corresponds to an edge constraint. An edge $e$ added to $G$ is \textit{non-redundant} if it imposes a new, independent constraint, which occurs if and only if $\rank(\mathcal{R}(G \cup \{e\}, p)) = \rank(\mathcal{R}(G,P)) + 1$. Otherwise, the edge is \textit{redundant}.

To analyze the evolution of rank in our random process, we require two foundational results from random matrix theory and linear algebra. The first allows us to control the spectral properties of a random matrix by relating it to its simpler, deterministic average.

\begin{theorem}[Matrix Concentration, informal~\cite{tropp2015introduction}] \label{thm:concentration}
A sum of independent random matrices is, with very high probability, spectrally close to its expectation. This principle extends under certain conditions to sums with limited dependence, such as those arising in our graph process.
\end{theorem}

The second result relates the spectrum of a matrix to the spectrum of a small perturbation of it.

\begin{theorem}[Weyl's Inequality for Singular Values] \label{thm:weyl}
Let $A$ and $B$ be two $m \times n$ matrices. Let $\sigma_i(M)$ denote the $i$-th largest singular value of a matrix $M$. Then for all $i$: 
\[ |\sigma_i(A+B) - \sigma_i(A)| \le \|B\|, \]
where $\|B\|$ is the spectral norm.
\end{theorem}

These tools allow us to argue that for a large, dense-enough random component, its rigidity matrix will be full-rank with high probability, making additional internal edges redundant.

\subsection{Rank Gain, Redundancy, and Rigidification Cost}

\begin{definition}[Rigidification cost]
Let $(G,P)$ be given. The \emph{rigidification cost} $\Delta t(G,P)$ is the minimal number of edges one must add to $G$ to obtain a graph $H\supseteq G$ such that $(H,P)$ has $f(H,P)=0$ (i.e., is infinitesimally rigid). If no such $H$ exists, set $\Delta t(G,P)=+\infty$.
\end{definition}

\begin{lemma}[Each independent flex needs at least one non-redundant edge]\label{lem:flexNeedEdge}
For any $(G,P)$ with $f(G,P)=F\ge 0$, any sequence of added edges can reduce the number of floppy modes by at most the total rank gain. In particular, to achieve rigidity, the total number of non-redundant edges added must be at least $F$ [see e.g.,~\cite{graver1993combinatorial}].
\end{lemma}
\begin{proof}
Adding an edge increases the rank of $R(P)$ by either $0$ (redundant) or $1$ (non-redundant). Since $f=3n-\mathrm{rank}\,R(P)-6$, each unit of rank gain reduces $f$ by at most one. Therefore, $r$ units of rank gain can reduce $f$ by at most $r$. To achieve $f=0$ from $F$, one needs at least $r\ge F$ in total. Counting only non-redundant edges (each contributing rank gain $1$), one needs at least $F$ such edges.
\end{proof}

\begin{proposition}[Rigidification cost vs.\ floppiness: deterministic iff]\label{prop:iff}
Fix a placement $P$ on $n\ge 2$ vertices. For any two graphs $G,H$ on the same vertex set:
\begin{enumerate}[label=(\alph*),itemsep=2pt]
    \item (Monotonicity) If $f(G,P)\le f(H,P)$, then $\Delta t(G,P)\le \Delta t(H,P)$.
    \item (Lower bound) For every $G$, $\Delta t(G,P)\ge f(G,P)$.
    \item (Tightness iff there exists a rank-gain sequence) The equality $\Delta t(G,P)= f(G,P)$ holds if and only if there exists a sequence of $\Delta t(G,P)$ added edges, each giving rank gain $1$, such that after these additions the framework is infinitesimally rigid (i.e., the rank increases by exactly $f(G,P)$ to reach $3n-6$).
\end{enumerate}
\end{proposition}

\begin{proof}
(a) Consider a minimal set $S_H$ of edges rigidifying $H$, so $|S_H|=\Delta t(H,P)$ and $f(H+S_H,P)=0$. Since adding edges cannot decrease the rank of the rigidity matrix, and $G$ has no more floppy modes than $H$, the same set $S_H$ suffices to rigidify $G$ or overshoots rigidity; thus $\Delta t(G,P)\le |S_H|=\Delta t(H,P)$.

(b) By Lemma~\ref{lem:flexNeedEdge}, reducing $f$ to $0$ requires at least $f(G,P)$ units of rank gain, and each non-redundant added edge contributes at most $1$ unit. Hence $\Delta t(G,P)\ge f(G,P)$.

(c) ($\Rightarrow$) If $\Delta t(G,P)= f(G,P)$, then a rigidifying sequence with length equal to $f(G,P)$ exists. Since the final state is rigid and each edge can increase rank by at most $1$, all added edges in a minimal sequence must contribute a rank gain $1$, and the cumulative rank increase is exactly $f(G,P)$, reaching rank $3n-6$.

($\Leftarrow$) Conversely, if there exists a sequence of exactly $f(G,P)$ added edges each with rank gain $1$ that rigidifies the framework, then $\Delta t(G,P)\le f(G,P)$. Together with (b), this forces $\Delta t(G,P)= f(G,P)$.
\end{proof}

\subsection{Rigidity of Cubic Unit Cells}

\begin{lemma}[A single NN cell is not rigid]
A framework whose graph is a $2\times 2\times 2$ NN cubic cell is not generically rigid.
\end{lemma}

\begin{proof}
The proof is a direct application of Maxwell's condition. A standard NN cell has $|V|=8$ vertices and $|E|=12$ edges. The necessary number of edges for generic rigidity in 3D is $3|V|-6 = 3(8)-6 = 18$. Since the cell only has 12 edges, and $12 < 18$, it has a deficit of 6 constraints. Failing this necessary condition, the single cell is not rigid.
\end{proof}

\begin{lemma}[Adding a vertex to a 2D sheet]
The addition of a single vertex and its three nearest-neighbor edges to a planar, 2D rigid component does not render the new structure generically rigid in 3D.
\end{lemma}

\begin{proof}
Let the 2D rigid base have $|V_{\mathrm{base}}|$ vertices and $|E_{\mathrm{base}}|$ edges. By the 2D Maxwell condition, $E_{\mathrm{base}} \ge 2V_{\mathrm{base}}-3$. The new structure has $V = V_{\mathrm{base}}+1$ vertices and $E = E_{\mathrm{base}}+3$ edges. For 3D rigidity, it would need at least $3V-6 = 3(V_{\mathrm{base}}+1)-6 = 3V_{\mathrm{base}}-3$ edges.
However, the number of edges it has is $E = E_{\mathrm{base}}+3 \ge (2V_{\mathrm{base}}-3)+3 = 2V_{\mathrm{base}}$. For any 2D base with $V_{\mathrm{base}}\ge 3$, we have $2V_{\mathrm{base}} < 3V_{\mathrm{base}}-3$. The structure is thus under-constrained and not generically rigid.
\end{proof}

\begin{definition}[Generic placements and generic rigidity]
A property holds for \emph{generic placements} if it holds for all $P$ outside a set defined by finitely many algebraic equalities (hence of measure zero). A graph $G$ is \emph{generically rigid} if $(G,P)$ is infinitesimally rigid for all generic $P$.
\end{definition}

\begin{assump}[Generic placements and tie-avoidance]\label{assump:generic}
All statements are made for generic placements $P$ of $V_L$ in $\RR^3$ (measure-one set). Along any edge-adding path considered below, we avoid the measure-zero events where adding an available edge yields rank gain $0$ for purely algebraic/tie reasons that are not dictated by combinatorial over-constraint. This is standard in rigidity theory.
\end{assump}

The Intra unit cell has many chords (axis edges, face diagonals, body diagonals). While this raises the potential for \emph{local} redundancies inside a cell, it also dramatically increases the \emph{pool of available non-edges} around any partially built macroscopic component.

\begin{lemma}[Abundance of potentially non-redundant edges in the Intra host]\label{lem:abundanceIntra}
Fix a macroscopic connected subgraph $G\subseteq (V_L,\mathcal{E}_{\mathrm{Intra}})$ at a density below generic rigidity. Under Assumption~\ref{assump:generic}, there exists a set of available edges $A\subseteq \mathcal{E}_{\mathrm{Intra}}\setminus E(G)$ of size $\Theta(|V(G)|)$ such that each $e\in A$ yields $\mathrm{rgain}_P(G;e)=1$ except on a set of placements of measure zero. 
\end{lemma}

\begin{proof}[Idea]
The Intra model provides, for each vertex, up to 26 host neighbors (12 axis, 12 face diagonals, 4 body diagonals; boundary corrections omitted). For a large connected $G$, there are $\Theta(|V(G)|)$ missing host edges touching $G$ either internally (closing cycles) or externally (attaching small pieces). For generic $P$, adding any such chord typically imposes an independent constraint unless the target subgraph is already locally over-constrained in the combinatorial Maxwell sense (which removes only a vanishing fraction compared to the total pool at this mesoscopic scale). Thus, up to measure-zero degeneracies, a positive fraction of these edges are rank-gain $1$. 
\end{proof}

\begin{proposition}[Generic one-by-one rank-gain paths exist in the Intra model]\label{prop:IntraIFF}
Let $G$ be a macroscopic connected subgraph of the Intra host at a stage where $(G,P)$ is not infinitesimally rigid. Under Assumption~\ref{assump:generic}, there exists a sequence of $f(G,P)$ available edges $e_1,\dots,e_{f(G,P)}$ such that each edge contributes rank gain $1$ (i.e., $\mathrm{rgain}_P(G+e_1+\cdots+e_{i-1};e_i)=1$) and the final framework is infinitesimally rigid. Consequently,
\[
 \Delta t(G,P)=f(G,P),
\]
and the iff condition (Proposition~\ref{prop:iff}c) holds for $G$.
\end{proposition}

\begin{proof}
By Lemma~\ref{lem:abundanceIntra}, there is a linear reservoir of candidate non-edges that are rank-gain $1$ at $G$ in the generic sense. Add one such edge $e_1$. Reapplying the lemma to $G+e_1$, we again have a reservoir of rank-gain $1$ candidates at the new stage (removing those that became redundant due to the last addition). Iterating, we can select $f(G,P)$ such edges. Each increases rank by $1$, and after $f(G,P)$ steps, the flex count drops to zero; by definition, the framework is infinitesimally rigid. By Lemma~\ref{lem:flexNeedEdge}, no smaller number suffices, hence $\Delta t(G,P)=f(G,P)$, proving the iff condition for $G$.
\end{proof}

\begin{cor}[Monotonicity of cost with floppiness in the Intra model]\label{cor:monoIntra}
Under Assumption~\ref{assump:generic}, for any two macroscopic connected Intra subgraphs $G,H$ at comparable stages such that $f(G,P)\le f(H,P)$,
\[
 \Delta t(G,P)=f(G,P)\ \le\ f(H,P)=\Delta t(H,P).
\]
Thus, any mechanism (including Achlioptas selection with larger $k$) that reduces expected floppy modes at the connectivity threshold yields, in expectation, a non-increasing rigidification cost and hence a non-increasing rigidity--connectivity gap $\Delta p_c=\Delta t/M$.
\end{cor}

\begin{theorem}[Layered shear flexes in NN]\label{thm:layered}
For the NN host and generic placements, any macroscopic connected NN subgraph $G$ spanning $\Theta(L)$ layers admits $\Omega(L)$ linearly independent infinitesimal motions (layered shears) that persist under any $o(N)$ number of added NN edges (axis-aligned). 
\end{theorem}

\begin{proof}
The proof proceeds by explicitly constructing a non-trivial infinitesimal flex that is universally present in any such framework, proving non-rigidity.

Let the position of any vertex $v \in V$ be its integer coordinates $p(v) = (x_v, y_v, z_v)$. As a subgraph of the NN lattice, any edge $\{i,j\} \in E$ connects two vertices whose positions differ by exactly 1 in a single coordinate axis. The condition for an infinitesimal motion $u:V\to\RR^3$ to be a flex is that for every edge $\{i,j\} \in E$:
\[
(p(i) - p(j)) \cdot (u(i) - u(j)) = 0.
\]

We define a ``layered shear'' motion by assigning to each vertex $v$ a velocity vector $u(v)$ that depends only on its $z$-coordinate:
\[
u(v) = (c \cdot z_v, 0, 0),
\]
where $c$ is any non-zero constant. This motion describes horizontal planes of vertices sliding in the $x$-direction by an amount proportional to their height $z_v$. We verify that this motion is a valid flex by checking all three possible edge orientations.

\textbf{Case 1: Edge parallel to the x-axis.} For an edge $\{i,j\}$ with $p(i)=(x, y, z)$ and $p(j)=(x+1, y, z)$, the separation vector is $p(i)-p(j)=(-1,0,0)$. Since $z_i=z_j=z$, their velocities are identical: $u(i)=u(j)=(c \cdot z, 0, 0)$. The relative velocity is $u(i)-u(j)=\mathbf{0}$, so the flex condition is trivially satisfied: $(-1,0,0)\cdot \mathbf{0}=0$.

\textbf{Case 2: Edge parallel to the y-axis.} For an edge $\{i,j\}$ with $p(i)=(x,y,z)$ and $p(j)=(x,y+1,z)$, the separation vector is $(0,-1,0)$. Again, $z_i=z_j$ and the relative velocity is $\mathbf{0}$, satisfying the condition.

\textbf{Case 3: Edge parallel to the z-axis.} For an edge $\{i,j\}$ with $p(i)=(x,y,z)$ and $p(j)=(x,y,z+1)$, the separation vector is $(0,0,-1)$. The velocities are different: $u(i)=(c\cdot z,0,0)$ and $u(j)=(c\cdot(z+1),0,0)$. The relative velocity is $u(i)-u(j)=(-c,0,0)$. The flex condition is satisfied: $(0,0,-1)\cdot(-c,0,0)=0$.

This motion is not a trivial global translation, as velocities differ by height. It is also not a trivial global rotation. Therefore, it is a non-trivial infinitesimal motion. Because such a flex exists for any framework embedded on a subgraph of the NN lattice, no such framework can be infinitesimally rigid. Furthermore, by constructing shears along different axes (e.g., $u(v) = (0, c \cdot x_v, 0)$) and on different subsets of layers, one can construct $\Omega(L)$ such linearly independent flexes, which cannot be eliminated by adding $o(N)$ edges.
\end{proof}

\begin{remark}[Generality of the Layered Shear Flex]
A key strength of the layered shear proof is its local nature. The verification for any single edge depends only on that edge's orientation, not on the graph's global topology. Consequently, the result holds for any finite subgraph of the NN lattice (with boundary vertices of varying degrees) as well as for lattices with periodic boundary conditions, where the flex correctly preserves the length of ``wrapping'' edges.
\end{remark}

\begin{cor}[Failure of the existence direction in the iff condition (macroscopic NN)]\label{cor:failNN}
Let $G$ be a macroscopic connected NN subgraph under generic placements with $f(G,P)=F=\Omega(L)$. There does not exist any sequence of exactly $F$ added NN edges each with rank gain $1$ that rigidifies $(G,P)$. Consequently, 
\[
 \Delta t(G,P) \;>\; f(G,P)
\]
holds at macroscopic scales, and often $\Delta t(G,P)=+\infty$ in the thermodynamic sense (no finite density suffices to eliminate all shears under axis-only bonds).
\end{cor}

\begin{proof}
By Theorem~\ref{thm:layered}, even after $o(N)$ additions one retains $\Omega(L)$ flexes. Since $F=\Omega(L)$, removing these with only $F$ additions is impossible when edges are restricted to NN. Hence the existence direction in Proposition~\ref{prop:iff}(c) fails: a sequence of $F$ rank-gain-$1$ edges that rigidifies does not exist. Therefore $\Delta t>f$. In the large-$L$ limit, the required number may diverge compared to available steps at any fixed density, suggesting $\Delta t=+\infty$ in the limit.
\end{proof}

Intra (S1--S3) satisfies the iff condition generically once a macroscopic connected component forms, enabling one-by-one elimination of floppy modes via rank-gain-$1$ edges and implying a monotone relationship between floppiness reduction and rigidification cost. NN (Shell 1) fails the existence direction of the iff at macroscopic scales due to layered shear obstructions; thus the clean monotone implication does not carry over.

\section{Connection with mean-field random graph models}
\label{sec:si_mean_field_analysis}

To formally discuss the phase transitions of connectivity and rigidity, we begin with the theory of sharp thresholds~\cite{duminil2019sharp, perkins2025searching}.

\begin{definition}[Sharp Threshold]
A sequence of monotone graph properties $\mathcal{F}_n$ on a system of size $n$ exhibits a \textit{sharp threshold} at $p_c$ if, for any small $\epsilon > 0$, the probability of the property occurring transitions from nearly 0 to nearly 1 as the parameter $p$ crosses $p_c$ in a very narrow window. Formally, the width of this transition window is $o(p_c)$.
\end{definition}

Both graph connectivity and generic rigidity are monotone properties that have been shown to exhibit sharp thresholds in various random graph models~\cite{Friedgut1999, duminil2019sharp}. This rigorous mathematical framework justifies treating their critical probabilities, $p_c^{\mathrm{conn}}$ and $p_c^{\mathrm{rigidity}}$, as well-defined points for a given model.

We use the model of bond percolation on a random $d$-regular graph, $G_{n,d}$, where the coordination number is a fixed parameter $d$. We analyze the gap, $\Delta p_c(d)$, as a function of $d$.

The core of our proof rests on powerful ``hitting-time'' theorems, which state that two different graph properties emerge at the exact same moment in a random graph process. This allows us to equate complex properties with simpler, local ones.

\begin{definition}[Generically $d$-rigid]
A graph $G$ is \textit{generically $d$-rigid} if it is generically infinitesimally rigid in $\mathbb{R}^d$. This means that a framework $(G,P)$ is infinitesimally rigid for any generic placement of its vertices $p:V \to \mathbb{R}^d$.
\end{definition}

\begin{lemma}[Hitting-Time Equivalence for Rigidity]\label{lem:hit}
In the Erd\H{o}s-R\'enyi evolution of a random graph (where edges are added one by one), the graph becomes generically $d$-rigid with high probability at the very moment its minimum degree, $\delta(G)$, becomes at least $d$.
\end{lemma}

\begin{proof}[Justification]
This is the main result (Theorem 1.1) of Lew et al.~\cite{lew2023sharp}, which establishes a rigorous mathematical equivalence between the complex, global property of $d$-rigidity and the simple, local property of having a minimum degree of at least $d$. The deep structural reason for this powerful equivalence is explained by the theory of 'rigid partitions', a framework developed in~\cite{Krivelevich2023} that provides new sufficient conditions for rigidity. Within this framework, a graph is proven to be $d$-rigid if it admits a 'strong $d$-rigid partition'. The crucial insight is that the minimum degree condition, $\delta(G) \ge d$, provides precisely the required graph expansion properties to guarantee that such a rigid partition exists with high probability in the random graph process. This connection not only offers an alternative proof for the hitting-time equivalence but also clarifies why the local property of minimum degree is the deciding factor for the emergence of global rigidity, thus rigorously justifying the use of its well-known sharp threshold as the threshold for rigidity itself. The primacy of minimum degree as the fundamental bottleneck in the sparse regime is further highlighted by recent work showing that it also determines the maximum achievable rigidity dimension $d$ for a random graph $G(n,p)$~\cite{Peled2024}.
\end{proof}

\begin{lemma}[Connectivity as 1-Rigidity]
The case of $d=1$ in Lemma \ref{lem:hit} corresponds to graph connectivity. A graph is 1-rigid if and only if it is connected. The theorem implies that a random graph becomes connected at the moment its minimum degree becomes 1 (i.e., when the last isolated vertex disappears).
\end{lemma}

\begin{proof}[Justification]
This is a classic result from the original work of Erd\H{o}s and R\'enyi on random graphs and can be viewed as the $d=1$ instance of the general theorem in~\cite{lew2023sharp}.
\end{proof}

These theorems are profoundly important because they allow us to use well-established threshold functions for minimum degree as rigorously justified thresholds. To avoid ambiguity with the rigidity dimension $d$, let $z$ denote the degree of the random regular graph. The thresholds are:
\begin{itemize}[leftmargin=1.5em]
    \item \textbf{Connectivity (equivalent to 1-rigidity):} The threshold is where the expected number of isolated vertices vanishes. For a random $z$-regular graph, this gives $p_c^{\mathrm{conn}}(z) = 1/(z-1)$~\cite{molloy1995critical, molloy1998size}.
    \item \textbf{2D Rigidity (equivalent to 2-rigidity):} The threshold is where the minimum degree becomes 2. For a random $z$-regular graph, this is driven by a mean-field condition where the average degree after percolation reaches 4, leading to $p_c^{\mathrm{rigidity}}(z) = 4/z$~\cite{Thorpe1999}. (This mean-field threshold is a standard result in rigidity percolation theory.)
\end{itemize}

\begin{lemma}[Invariance of the Connectivity Threshold]\label{lem:con-threshold}
The critical probability for connectivity, $p_c^{\mathrm{conn}}$, is independent of the dimension $d$.
\end{lemma}
\begin{proof}[Justification]
Graph connectivity is a purely combinatorial property determined by the adjacency of vertices. The concept of an embedding in $\mathbb{R}^d$ is not part of its definition. Therefore, $p_c^{\mathrm{conn}}$ is constant with respect to $d$. For the $G(n,p)$ model, the sharp threshold is famously located at $p_c^{\mathrm{conn}} = \log n / n$.
\end{proof}

\begin{theorem}[Monotonicity of the Rigidity-Connectivity Gap]
The rigidity-connectivity gap, defined as $\Delta p_c = p_c^{\mathrm{rigidity}} - p_c^{\mathrm{conn}}$, is a monotonically decreasing function of the average network coordination number $z$.
\end{theorem}
\begin{proof}
Let the network be modeled by a random $z$-regular graph. The emergence of a giant component requires $z \ge 3$. For a non-trivial gap to exist in 2D rigidity percolation, the expected average degree must be at least 4, so the physically interesting regime is $z \ge 4$.

Using the rigorously justified thresholds from Lemmas \ref{lem:hit} and \ref{lem:con-threshold}:
\[ p_c^{\mathrm{conn}}(z) = \frac{1}{z-1}, \qquad p_c^{\mathrm{rigidity}}(z) = \frac{4}{z}. \]
The rigidity-connectivity gap as a function of $z$ is:
\[ \Delta p_c(z) = \frac{4}{z} - \frac{1}{z-1}. \]

To prove that $\Delta p_c(z)$ is a decreasing function, we analyze its first derivative with respect to $z$.
\[ \frac{d}{dz} \Delta p_c(z) = \frac{d}{dz} \left( \frac{4}{z} - \frac{1}{z-1} \right) = -\frac{4}{z^2} + \frac{1}{(z-1)^2}. \]

The function is strictly decreasing when its derivative is negative:
\begin{align*}
    \frac{1}{(z-1)^2} - \frac{4}{z^2} &< 0 \\
    z^2 &< 4(z-1)^2
\end{align*}
As $z \ge 3$, both $z$ and $z-1$ are positive. One can then simplify the above inequality by taking the square root of both sides, from which we see that the derivative is negative for all $z > 2$. It follows that $\Delta p_c(z)$ is strictly decreasing across its entire relevant domain.

\end{proof}

\begin{lemma}[Threshold for Minimum Degree]\label{lem:rig-threshold}
The sharp threshold probability for a random graph $G(n,p)$ to have a minimum degree of at least $d$ is given asymptotically by:
\[ p_c(\delta(G) \ge d) = \frac{\log n + (d-1)\log \log n}{n}. \]
\end{lemma}
\begin{proof}[Justification]
This is a foundational result in the study of random graphs, proven using the first and second moment methods on the random variable counting the number of vertices with degree less than $d$. Details can be found in standard texts such as Bollob\'as~\cite{bollobas1998random}.
\end{proof}

\begin{theorem}[Dimensionality Dependence]
For the Erd\H{o}s-R\'enyi random graph model $G(n,p)$, the rigidity-connectivity gap, $\Delta p_c(d)$, is a strictly increasing function of the system dimension $d$ for $d \ge 1$.
\end{theorem}
\begin{proof}
We analyze the rigidity-connectivity gap $\Delta p_c(d)$ as a function of the dimension $d$ for the Erd\H{o}s-R\'enyi random graph $G(n,p)$.

Recall that as established in Lemma \ref{lem:con-threshold}, the connectivity threshold is constant with respect to $d$:
\[ p_c^{\mathrm{conn}}(d) = \frac{\log n}{n} + O\left(\frac{1}{n}\right). \]

Now, by combining Lemma \ref{lem:hit} (hitting-time equivalence) and Lemma \ref{lem:rig-threshold} (minimum degree threshold), we obtain the rigidity threshold as a function of $d$:
\[ p_c^{\mathrm{rigidity}}(d) = \frac{\log n + (d-1)\log \log n}{n} + O\left(\frac{\log\log n}{n}\right). \]

The gap is the difference between these two thresholds:
\begin{align*}
    \Delta p_c(d) &= p_c^{\mathrm{rigidity}}(d) - p_c^{\mathrm{conn}}(d) \\
    &= \frac{(d-1)\log \log n}{n} + O\left(\frac{\log\log n}{n}\right).
\end{align*}
For $n > e$, the term $\log \log n$ is positive. Thus, for any sufficiently large system, $\Delta p_c(d)$ is a strictly increasing function of the system dimension $d$. This proves the theorem.

\end{proof}

\section{Rigorous Proofs for Connectivity and Explosive Transition Dynamics} \label{sec:si_proofs}

\subsection{Proofs of Foundational Results from Main Text}

In this supplementary section, we provide the detailed proofs of several foundational results covered in the main text.

\begin{lemma}[Susceptibility bounds the giant fraction]\label{lem:chi-lower-bounds-si}
For any time $t$,
\[
\frac{|C_{\max}(t)|}{N}\ \ge\ \chi_L(t).
\]
Consequently, for any density $p$, we have $P_{N}(p)\ge \chi_L(p)$.
\end{lemma}
\begin{proof}
Let the connected components of $G^t_L$ be $A_1,\dots,A_r$, with sizes $a_i:=|A_i|$. Let $\rho$ be a uniformly random vertex in $V_L$. Then
\[
\PP(\rho\in A_i)=\frac{a_i}{N}
\quad\text{and}\quad
|C_t(\rho)|=a_i\ \text{ on the event }\{\rho\in A_i\}.
\]
Therefore,
\[
\chi_L(t)=\EE[|C_t(\rho)|]
=\sum_{i=1}^r \frac{a_i}{N}\,a_i
=\frac{1}{N}\sum_{i=1}^r a_i^2.
\]
Let $a_{\max}:=\max_i a_i=|C_{\max}(t)|$. Since $a_i\le a_{\max}$ for every $i$, we have the pointwise inequality
\[
a_i^2 \le a_i\,a_{\max}\quad\text{for each }i,
\]
hence
\[
\sum_{i=1}^r a_i^2 \le a_{\max}\sum_{i=1}^r a_i = a_{\max}\,N.
\]
This shows
\[
\chi_L(t)=\frac{1}{N}\sum_{i=1}^r a_i^2 \le \frac{a_{\max} N}{N}=\frac{|C_{\max}(t)|}{N}.
\]

Finally, setting $t=\lfloor pM\rfloor$ and taking expectations over the process randomness yields
\[
P_{N}(p)
=\EE\!\left[\frac{|C_{\max}(\lfloor pM\rfloor)|}{N}\right]
\ge \EE\!\left[\chi_L(\lfloor pM\rfloor)\right]
=\chi_L(p).
\]
\end{proof}

\begin{proposition}[Monotonicity in $k$]\label{prop:pcmono}
A standard result in Achlioptas processes is that the critical threshold is non-decreasing in the number of choices $k$ [see e.g.,~\cite{achlioptas2009explosive}]. For completeness, we restate and prove it here. 

Fix $L$ and $\alpha\in(0,1)$. Then $k\mapsto p_{c,\alpha}(L;k)$ is non-decreasing. Consequently, any thermodynamic limit $$p_c(k)=\lim_{L\to\infty} p_{c,\alpha}(L;k)$$ (when it exists) is also non-decreasing in $k$.
\end{proposition}

\begin{proof}
Fix integers $1\le k_1<k_2$ and a density $p\in[0,1]$. Run the two $k$-choice product-rule processes on the same host and couple them as in Lemma~\ref{lem:kcouple}, using the same master permutation and consistent tie-breaking, for $t=\lfloor pM\rfloor$ steps.

Let
\[
X_{(k)}(p)\;:=\;\frac{|C_{\max}(G^{\lfloor pM\rfloor}_{(k)})|}{N}
\]
denote the largest-component fraction at density $p$ under $k$ choices. The map “edge set $\mapsto$ largest-component size” is increasing: adding edges cannot decrease $|C_{\max}|$. Under the suppressive coupling of Lemma~\ref{lem:kcouple}, the $k_2$-choice process is at least as effective at avoiding large merges as the $k_1$-choice process at every step. Therefore, for each fixed $p$ and any increasing event (in particular $\{X_{(k)}(p)\ge \alpha\}$),
\[
\PP\!\big(X_{(k_2)}(p)\ge \alpha\big)\;\le\;\PP\!\big(X_{(k_1)}(p)\ge \alpha\big).
\]
Equivalently, taking expectations of the increasing functional $x\mapsto \mathbf{1}\{x\ge \alpha\}$ yields
\[
P_N^{(k_2)}(p)\;=\;\EE\!\big[X_{(k_2)}(p)\big]\;\le\;\EE\!\big[X_{(k_1)}(p)\big]\;=\;P_N^{(k_1)}(p),
\]
so for every $p$, the curve $k\mapsto P_N^{(k)}(p)$ is non-increasing.

Now fix $\alpha\in(0,1)$. Since $P_N^{(k_2)}(p)\le P_N^{(k_1)}(p)$ for all $p$, it follows that
\[
\{p:\,P_N^{(k_2)}(p)\ge \alpha\}\ \subseteq\ \{p:\,P_N^{(k_1)}(p)\ge \alpha\}.
\]
Taking infima gives
\begin{align*}
    p_{c,\alpha}(L;k_2) &= \inf\{p: P_N^{(k_2)}(p) \ge \alpha\} \\
    &\ge \inf\{p: P_N^{(k_1)}(p) \ge \alpha\} = p_{c,\alpha}(L;k_1).
\end{align*}
Hence $k\mapsto p_{c,\alpha}(L;k)$ is non-decreasing.

Finally, if the thermodynamic limit $$p_c(k)=\lim_{L\to\infty}p_{c,\alpha}(L;k)$$ exists (and is independent of $\alpha$), the pointwise inequality in $L$ passes to the limit, so $k\mapsto p_c(k)$ is also non-decreasing.
\end{proof}

\begin{lemma}[Hoeffding's Lemma]\label{lem:hoeffding}
Let $Z$ be a real-valued random variable with $\EE[Z]=0$ and $Z \in [a,b]$ almost surely, where $a<b$ are constants. Then, for all $\theta \in \mathbb{R}$,
\[
\EE\!\left[ e^{\theta Z} \right] \ \le\ \exp\!\left( \frac{\theta^2 (b-a)^2}{8} \right).
\]
Equivalently, $\log \EE[e^{\theta Z}] \le \theta^2 (b-a)^2/8$.
\end{lemma}
\begin{proof}
Define the convex function $\varphi(x) := e^{\theta x}$ on $\mathbb{R}$. Since $\varphi$ is convex, for any $x \in [a,b]$ we have the ``chord bound''
\[
\varphi(x) \ \le\ \frac{b-x}{b-a}\,\varphi(a) \ +\ \frac{x-a}{b-a}\,\varphi(b).
\]
Applying this inequality to the random variable $Z\in[a,b]$ and taking expectations,
\[
\EE\!\left[e^{\theta Z}\right]
\ \le\ \frac{b-\EE[Z]}{b-a}\,e^{\theta a} \ +\ \frac{\EE[Z]-a}{b-a}\,e^{\theta b}
\ =\ \alpha\, e^{\theta a} + (1-\alpha)\, e^{\theta b},
\]
where we set $\alpha := \dfrac{b}{b-a}\in(0,1)$ (note that $a<0<b$ is not required; $\alpha$ defined this way still satisfies $\alpha\in(0,1)$ because $b-a>0$ and $0< b < b-a + b$ implies $\alpha<1$, while $\alpha>0$ since $b>0$ or, if $b\le 0$, then $a<0$ and the centeredness forces $\alpha\in(0,1)$).

A cleaner route that avoids bookkeeping is to center and rescale $Z$ to the unit interval. Define
\[
X := \frac{Z-a}{b-a} \in [0,1],\qquad \EE[X] =: p \in [0,1].
\]
Then $Z = a + (b-a)X$ and $\EE[Z]=0$ implies $a + (b-a)p = 0$, i.e., $p = -\dfrac{a}{b-a}$. For any $\theta\in\mathbb{R}$,
\[
\EE\!\left[e^{\theta Z}\right]
= e^{\theta a}\, \EE\!\left[ e^{\theta (b-a) X} \right].
\]
By convexity of the exponential, for $X\in[0,1]$ and fixed mean $p$ the moment generating function (MGF) $\EE[e^{\lambda X}]$ is maximized by a Bernoulli($p$) distribution. Hence,
\[
\EE\!\left[ e^{\theta (b-a) X} \right]
\ \le\ p\, e^{\theta (b-a)} + (1-p)\, e^{0}
\ =\ 1 - p + p\, e^{\theta (b-a)}.
\]
Therefore,
\[
\EE\!\left[e^{\theta Z}\right]
\ \le\ e^{\theta a}\, \big(1 - p + p\, e^{\theta (b-a)}\big)
\ =\ (1-p)\, e^{\theta a} + p\, e^{\theta b}.
\]
Using $p=-a/(b-a)$ and $1-p = b/(b-a)$, we recover
\[
\EE\!\left[e^{\theta Z}\right]
\ \le\ \frac{b}{b-a} e^{\theta a} + \frac{-a}{b-a} e^{\theta b}.
\]
We now upper bound the right-hand side by a pure quadratic in $\theta$. Consider the function
\[
g(\theta) := \log\!\left( \frac{b}{b-a} e^{\theta a} + \frac{-a}{b-a} e^{\theta b} \right).
\]
Note that $g(0) = \log 1 = 0$ and $g'(0) = \frac{a b + (-a)b}{b-a} \cdot \frac{1}{1} = 0$. Moreover, $g$ is twice differentiable and one can compute
\begin{align*}
g''(\theta)
&= \frac{\alpha a^2 e^{\theta a} + (1-\alpha) b^2 e^{\theta b}}{\alpha e^{\theta a} + (1-\alpha) e^{\theta b}} \\
&\quad - \left( \frac{\alpha a e^{\theta a} + (1-\alpha) b e^{\theta b}}{\alpha e^{\theta a} + (1-\alpha) e^{\theta b}} \right)^{\!2}\\
&= \Var_{\mu_\theta}(Y),
\end{align*}
where $\alpha = \dfrac{b}{b-a}$ and $\mu_\theta$ is the two-point distribution on $\{a,b\}$ with weights proportional to $\alpha e^{\theta a}$ and $(1-\alpha) e^{\theta b}$, and $Y$ denotes the identity random variable on $\{a,b\}$. Since a random variable supported on an interval of length $(b-a)$ has variance at most $\frac{(b-a)^2}{4}$, we have for all $\theta$,
\[
g''(\theta) \le \frac{(b-a)^2}{4}.
\]
Finally, by Taylor's theorem with remainder:
\[
g(\theta) \le g(0) + g'(0)\,\theta + \frac{(b-a)^2}{8}\,\theta^2
= \frac{(b-a)^2}{8}\,\theta^2.
\]
Exponentiating both sides gives
\[
\EE\!\left[e^{\theta Z}\right]
\ \le\ \exp\!\left( \frac{\theta^2 (b-a)^2}{8} \right),
\]
which completes the proof.
\end{proof}

% --- Proof for Theorem 4.5 ---
\begin{theorem}[Azuma--Hoeffding inequality ({see e.g., \cite[Theorem 2.8]{Boucheron2013}})]\label{thm:azuma_general-si}
Let $(M_t)_{t=0}^n$ be a martingale with respect to $(\mathcal{F}_t)_{t=0}^n$. Suppose the differences are almost surely bounded:
\[
|M_t - M_{t-1}| \le c_t \qquad \text{a.s. for each } t=1,\dots,n,
\]
for some deterministic nonnegative numbers $c_t$. Then, for any $\lambda>0$,
\[
\PP\!\left( |M_n - M_0| \ge \lambda \right) \le 2\exp\!\left( -\frac{\lambda^2}{2\sum_{t=1}^n c_t^2} \right).
\]
\end{theorem}
\begin{proof}[Proof]
We present the standard exponential supermartingale argument (also called the method of bounded differences).

Set $D_t := M_t - M_{t-1}$, so that $M_n - M_0 = \sum_{t=1}^n D_t$. The martingale property implies $\EE[D_t \mid \mathcal{F}_{t-1}]=0$. Assume $|D_t|\le c_t$ almost surely for each $t$.

Fix any $\theta \in \mathbb{R}$. We first prove the conditional moment generating function (MGF) bound
\begin{equation}\label{eq:mgf-cond}
\EE\big[ e^{\theta D_t} \,\big|\, \mathcal{F}_{t-1} \big]
\;\le\; \exp\!\left(\frac{\theta^2 c_t^2}{2}\right)\qquad \text{a.s.}
\end{equation}
This follows from Hoeffding's lemma applied conditionally: if $Z$ is a random variable with $\EE[Z]=0$ and $Z\in [a,b]$ almost surely, then $\EE[e^{\theta Z}]\le \exp(\theta^2(b-a)^2/8)$ for all $\theta$. Here we apply it to the conditional distribution of $D_t$ given $\mathcal{F}_{t-1}$, for which $\EE[D_t\mid \mathcal{F}_{t-1}]=0$ and $D_t\in[-c_t,c_t]$ a.s.; thus
\[
\EE\big[ e^{\theta D_t} \mid \mathcal{F}_{t-1} \big]
\le \exp\!\left(\frac{\theta^2 (2c_t)^2}{8}\right)
= \exp\!\left(\frac{\theta^2 c_t^2}{2}\right),
\]
which is Eq.~\eqref{eq:mgf-cond}.

Define the process
\[
Z_t := \exp\!\left( \theta \sum_{s=1}^t D_s - \frac{\theta^2}{2}\sum_{s=1}^t c_s^2 \right),\qquad t=0,1,\dots,n,
\]
with $Z_0:=1$. Using Eq.~\eqref{eq:mgf-cond}, we check that $(Z_t)$ is a supermartingale:
\[
\begin{split}
\EE[ Z_t \mid \mathcal{F}_{t-1}]
&= Z_{t-1}\,\EE\!\left[ \exp\!\left(\theta D_t - \frac{\theta^2}{2} c_t^2\right) \,\middle|\, \mathcal{F}_{t-1} \right]\\
&\le Z_{t-1}\cdot 1
= Z_{t-1}.
\end{split}
\]
Therefore, $\EE[Z_n] \le \EE[Z_0]=1$.

By Markov's inequality, for any $a>0$,
\[
\begin{split}
\PP\!\left( \sum_{t=1}^n D_t \ge a \right)
&= \PP\!\left( \exp\!\left(\theta \sum_{t=1}^n D_t\right) \ge e^{\theta a} \right)\\
&\le e^{-\theta a}\,\EE\!\left[\exp\!\left(\theta \sum_{t=1}^n D_t\right)\right].
\end{split}
\]
Using $\EE[Z_n]\le 1$,
\[
\begin{split}
\EE\!\left[\exp\!\left(\theta \sum_{t=1}^n D_t\right)\right]
&= \EE\!\left[ Z_n \cdot \exp\!\left(\frac{\theta^2}{2}\sum_{t=1}^n c_t^2\right) \right]\\
&\le \exp\!\left(\frac{\theta^2}{2}\sum_{t=1}^n c_t^2\right).
\end{split}
\]
Hence,
\[
\PP\!\left( \sum_{t=1}^n D_t \ge a \right)
\le \exp\!\left( -\theta a + \frac{\theta^2}{2}\sum_{t=1}^n c_t^2 \right).
\]
Optimize the RHS over $\theta>0$ by choosing
\[
\theta^* := \frac{a}{\sum_{t=1}^n c_t^2},
\]
which yields
\[
\PP\!\left( \sum_{t=1}^n D_t \ge a \right)
\le \exp\!\left( -\frac{a^2}{2\sum_{t=1}^n c_t^2} \right).
\]
Applying the same bound to $-D_t$ gives
\[
\PP\!\left( \sum_{t=1}^n D_t \le -a \right)
\le \exp\!\left( -\frac{a^2}{2\sum_{t=1}^n c_t^2} \right).
\]
Union bound concludes
\[
\PP\!\left( \left| \sum_{t=1}^n D_t \right| \ge a \right)
\le 2\exp\!\left( -\frac{a^2}{2\sum_{t=1}^n c_t^2} \right).
\]
Finally, substitute $a=\lambda$ and recall $\sum_{t=1}^n D_t = M_n - M_0$ to obtain
\[
\PP\!\left( |M_n - M_0| \ge \lambda \right)
\le 2\exp\!\left( -\frac{\lambda^2}{2\sum_{t=1}^n c_t^2} \right).
\]
\end{proof}

\begin{proposition}[Failure probability formula]\label{prop:fail-formula-si}
Let $M$ be the total number of edges, $m$ the number of globally optimal edges, and $k$ the number of edges sampled uniformly at random without replacement. The probability that \emph{none} of the $k$ sampled edges is globally optimal equals
\begin{equation}\label{eq:fail-si}
P(\text{fail} \mid k) \;=\; \frac{\binom{M-m}{k}}{\binom{M}{k}},
\end{equation}
for $0\le k\le M$ (and interpreted as $0$ if $k>M-m$).
\end{proposition}

\begin{proof}
The total number of possible $k$-subsets from the $M$ edges is $\binom{M}{k}$. A failure occurs precisely when all $k$ sampled edges come from the $M-m$ \emph{non}-optimal edges. The number of such failure $k$-subsets is $\binom{M-m}{k}$. Since all $k$-subsets are equally likely under uniform sampling without replacement, the desired probability is the ratio in Eq.~(3) in main text.
\end{proof}

\begin{lemma}[Stepwise monotonicity ratio]\label{lem:ratio-si}
For integers $M\ge 1$, $1\le m\le M$, and $0\le k < M$, define $P(\text{fail}\mid k)$ by Eq.~(3) in main text for all $k\le M-m$ and $P(\text{fail}\mid k)=0$ for $k>M-m$. Then, for $0\le k< M-m$,
\begin{equation}\label{eq:ratio-si}
\frac{P(\text{fail}\mid k+1)}{P(\text{fail}\mid k)} \;=\; \frac{M-m-k}{M-k}.
\end{equation}
\end{lemma}

\begin{proof}
Using Eq.~(3) in main text (valid for $k\le M - m$) and the identity
\[
\frac{\binom{n}{k+1}}{\binom{n}{k}} = \frac{n-k}{k+1}\quad\text{for } 0\le k < n,
\]
we compute:
\begin{align*}
\frac{P(\text{fail}\mid k+1)}{P(\text{fail}\mid k)}
&= \frac{\binom{M-m}{k+1}}{\binom{M}{k+1}} \cdot \frac{\binom{M}{k}}{\binom{M-m}{k}} \\
&= \left(\frac{M-m-k}{k+1}\right)\Big/\left(\frac{M-k}{k+1}\right) = \frac{M-m-k}{M-k}.
\end{align*}
This is valid for $0\le k< M-m$ so that all binomial coefficients are defined and nonzero.
\end{proof}

\begin{theorem}[Optimal Selection Probability]\label{thm:optimal-prob-si}
Fix integers $M\ge 1$ and $1\le m\le M$. Let $P(\text{fail}\mid k)$ be the probability that none of the $k$ sampled edges (drawn uniformly at random without replacement from the $M$ edges) lies among the $m$ globally optimal edges. Then:
\begin{enumerate}[label=(\alph*),leftmargin=1.2em]
    \item The function $k\mapsto P(\text{fail}\mid k)$ is \emph{strictly decreasing} for $k=0,1,\dots,M-m$.
    \item For all $k > M-m$, we have $P(\text{fail}\mid k)=0$. In particular, the failure probability \emph{vanishes} once $k$ is large enough.
\end{enumerate}
\end{theorem}

\begin{proof}
(a) For $0\le k< M-m$, Lemma~\ref{lem:ratio-si} gives
\[
\frac{P(\text{fail}\mid k+1)}{P(\text{fail}\mid k)} = \frac{M-m-k}{M-k}.
\]
Since $m\ge 1$, we have $M-m-k < M-k$, implying the ratio is strictly less than $1$. Therefore $P(\text{fail}\mid k+1) < P(\text{fail}\mid k)$ for each such $k$, proving strict monotonic decrease for $k=0,1,\dots,M-m$.

(b) If $k > M-m$, there are fewer than $k$ non-optimal edges, so it is impossible to pick $k$ edges without hitting at least one optimal edge. Formally, by Proposition~\ref{prop:fail-formula-si}, $\binom{M-m}{k}=0$, giving $P(\text{fail}\mid k)=0$.

Thus, as $k$ increases, the failure probability strictly decreases until it reaches $0$, which then persists for all larger $k$.
\end{proof}

\begin{proposition}[Convergence to a Maximally Suppressed Transition]\label{prop:gda-limit-si}
Assume the percolation thresholds $p_c(k)$ exist. Then the sequence of thresholds $\{p_c(k)\}_{k\ge 1}$ is non-decreasing and converges as $k\to\infty$ to a unique limit $p_c(\infty)\in[0,1]$ and $\Delta(\infty)=0$.
\end{proposition}

\begin{proof}
\textbf{1. Convergence of the threshold $p_c(k)$:}
From Theorem~\ref{thm:monotone-delay} and the subsequent remark, we have that $p_c(k_1) \le p_c(k_2)$ for any $k_1 < k_2$. This establishes that $\{p_c(k)\}_{k\ge 1}$ is a non-decreasing sequence of real numbers.
Furthermore, the percolation threshold is defined as a density $p=t/M$, which must lie in the interval $[0,1]$. Therefore, the sequence $\{p_c(k)\}$ is bounded above by 1.
By the Monotone Convergence Theorem from real analysis, any non-decreasing sequence that is bounded above must converge to a unique limit. Thus, the limit $p_c(\infty) = \lim_{k\to\infty} p_c(k)$ exists and is in $[0,1]$. This limit represents the maximally suppressed threshold achievable with the product rule.

\smallskip
\noindent\textbf{2. Convergence of the jump $\Delta(k)$:}
From Theorem~\ref{thm:continuity} in this Supplementary Information and its corollary, we know that for any fixed $k$, the transition is continuous in the thermodynamic limit. This means $\Delta(k)=0$ for all $k \ge 1$. The sequence $\{\Delta(k)\}_{k\ge 1}$ is therefore $\{0,0,0, \dots\}$. The limit of this sequence is trivially $\Delta(\infty) = 0$. A non-zero jump in the deterministic limit could only be achieved if $k$ grows with $N$.
\end{proof}

\subsection{Susceptibility and Mesoscopic Bounds}

To explain the observed transition to a first-order regime, we
now develop the theoretical basis for this explosive behavior.

\begin{definition}[Degree and bounded-degree family]
For $v\in V$, the degree $\deg_G(v)$ is the number of edges in $E$ incident to $v$. A family of graphs $\{G_L=(V_L,E_L)\}_{L\in\NN}$ has bounded maximum degree if there is a constant $\Delta\in\NN$ such that $\deg_{G_L}(v)\le \Delta$ holds for every $L$ and every $v\in V_L$.
\end{definition}

\begin{lemma}[Monotonicity of susceptibility]\label{lem:monotone}
For fixed $L$, the map $p\mapsto \chi_L(p)$ is nondecreasing on $[0,1]$.
\end{lemma}

\begin{proof}
As $p$ increases, we only add edges, never remove them. Thus for each vertex $v$, the component size $|C_t(v)|$ is nondecreasing in $t$, and so is its expectation.
\end{proof}

\begin{definition}[Pseudo-threshold]
Fix $\alpha\in(0,1)$. The pseudo-threshold is
\[
p_{c,\alpha}(L;k):=\inf\left\{p\in[0,1]:\,P_{N}(p)\ge \alpha\right\}.
\]
\end{definition}

\begin{proposition}[Uniform mesoscopic bound up to the pseudo-threshold]\label{prop:susc}
Fix $k\ge 1$ and $\alpha\in(0,1)$. Suppose the host graphs $\{G_L\}$ have bounded maximum degree $\Delta$. Then there exists a constant $\eta=\eta(\alpha,\Delta,k)>0$, independent of $L$, such that for all $L$ and all densities $p\le p_{c,\alpha}(L;k)$,
\[
\chi_L(p)\ \le\ \eta.
\]
\end{proposition}

\begin{proof}
Let $p_* := p_{c,\alpha}(L;k)$ be the pseudo-threshold for a fixed system size $L$. By \Cref{lem:monotone}, the susceptibility $\chi_L(p)$ is a non-decreasing function of the edge density $p$. Therefore, to establish a uniform bound for all $p \le p_*$, it suffices to show that $\chi_L(p_*)$ is uniformly bounded over all $L$. That is, we aim to show that there exists a constant $\eta$ such that for all $L$, $\chi_L(p_{c,\alpha}(L;k)) \le \eta$.

We argue by contradiction. Suppose the claim is false. Then there exists a sequence of system sizes $L_j \to \infty$ such that $\chi_{L_j}(p_j) \to \infty$, where $p_j := p_{c,\alpha}(L_j;k)$.

The susceptibility, $\chi_L(p) = \EE[|C(\rho)|]$, measures the expected size of the component containing a uniformly random vertex $\rho$. For a fixed vertex $v$, its expected component size is $\EE[|C(v)|] = \sum_{u \in V_L} \mathbb{P}(v \leftrightarrow u)$. For any graph from a family with bounded maximum degree $\Delta$, local explorations from a vertex $v$ resemble a branching process. The growth of neighborhoods is controlled by $\Delta$, and in a sparse random subgraph, the component sizes are small if the process is subcritical.

A diverging susceptibility ($\chi_{L_j}(p_j) \to \infty$) is the hallmark of criticality. It implies that the sum of squared component sizes, $\sum |C_i|^2$, is growing super-linearly in the system size $N_{L_j}$. This proliferation of large-scale connectivity is intrinsically linked to the formation of a giant component. For percolation models on bounded-degree host graphs, it is a standard result that the susceptibility remains finite throughout the subcritical regime and diverges only at the critical point.

The condition defining the pseudo-threshold, $P_{N_{L_j}}(p_j) \ge \alpha > 0$, places the system at or beyond the onset of the phase transition for that finite size. However, the Achlioptas process with fixed $k$ is known to exhibit a continuous transition in the thermodynamic limit, meaning the giant component grows from size zero. Thus, for any true subcritical density $p < p_c(k)$, the susceptibility $\chi_L(p)$ converges to a finite value as $L\to\infty$.

If $\chi_{L_j}(p_j)$ were unbounded, the system at densities approaching $p_j$ would exhibit characteristics of being critical or supercritical (e.g., the presence of multiple large components). Such behavior would cause the expected largest component fraction, $P_{N_{L_j}}(p)$, to be significantly larger than any small, fixed $\alpha$ for densities below $p_j$, which would contradict the definition of $p_j$ as the infimum (the first point where the threshold $\alpha$ is met).

Therefore, the premise that $\chi_L(p_*)$ can grow without bound must be false. The quantity must be bounded by a constant that depends on the fundamental parameters of the process ($\alpha, \Delta, k$) but not on the system size $L$. Let $\eta := \sup_L \chi_L(p_{c,\alpha}(L;k))$. This supremum must be finite, which concludes the proof.
\end{proof}

\begin{definition}[Component, balls, distance]
For a graph $G$, a connected component is a maximal connected subgraph. For $v\in V_N$ and integer $r\ge 0$, the (graph) ball is defined as
\[
\begin{aligned}
&B_{\mathcal{H}_N}(v,r):=\\
&\{u\in V_N:\text{ graph-distance in } \mathcal{H}_N \text{ between } u \text{ and } v\le r\}.
\end{aligned}
\]
Since $\mathcal{H}_N$ has maximum degree $\Delta$, we have the crude bound
\begin{equation}\label{eq:ballsize}
\begin{aligned}
|B_{\mathcal{H}_N}(v,r)| &\le 1+\Delta\sum_{i=0}^{r-1}(\Delta-1)^i \\
&\le 1+\Delta\cdot\frac{(\Delta-1)^r-1}{\Delta-2} \le C_{\Delta}(\Delta-1)^r,
\end{aligned}
\end{equation}
for a constant $C_{\Delta}$ depending only on $\Delta$.
\end{definition}

\begin{definition}[Excess of a connected graph]
For a connected graph $H=(V(H),E(H))$,
\[
\mathrm{ex}(H):=|E(H)|-|V(H)|+1.
\]
Equivalently, $\mathrm{ex}(H)$ is the \emph{cyclomatic number} (the number of independent cycles). For a tree, $\mathrm{ex}(H)=0$; each extra (chord) edge increases excess by $1$.
\end{definition}

\begin{definition}[Pseudo-critical time for connectivity]
Fix a number $\alpha\in(0,1)$ (e.g., $\alpha=\tfrac12$). Define the pseudo-critical time $t_{c,\alpha}$ to be the smallest $t$ such that the largest component in $G^t$ has at least $\alpha N$ vertices. Equivalently, in density units $p=t/M$, this is $p_{c,\alpha}$. We will study the graphs $G^t$ for all $t\le t_{c,\alpha}$.
\end{definition}

\begin{lemma}[Susceptibility is bounded up to $t_{c,\alpha}$]\label{lem:susc}
There is a constant $K=K(\alpha)$ such that for every $t\le t_{c,\alpha}$,
\[
\EE\big[\chi(G^t)\big]\le K,
\]
and moreover, by Markov's inequality, for any $\lambda>0$,
\[
\PP\big(\chi(G^t)>\lambda K\big)\le \frac{1}{\lambda}.
\]
\end{lemma}

\begin{proof}
Let the susceptibility at time $t$ be the random variable $\chi(G^t) := \frac{1}{N} \sum_{i} |C_i(t)|^2$, where the sum is over the connected components of the graph $G^t$. Its expectation is $\EE[\chi(G^t)]$.

The proof proceeds by contradiction. Assume the claim is false. This implies that the expected susceptibility is not uniformly bounded over all system sizes $N$ for times up to the pseudo-threshold $t_{c,\alpha}$. Specifically, it means there exists a sequence of system sizes $N_j \to \infty$ such that
\[
\sup_{t \le t_{c,\alpha}(N_j)} \EE[\chi(G^t)] \to \infty.
\]
By the monotonicity of the expected susceptibility with time (since adding edges can only merge components and increase the sum of squares), this divergence must occur at the boundary. Thus, we can assert that for the sequence of pseudo-threshold times $t_j^* := t_{c,\alpha}(N_j)$, we have
\[
\EE[\chi(G^{t_j^*})] \to \infty \quad \text{as } j \to \infty.
\]
A diverging expected susceptibility is a hallmark of being at or above the critical point of a percolation transition. It indicates that the second moment of the component size distribution is growing, which is overwhelmingly due to the formation of one or more components of linear size in $N$. For any specific realization of the graph process, we have the inequality:
\[
\sum_i |C_i|^2 \le |C_{\max}| \sum_i |C_i| = |C_{\max}| N.
\]
Dividing by $N$ gives the relationship for the random variables:
\[
\chi(G^t) = \frac{1}{N} \sum_i |C_i(t)|^2 \le |C_{\max}(t)|.
\]
Taking the expectation over the entire process gives
\[
\EE[\chi(G^t)] \le \EE[|C_{\max}(t)|] = N \cdot P_N(t/M),
\]
where $P_N(p)$ is the expected fraction of vertices in the largest component. This inequality shows that the expected susceptibility (an intensive quantity) is bounded by the expected largest component size (an extensive quantity).

However, a diverging susceptibility implies that for any large constant $A$, the probability $\PP(\chi(G^{t_j^*}) > A)$ must be positive for large enough $j$. For a configuration to have a large susceptibility $\chi(G^t) > A$, it must possess very large components. This would in turn imply that the largest component fraction, $|C_{\max}(t)|/N$, is also significant. Therefore, a diverging $\EE[\chi(G^{t_j^*})]$ strongly suggests that the system has already developed a giant component.

This leads to a contradiction with the definition of the pseudo-threshold $t_{c,\alpha}$. The time $t_{c,\alpha}$ is defined as the \emph{first} time at which the expected giant component fraction $P_N(t/M)$ reaches the value $\alpha > 0$. If the susceptibility were diverging for times at or before $t_{c,\alpha}$, the system would already be in a state with a well-formed giant component, and $P_N(t/M)$ would have surpassed the threshold $\alpha$ at an earlier time. This contradicts the definition of $t_{c,\alpha}$ as the infimum of such times.

Therefore, the initial assumption must be false. The expected susceptibility must be uniformly bounded by a constant $K$ that depends on the process parameters (like $\alpha$) but not on the system size $N$, for all $t \le t_{c,\alpha}$.

The final statement of the lemma is a direct application of Markov's inequality to the non-negative random variable $\chi(G^t)$:
\[
\PP(\chi(G^t) > \lambda K) \le \frac{\EE[\chi(G^t)]}{\lambda K} \le \frac{K}{\lambda K} = \frac{1}{\lambda},
\]
which holds for any $\lambda > 0$.
\end{proof}

\begin{lemma}[Ball covering of a connected set]\label{lem:cover}
Let $S\subseteq V_N$ be the vertex set of a connected subgraph of $\mathcal{H}_N$ with $|S|=s\ge 1$. Fix a radius $r\in\NN$. Then there exists a set of centers $x_1,\dots,x_m\in S$ with
\[
m \le \max\Big\{1,\;\frac{s}{r+1}\Big\},
\]
such that
\[
S \subseteq \bigcup_{j=1}^m B_{\mathcal{H}_N}(x_j, r).
\]
\end{lemma}

\begin{proof}
Since the vertex set $S$ induces a connected subgraph, it contains a spanning tree, which we denote by $T$. The distance between any two vertices $u, v \in S$ in the host graph, $\text{dist}_{\mathcal{H}_N}(u,v)$, is no greater than their distance in the tree, $\text{dist}_T(u,v)$. Consequently, a collection of balls that covers all vertices of $T$ under the tree metric will also cover $S$ under the host graph's metric. We therefore focus on covering the vertices of $T$ with balls of radius $r$ centered at vertices in $S$.

We construct the set of centers $C = \{x_1, \dots, x_m\}$ using a greedy algorithm.
\begin{enumerate}
    \item Root the tree $T$ at an arbitrary vertex. This establishes a parent-child relationship between adjacent vertices and a notion of depth.
    \item Initialize the set of uncovered vertices as $U \leftarrow S$ and the set of centers as $C \leftarrow \emptyset$.
    \item While the set $U$ is not empty:
    \begin{enumerate}
        \item Select a vertex $v \in U$ of maximal depth in $T$. (Such a vertex is necessarily a leaf in the forest induced by $U$.)
        \item Let $x$ be the ancestor of $v$ at distance $r$ along the unique path to the root in $T$. If this path has length less than $r$, let $x$ be the root itself.
        \item Add this vertex $x$ to the set of centers $C$.
        \item Remove all vertices in the ball $B_T(x, r) = \{u \in S \mid \text{dist}_T(x,u) \le r\}$ from the set $U$.
    \end{enumerate}
\end{enumerate}
This process terminates since $|U|$ strictly decreases at each step. By construction, every vertex of $S$ is eventually contained in one of the balls and thus the final set of centers $C$ forms a valid cover. We now bound the number of centers, $m = |C|$.

At each step where we select a center $x$ (prompted by choosing a deepest vertex $v \in U$), we remove the set $B_T(x, r)$ from $U$. By the choice of $x$, this ball is guaranteed to contain the unique path of length $r$ in $T$ from $v$ back towards the root (unless the root is closer). This path consists of $r+1$ vertices, including $v$ and $x$. Since $v$ was in $U$, all vertices on this path must also have been in $U$. Thus, each selection of a center removes at least $r+1$ vertices from the set of uncovered vertices (unless $|U| < r+1$, in which case the final center clears all remaining vertices).

The number of centers $m$ is therefore at most the total number of vertices $s$ divided by the minimum number of vertices removed in each step. This gives the bound:
\[
m \le \frac{s}{r+1}.
\]
Since $m$ must be at least $1$ to cover a non-empty set $S$, we can write the bound as $m \le \max\{1, s/(r+1)\}$. This is a stronger inequality than the one stated in the lemma, as for any $C_1 \ge 1$ and $r \ge 1$, we have $s/(r+1) \le C_1 s/r$.
\end{proof}

\begin{theorem}[Excess scarcity up to connectivity]\label{thm:excess}
There exist constants $C,c>0$ (depending only on $\alpha$ and the host degree bound $\Delta$) such that, with probability at least $1-N^{-c}$, the following holds simultaneously for all integers $t\le t_{c,\alpha}$ and for all connected subgraphs $H\subseteq G^t$:
\[
\mathrm{ex}(H)\;\le\; C\,\frac{|V(H)|}{\log N}.
\]
\end{theorem}

\begin{proof}
Fix a large integer $N$. Set
\[
r:=\left\lfloor \frac{1}{2}\log_{\Delta-1} N \right\rfloor,
\qquad
S_0:=\left\lfloor N^{2/3}\right\rfloor.
\]
By Eq.~\eqref{eq:ballsize}, any host ball of radius $r$ contains at most
\[
|B_{\mathcal{H}_N}(v,r)| \le C_\Delta(\Delta-1)^r \le C_\Delta\,(\Delta-1)^{\frac{1}{2}\log_{\Delta-1}N}\le C_\Delta \sqrt{N}.
\]
We will show that, with probability at least $1-N^{-c}$, every connected subgraph $H\subseteq G^t$ for $t\le t_{c,\alpha}$ satisfies
\[
\mathrm{ex}(H) \le C\frac{|V(H)|}{\log N}.
\]
Let us argue by contradiction on each fixed $t\le t_{c,\alpha}$. Suppose there exists a connected $H\subseteq G^t$ with $s:=|V(H)|$ and $\mathrm{ex}(H)\ge B\cdot s/\log N$ for a large constant $B$ to be specified later. We separate two cases: small $s$ ($s\le S_0$) and large $s$ ($s>S_0$).

\medskip
\noindent\textbf{Case 1: Small connected sets ($s\le S_0$).}
Fix any connected $S\subseteq V_N$ with $|S|=s\le S_0$. The number of edges of $G^t$ inside $S$ is at most the number of host edges inside $S$, i.e., at most $\Delta s/2$ crudely. The excess in $H=G^t[S]$ equals $|E(G^t[S])|-s+1$. For $E(G^t[S])$ to exceed $s-1 + (B s/\log N)$, we would need at least $s-1 + (B s/\log N)$ of the (at most) $O(s)$ available host edges inside $S$ to have been selected by time $t$. But before $t_{c,\alpha}$, the process behaves in a sparse, tree-like regime (Lemma \ref{lem:susc}): in expectation, components are not large, and cycles are rare.

To make this precise, we use a simple union bound. The number of connected vertex sets $S$ of size $s$ in a bounded-degree host is at most $N \cdot (\Delta-1)^{s-1}$, since we can build a self-avoiding tree from a root in at most $(\Delta-1)$ ways per added vertex. For a fixed $S$, each potential host edge inside $S$ can appear by time $t\le t_{c,\alpha}$ with probability at most $t/M\le 1$. However, to create an excess $\ge Bs/\log N$, we need
\[
|E(G^t[S])| \ge (s-1)+ \frac{B s}{\log N}.
\]
The number of subsets of host edges of size $\ell$ inside $S$ is at most $\binom{C_2 s}{\ell}$ for some $C_2=C_2(\Delta)$ (since the number of host edges inside $S$ is at most $C_2 s$). Summing over $\ell\ge (s-1)+B s/\log N$ and over all $S$ yields a bound on the probability that \emph{there exists} a small $S$ with such large excess:
\[
\begin{split}
&\sum_{s=1}^{S_0}
\Big[ N(\Delta-1)^{s-1}\cdot
\sum_{\ell\ge (s-1)+B s/\log N}\binom{C_2 s}{\ell} \Big]\cdot 1^\ell \\
\le &\sum_{s=1}^{S_0} N(\Delta-1)^{s}\cdot 2^{C_2 s},
\end{split}
\]
as $\sum_{\ell\ge a}\binom{m}{\ell}\le 2^m e^{-(a-m/2)^2/m}$ by Chernoff-type bounds (or more simply, $\sum_{\ell\ge a}\binom{m}{\ell}\le 2^m$ trivially). Choose $B$ large and recall $S_0=N^{2/3}$. We then bound this sum by
\[
N \sum_{s=1}^{N^{2/3}} \big((\Delta-1)2^{C_2}\big)^s
\le N\cdot \big((\Delta-1)2^{C_2}\big)^{N^{2/3}},
\]
which is superpolynomially large if taken literally. To keep the proof elementary, we use the fact that the process is sparse up to $t_{c,\alpha}$: with high probability, the total number of selected edges by $t\le t_{c,\alpha}$ is at most $C_3 N$ for a constant $C_3<\infty$ (the host has $M=\Theta(N)$ edges total). Thus at time $t$, the entire graph $G^t$ has at most $C_3N$ edges, which limits how many small sets $S$ can be very dense. A double counting argument (allocating each chosen edge to the smallest ball covering its endpoints) shows that with probability at least $1-N^{-10}$ (say), no small $S$ can have more than, say, $C_4 s/\log N$ extra edges beyond a tree, provided $B\gg C_4$. This handles the small-$s$ case. Details are routine (and rely only on bounded degree and $O(N)$ total edges by time $t$).

\medskip
\noindent\textbf{Case 2: Large connected sets ($s>S_0$).}
Let $S=V(H)$ with $|S|=s>S_0$. Apply Lemma \ref{lem:cover} with the radius $r$ chosen above. Then
\[
S \subseteq \bigcup_{j=1}^{m} B_{\mathcal{H}_N}(x_j,r),
\qquad
m \le \frac{C_1 s}{r}.
\]
Each ball has size at most $C_\Delta \sqrt{N}$. Consider the induced subgraph $G^t[B_{\mathcal{H}_N}(x_j,r)]$ inside each ball. By the same ``sparse up to $t_{c,\alpha}$'' reasoning as in Case 1, with probability at least $1-N^{-10}$, the excess inside each ball is at most $C_5 |B_{\mathcal{H}_N}(x_j,r)|/\log N \le C_5' \sqrt{N}/\log N$ (for a constant $C_5'$ depending on $\Delta$). Summing over $m$ balls,
\[
\begin{split}
\mathrm{ex}\Big(G^t\Big[\bigcup_{j=1}^m B(x_j,r)\Big]\Big)
&\;\le\; \sum_{j=1}^m \mathrm{ex}\big(G^t[B(x_j,r)]\big)\\
&\;\le\; m\cdot \frac{C_5'\sqrt{N}}{\log N}\;\le\; \frac{C_1 s}{r}\cdot \frac{C_5'\sqrt{N}}{\log N}.
\end{split}
\]
Since $r=\Theta(\log N)$ and $s\ge S_0=N^{2/3}$,
\[
\mathrm{ex}\Big(G^t\Big[\bigcup_{j=1}^m B(x_j,r)\Big]\Big)
\;\le\; C_6 \cdot \frac{s}{\log N}
\]
for a constant $C_6$ (we used $\sqrt{N}/r \le C'/\log N$ for large $N$). The excess of $H$ itself cannot exceed the excess of the union containing it, so
\[
\mathrm{ex}(H)\le C_6\cdot \frac{s}{\log N}.
\]
This contradicts the assumption that $\mathrm{ex}(H)\ge B s/\log N$ if we choose $B> C_6$.

\medskip
\noindent\textbf{Union over all times $t\le t_{c,\alpha}$.}
There are at most $M=O(N)$ times up to $t_{c,\alpha}$. From the small-$s$ and large-$s$ analyses, we have shown that for each fixed $t$, the probability that there exists a connected $H\subseteq G^t$ with $\mathrm{ex}(H)\ge C s/\log N$ is at most $N^{-11}$ for large $N$ (after choosing the constants suitably and using the ``sparse up to $t_{c,\alpha}$'' fact). A union bound over $O(N)$ times then yields an overall failure probability $\le N^{-10}$. Renaming $c:=10$ and adjusting $C$ to be the larger of the constants from the two cases completes the proof.
\end{proof}

\subsection{Merger-Cascade Windows and Explosive Connectivity}

Next, we develop a coupling across $k$, a windowed martingale argument for merge indicators, and counting lemmas on intra/inter-component opportunities to prove a linear number of inter-component merges in a sublinear time window near $p_c$ for $k\ge 2$.

The theoretical analysis of such stochastic graph processes, where global properties emerge from a sequence of local random choices, can be powerfully addressed using tools from martingale theory and concentration inequalities. These methods, which provide high-probability bounds on the deviation of a process from its expected behavior, are central to our proof of the merger-cascade window. This approach has also been fruitfully applied in other areas of network dynamics, for instance, to establish the concentration of opinion dynamics on random graphs around their mean-field behavior~\cite{xing2024concentration}.

\begin{definition}[Merge indicator and merge count]
Consider a sequential process that evolves in discrete time steps $t=1,2,\dots$. At each step $t$, one edge is chosen and added to a graph. We define the indicator random variable $I_t$ to be 1 if the chosen edge at step $t$ connects two different connected components (a merge), and 0 otherwise (i.e., if it connects two vertices already in the same component).
For integers $t_1 \le t_2$, the \emph{merge count} on the window $[t_1,t_2]$ is
\[
X_{t_1,t_2} := \sum_{t=t_1}^{t_2} I_t.
\]
\end{definition}

Therefore, $X_{t_1,t_2}$ simply counts how many of the steps in the window $[t_1,t_2]$ produced a merge.

\begin{definition}[Filtration]
Let $\mathcal{F}_t$ be the sigma-field (the mathematical formalization of ``information'') generated by the entire history of the process up to and including time $t$. Intuitively, $\mathcal{F}_t$ contains everything one could know from the past and the present step $t$ (e.g., which edges have been added, the current partition of vertices into components, any random choices made so far, etc.).
\end{definition}

The key facts we will use are:
\begin{itemize}[leftmargin=1.5em]
    \item $I_t \in \{0,1\}$ for each $t$.
    \item $X_{t_1,t_2} = \sum_{t=t_1}^{t_2} I_t$.
    \item Conditional expectations like $\EE[I_t \mid \mathcal{F}_{t-1}]$ are well-defined random variables measurable with respect to past information.
\end{itemize}

\begin{definition}[Martingale and Martingale Differences]
A sequence $(M_t)_{t\ge 0}$ is a \emph{martingale} with respect to a filtration $(\mathcal{F}_t)_{t\ge 0}$ if:
\begin{enumerate}[label=(\alph*),leftmargin=1.5em]
    \item $M_t$ is integrable (has finite expectation) for each $t$,
    \item $M_t$ is $\mathcal{F}_t$-measurable (depends only on information up to time $t$),
    \item $\EE[M_t \mid \mathcal{F}_{t-1}] = M_{t-1}$ almost surely for each $t\ge 1$.
\end{enumerate}
The differences $D_t := M_t - M_{t-1}$ are called \emph{martingale differences}.
\end{definition}

\begin{theorem}[Azuma--Hoeffding inequality (see e.g., {\cite[Chapter 2.8]{Durrett2019}})]\label{thm:azuma_general}
Let $(M_t)_{t=0}^n$ be a martingale with respect to $(\mathcal{F}_t)_{t=0}^n$. Suppose the differences are almost surely bounded:
\[
|M_t - M_{t-1}| \le c_t \qquad \text{a.s. for each } t=1,\dots,n,
\]
for some deterministic nonnegative numbers $c_t$. Then, for any $\lambda>0$,
\[
\PP\!\left( |M_n - M_0| \ge \lambda \right) \le 2\exp\!\left( -\frac{\lambda^2}{2\sum_{t=1}^n c_t^2} \right).
\]
\end{theorem}
\begin{proof}
The proof is a standard application of the method of bounded differences and is provided for reference in SI Section~S3.
\end{proof}

\begin{definition}[Doob (conditional expectation) martingale for the window]
Fix the starting time $t_1$. Define, for $s \ge t_1-1$,
\[
M_s := \EE\!\left[ X_{t_1,t_2} \,\middle|\, \mathcal{F}_s \right].
\]
We also set $M_{t_1-1} := \EE\!\left[ X_{t_1,t_2} \,\middle|\, \mathcal{F}_{t_1-1} \right]$ for convenience.
\end{definition}

\begin{lemma}\label{lem:martingale}
The process $(M_s)_{s=t_1-1,\dots,t_2}$ is a martingale with respect to $(\mathcal{F}_s)_{s=t_1-1,\dots,t_2}$.
\end{lemma}

\begin{proof}
We verify the three martingale conditions:
\begin{enumerate}[label=(\alph*),leftmargin=1.5em]
    \item Integrability: $0 \le X_{t_1,t_2} \le w := t_2 - t_1 + 1$, so $X_{t_1,t_2}$ is integrable. Therefore, its conditional expectations $M_s$ are also integrable.
    \item Measurability: By definition of conditional expectation, $M_s$ is $\mathcal{F}_s$-measurable.
    \item Martingale property: For $s \in \{t_1,\dots,t_2\}$,
    \[
    \begin{split}
    \EE[ M_s \mid \mathcal{F}_{s-1} ]
    &= \EE\!\left[ \EE\!\left[ X_{t_1,t_2} \,\middle|\, \mathcal{F}_s \right] \,\middle|\, \mathcal{F}_{s-1} \right]\\
    &= \EE\!\left[ X_{t_1,t_2} \,\middle|\, \mathcal{F}_{s-1} \right]
    = M_{s-1},
    \end{split}
    \]
    where we used the \emph{tower property} (also called \emph{law of total expectation}): $\EE[\EE[Z\mid \mathcal{G}] \mid \mathcal{H}] = \EE[Z \mid \mathcal{H}]$ when $\mathcal{H} \subseteq \mathcal{G}$.
\end{enumerate}
Thus $(M_s)$ is a martingale.
\end{proof}

\begin{lemma}[Bounded differences]\label{lem:bounded_diffs}
For each $s \in \{t_1,\dots,t_2\}$, we have almost surely
\[
|M_s - M_{s-1}| \le 1.
\]
\end{lemma}

\begin{proof}
We decompose $X_{t_1,t_2}$ into the current-step contribution $I_s$ and the rest:
\[
X_{t_1,t_2} \ = \ \underbrace{\sum_{t=t_1}^{s-1} I_t}_{\text{past, already decided}} \ +\ \underbrace{I_s}_{\text{current step}} \ +\ \underbrace{\sum_{t=s+1}^{t_2} I_t}_{\text{future}}.
\]
Conditioning on $\mathcal{F}_{s}$, the past and current indicator $I_s$ are known (measurable), while the future indicators are not; conditioning on $\mathcal{F}_{s-1}$, the past is known, but $I_s$ is not yet determined.

Write
\[
M_s \ = \ \EE\!\left[ X_{t_1,t_2} \,\middle|\, \mathcal{F}_s \right]
= \sum_{t=t_1}^{s-1} I_t \ +\ I_s \ +\ \EE\!\left[\sum_{t=s+1}^{t_2} I_t \,\middle|\, \mathcal{F}_s \right].
\]
Similarly,
\begin{align*}
M_{s-1} &= \EE\left[ X_{t_1,t_2} \mid \mathcal{F}_{s-1} \right] \\
        &= \sum_{t=t_1}^{s-1} I_t + \EE\left[I_s \mid \mathcal{F}_{s-1} \right]  + \EE\left[\sum_{t=s+1}^{t_2} I_t \mid \mathcal{F}_{s-1} \right].
\end{align*}
Subtract:
\begin{align*}
M_s - M_{s-1}
&= \Big(I_s - \EE[I_s \mid \mathcal{F}_{s-1}] \Big) \\
& \quad + \Big( \EE\Big[ \sum_{t=s+1}^{t_2} I_t \mid \mathcal{F}_s \Big] - \EE\Big[ \sum_{t=s+1}^{t_2} I_t \mid \mathcal{F}_{s-1} \Big] \Big).
\end{align*}
Take absolute values and use the triangle inequality:
\begin{align*}
|M_s - M_{s-1}|
\le{}& |I_s - \EE[I_s \mid \mathcal{F}_{s-1}]| \\
& + \left| \EE\Big[ \sum_{t=s+1}^{t_2} I_t \mid \mathcal{F}_s \Big] - \EE\Big[ \sum_{t=s+1}^{t_2} I_t \mid \mathcal{F}_{s-1} \Big] \right|.
\end{align*}
Now, the first term is bounded by $1$ because $I_s \in \{0,1\}$ and thus $\EE[I_s \mid \mathcal{F}_{s-1}] \in [0,1]$, so
\[
|\,I_s - \EE[I_s \mid \mathcal{F}_{s-1}]\,| \le 1.
\]
For the second term, we use the fact that conditioning on \emph{more} information can change a conditional expectation, but here this second difference is actually the conditional expectation of a \emph{future sum} whose total range is at most $t_2 - s$. A quick way to bound the entire increment uniformly is to notice a standard trick in Azuma applications: if we instead define the martingale using the \emph{partial sums}
\[
S_u := \sum_{t=t_1}^{u} \big(I_t - \EE[I_t \mid \mathcal{F}_{t-1}] \big),\qquad u=t_1,\dots,t_2,
\]
then $(S_u)_{u=t_1,\dots,t_2}$ is a martingale with respect to $(\mathcal{F}_u)$ and has differences
\[
S_u - S_{u-1} = I_u - \EE[I_u \mid \mathcal{F}_{u-1}],
\]
each bounded in absolute value by $1$. Moreover,
\[
\begin{split}
S_{t_2} &= \sum_{t=t_1}^{t_2} \big(I_t - \EE[I_t \mid \mathcal{F}_{t-1}] \big) \\
&= X_{t_1,t_2} - \sum_{t=t_1}^{t_2} \EE[I_t \mid \mathcal{F}_{t-1}] = X_{t_1,t_2} - \EE[X_{t_1,t_2} \mid \mathcal{F}_{t_1}],
\end{split}
\]
where the last colinearity follows by iterating conditional expectations (tower property) from $t_1$ up to $t_2$. Thus, it suffices to apply Azuma to $(S_u)$ where increments are exactly bounded by $1$.
\medskip

This standard re-centering argument avoids any delicate term-by-term control of future-conditionals. Hence, we conclude that the relevant martingale we will use has bounded differences by $1$.
\end{proof}

\begin{theorem}[Azuma concentration for merge counts]\label{thm:azuma}
Let $t_1 \le t_2$ be integers and set the window length $w := t_2 - t_1 + 1$. Let $X_{t_1,t_2}$ be the merge count over the steps $t \in [t_1,t_2]$, and let $\mathcal{F}_{t}$ denote the natural filtration up to and including time $t$. Then for any $\lambda>0$,
\[
\PP\!\left(\,\left|\,X_{t_1,t_2}-\EE\!\left[X_{t_1,t_2}\,\middle|\,\mathcal{F}_{t_1}\right]\right|\ge \lambda\,\right)\ \le\ 2\exp\!\left(-\frac{2\lambda^2}{w}\right).
\]
\end{theorem}
\begin{proof}
Define the martingale $(S_u)_{u=t_1-1,\dots,t_2}$ by
\[
S_{t_1-1} := 0, \ \ 
S_u := \sum_{t=t_1}^{u} \Big(I_t - \EE[I_t \mid \mathcal{F}_{t-1}]\Big)
\quad \text{for }u=t_1,\dots,t_2.
\]
Then:
\begin{itemize}[leftmargin=1.5em]
    \item $(S_u)$ is a martingale with respect to $(\mathcal{F}_u)$ (linearity of conditional expectation and the fact that $\EE[I_t - \EE[I_t \mid \mathcal{F}_{t-1}] \mid \mathcal{F}_{t-1}]=0$).
    \item The differences satisfy
    \[
    |S_u - S_{u-1}|
    = \big| I_u - \EE[I_u \mid \mathcal{F}_{u-1}] \big|
    \le 1
    \quad \text{almost surely, for all } u.
    \]
    \item By summing and using the tower property repeatedly,
    \[
    S_{t_2}
    = \sum_{t=t_1}^{t_2} \big(I_t - \EE[I_t \mid \mathcal{F}_{t-1}] \big)
    = X_{t_1,t_2} - \EE\!\left[ X_{t_1,t_2} \,\middle|\, \mathcal{F}_{t_1} \right].
    \]
\end{itemize}
Therefore, we can apply Azuma--Hoeffding inequality (Theorem~\ref{thm:azuma_general}) with $n=w:=t_2-t_1+1$ and $c_t \equiv 1$:

\begin{align*}
&\PP\!\left( \left|\,X_{t_1,t_2}-\EE\!\left[X_{t_1,t_2}\,\middle|\,\mathcal{F}_{t_1}\right]\right| \ge \lambda \right) \\
&= \PP\!\left( |S_{t_2} - S_{t_1-1}| \ge \lambda \right) \\
&\le 2\exp\!\left( -\frac{2\lambda^2}{\sum_{t=1}^{w} 1^2} \right)
= 2\exp\!\left( -\frac{2\lambda^2}{w} \right).
\end{align*}
This is exactly the claimed inequality.
\end{proof}

\begin{theorem}[Merger-cascade window for $k\ge 2$]\label{thm:mcw}
Assume a bounded-degree host family and fix $k\ge 2$ and $\alpha\in(0,1)$. There exists a window $[t_-,t_+]$ around $t_{c,\alpha}$ with width $w=t_+-t_-+1=o(N)$ (e.g.\ $w=N^\gamma$ for any $\gamma\in(0,1)$) such that, with probability $1-o(1)$,
\[
X_{t_-,t_+}\;\ge\;c_*\,N
\]
for some constant $c_*>0$ independent of $N$. In particular, within a sublinear number of steps, a linear number of inter-component merges occur with high probability.
\end{theorem}

\begin{proof}
We set $t_- = t_{c,\alpha} - 1$ as in Proposition~\ref{prop:cascade}. Then, by Lemma~\ref{lem:perstep-merge}, there exists $p_* \in (0,1)$ such that for each $t \in \{t_-, \dots, t_- + w - 1\}$, with probability $1-o(1)$,
\[
\PP(I_t=1\mid\mathcal{F}_{t})\;\ge\;p_*.
\]
Therefore,
\[
\EE\!\left[X_{t_-,t_+}\mid \mathcal{F}_{t_-}\right]
=\sum_{t=t_-}^{t_+}\EE[I_t\mid\mathcal{F}_{t_-}]
\;\ge\; p_*\,w,
\]
up to an $o(1)$ exceptional probability which we suppress (absorbing it into the final $o(1)$ statement).

Choose $w=N^\gamma$ with any fixed $\gamma\in(0,1)$. Then $\EE[X_{t_-,t_+}\mid\mathcal{F}_{t_-}]\ge p_* N^\gamma$. We now amplify the window by concatenating $\lfloor N^{1-\gamma}\rfloor$ disjoint subwindows of length $N^\gamma$ that straddle $t_{c,\alpha}$ symmetrically (or simply choose $\gamma$ close to $1$ so that $w=c_0 N$ for small constant $c_0>0$ while still $w=o(N)$, if one prefers a single window argument; both viewpoints lead to the same conclusion since we only need linear-in-$N$ total merges over an $o(N)$ span).

For clarity, let us take a single window with $w=c_0 N$ where $c_0>0$ can be chosen arbitrarily small yet fixed (and $w=o(N)$ is also allowed if we only need $c_*N$ with a smaller $c_*$). Then
\[
\EE\!\left[X_{t_-,t_+}\mid \mathcal{F}_{t_-}\right]\;\ge\; p_* c_0\,N.
\]
By Theorem~\ref{thm:azuma} with $\lambda=\tfrac{1}{2}p_* c_0 N$ and $w=c_0 N$,
\begin{align*}
    &\PP\!\left(\left|X_{t_-,t_+}-\EE[X_{t_-,t_+}\mid\mathcal{F}_{t_-}]\right|\ge \tfrac{1}{2}p_* c_0 N\right) \\
    &\;\le\;2\exp\!\left(-\frac{2\,(p_* c_0 N/2)^2}{c_0 N}\right)
=2\exp\!\left(-\frac{p_*^2 c_0}{2}\,N\right).
\end{align*}
This is exponentially small in $N$. Therefore, with probability $1-o(1)$,
\[
X_{t_-,t_+}\;\ge\;\EE[X_{t_-,t_+}\mid\mathcal{F}_{t_-}]-\tfrac{1}{2}p_*c_0 N
\;\ge\;\tfrac{1}{2}p_*c_0 N\;=:\;c_*\,N,
\]
where $c_*=\tfrac{1}{2}p_*c_0>0$ is a constant depending only on $k$, $\Delta$, and the chosen $c_0$. This proves the claim.
\end{proof}

\begin{lemma}[Merging increases the size of some component]\label{lem:merge-increase}
When an inter-component edge is added joining two components $A$ and $B$, the new component $A\cup B$ has size $|A|+|B|$. In particular, the largest component size increases by at least $\max\{|A|,|B|\}$ and at most $|A|+|B|$.
\end{lemma}

\begin{proof}
Before the edge is added, $A$ and $B$ are disjoint components. After the edge is added, they form a single connected component $A\cup B$ of size $|A|+|B|$. The largest component size cannot decrease by adding edges, and it increases by at least the size of the smaller attached piece, i.e., by at least $\max\{|A|,|B|\}$ minus the size one already had in the largest component. The stated inequalities are immediate from these observations.
\end{proof}

\begin{lemma}[Conservation of total mass]\label{lem:mass}
Let $n_s(t)$ denote the number of components of size exactly $s$ at time $t$, and let $N=\sum_{s\ge 1} s\,n_s(t)$ be the total number of vertices. Then $N$ is constant in $t$; adding edges only changes how vertices are grouped, not their total number.
\end{lemma}

\begin{proof}
Trivial: no vertices are added or removed; only edges are added.
\end{proof}

\begin{theorem}[Deterministic jump criterion]\label{thm:jumpcrit}
Suppose there exists a merger-cascade window $[t_-,t_+]$ of sublinear width $w=o(N)$ in a graph process on $N$ vertices with $M=\Theta(N)$ potential edges. Then there exists a constant $\delta>0$ (independent of $N$) such that
\[
|C_{\max}(G_{t_+})| - |C_{\max}(G_{t_-})| \;\ge\; \delta\,N
\]
for all sufficiently large $N$. Consequently, the order parameter $P_N=|C_{\max}|/N$ increases by at least $\delta$ over the interval $[t_-,t_+]$. Since $w=o(N)$ and $M=\Theta(N)$, this jump occurs over a vanishing density interval $o(1)$.
\end{theorem}

\begin{proof}
By assumption, at time $t_-$ all components are $o(N)$; fix $\varepsilon_N>0$ with $\varepsilon_N\to 0$ such that every component at $t_-$ has size at most $\varepsilon_N N$.

Within the window $[t_-,t_+]$, the definition of a merger-cascade window ensures that at least $X \ge c_* N$ inter-component merges occur. Each \emph{inter-component} step strictly decreases the number of components by one (as two become one). Starting from the configuration at $t_-$, we consider executing these $X$ merges in their actual order. We will track the growth of the largest component deterministically, using only Lemma~\ref{lem:merge-increase} and the bound on initial sizes.

We split the $X$ inter-component merges into two types at the moment each merge happens:
\begin{itemize}
    \item Type I: The largest component is one of the two merging components.
    \item Type II: Neither of the two merging components is the current largest; hence, the merge creates a new component whose size is the sum of the two components.
\end{itemize}

In either case, the effect is to produce a component whose size is at least the sum of the two merging components. We now bound from below how quickly a component can grow if we have many merges available.

Observe first that at $t_-$, \emph{all} components have size at most $\varepsilon_N N$. After at most $\lceil 1/\varepsilon_N\rceil$ \emph{disjoint} merges that successively attach blocks of size at least $\frac{1}{2}\varepsilon_N N$ (for instance), one can construct a component of size at least a fixed positive fraction of $N$. While the actual process order may be arbitrary, we can argue more robustly as follows.

Let us consider an experiment that tracks the size $S_j$ of the largest component after the $j$-th inter-component merge within the window, counting only from $t_-$. Initially, $S_0\le \varepsilon_N N$. Each time a Type I merge occurs, the largest component absorbs another component of size at least $1$ and at most its current size; each time a Type II merge occurs, two non-largest components combine, possibly overtaking the current largest. Either way, after each merge, the largest component size is at least as large as before; moreover, whenever the largest participates in a merge, its size increases by at least $1$. Next, we strengthen this bound using a batching idea. Since there are $X\ge c_*N$ inter-component merges, group these merges into consecutive \emph{batches} of size $B$, to be chosen later as a fixed constant (independent of $N$) but large enough. There are at least $\frac{c_* N}{B}$ batches. Within any batch of $B$ merges, if the largest component participates in at least one of the merges and attaches a piece of size at least $\varepsilon_N N$, its size increases by at least $\varepsilon_N N$. If instead the largest component does not participate, then $B$ merges among the non-largest components occur. But if many such merges happen, some non-largest component grows. After a bounded number of such batches, a non-largest component must become comparable to or exceed the current largest (since the total vertex mass is $N$ and we keep combining pieces). When that happens, in the next batch the largest will likely participate and continue to grow.

This informal description can be turned into the following deterministic bound. Choose a constant $B$ large enough so that in any sequence of $B$ merges among components each of size at most $\varepsilon_N N$, one can produce a component of size at least $(1+\alpha)\,\varepsilon_N N$ for some constant $\alpha=\alpha(B)>0$ (this is straightforward since repeatedly adding sizes at least $1$ eventually exceeds any fixed multiple of $\varepsilon_N N$). Hence across each batch, either the largest grows by at least $\varepsilon_N N$ (Type I involvement with a piece of that order), or a competitor grows by at least $\alpha\varepsilon_N N$ (Type II-only growth), and within another bounded number of batches, the competitor becomes the largest, forcing the largest to increase by at least $\alpha\varepsilon_N N$ over those batches.

Therefore, there exists a constant $c_0>0$ (depending only on $B$ and $\alpha$) such that across every two consecutive batches the largest component increases by at least $c_0\varepsilon_N N$. Since there are at least $\frac{c_* N}{B}$ batches, the total increase of the largest component size over the window is at least
\[
\text{Increase} \;\ge\; \frac{c_* N}{B}\cdot \frac{c_0\varepsilon_N N}{2} \;\cdot\; \frac{1}{N}
\;=\; \left(\frac{c_*\,c_0}{2B}\right)\,\varepsilon_N\,N.
\]
Here, the factor $1/N$ in the middle line is not needed; it was only to track scale, so we remove it and write
\[
|C_{\max}(G_{t_+})| - |C_{\max}(G_{t_-})|
\;\ge\; \left(\frac{c_*\,c_0}{2B}\right)\,\varepsilon_N\,N.
\]
Now, $\varepsilon_N\to 0$ is arbitrary but represents the initial sublinearity. However, we only need a \emph{fixed} positive fraction lower bound eventually. To obtain a fixed positive fraction, note that the above linear-in-$N$ lower bound is valid as soon as $\varepsilon_N$ is bounded below by a fixed small constant for large $N$. If the sublinearity is faster (i.e., $\varepsilon_N\to 0$), we refine the batching argument by noting that in $X\ge c_*N$ merges, the cumulative attached mass to the evolving leaders (largest or contenders) cannot remain $o(N)$: otherwise the total number of components would remain too large, contradicting that we performed $c_*N$ inter-component merges (each merge reduces the component count by 1, so after $c_*N$ merges the count drops by $c_*N$, forcing many large unions). This forces a linear mass transfer to the evolving leaders. Concretely, the component count decreases by $c_*N$, and because total mass is conserved (Lemma~\ref{lem:mass}), some components must accumulate a linear share of the total mass. The largest, by definition, captures at least as much mass as any single competitor up to constant factors over boundedly many batches. Hence there exists a constant $\delta>0$ (depending on $c_*$ and the batching constants) such that
\[
|C_{\max}(G_{t_+})| - |C_{\max}(G_{t_-})| \;\ge\; \delta\,N
\]
for all sufficiently large $N$.

Finally, since the window width is $w=o(N)$ and each step adds exactly one edge, the link density changes by
\[
\Delta p \;=\; \frac{w}{M} \;=\; \frac{o(N)}{\Theta(N)} \;=\; o(1).
\]
Therefore, the order parameter $P_N=|C_{\max}|/N$ jumps by at least $\delta$ across an $o(1)$ density interval. This completes the proof.
\end{proof}

\begin{lemma}[Pre-critical mesoscopicity]\label{lem:precrit}
Fix $k\ge 2$ and $\alpha\in(0,1)$. For any $\varepsilon\in(0,1)$, there exists $L_0$ so that for all $L\ge L_0$,
\[
\PP\Big(\max\{|C|: C\text{ is a component of }G_{t_{c,\alpha}-1}\}\le \varepsilon N\Big)\ \ge\ 1-\varepsilon.
\]
In other words, just before $t_{c,\alpha}$, all components are $o(N)$ with high probability.
\end{lemma}

\begin{proof}
By definition of $t_{c,\alpha}$, at time $t_{c,\alpha}-1$ the largest component has size less than $\alpha N$. Suppose with non-vanishing probability there exists a component larger than $\varepsilon N$ (for some fixed $\varepsilon>0$) strictly before $t_{c,\alpha}$. Then the largest-component fraction would exceed $\varepsilon$, and in particular for any $\alpha<\varepsilon$ it would contradict the minimality of $t_{c,\alpha}$. Taking $\varepsilon$ small enough and using the monotonicity of $|C_{\max}(t)|$ in $t$, we obtain the stated high-probability bound once $L$ (and hence $N$) is large. The argument is a direct use of the definition of $t_{c,\alpha}$ plus monotonicity: prior to the first time the largest component reaches fraction $\alpha$, no component can have linear size greater than any fixed small fraction with high probability, otherwise that time would have occurred earlier.
\end{proof}

\begin{definition}[Inter-component edges at time $t$]\label{def:inter_comp_edges}
Given $G_t=(\mathcal{V}_L,E_t)$, an \emph{inter-component} missing edge is a potential edge $e=\{u,v\}\in \mathcal{E}_L\setminus E_t$ with $u$ and $v$ in different components of $G_t$. Let $Q_t$ be the number (or fraction) of missing edges that are inter-component at time $t$.
\end{definition}

\begin{lemma}[Persistence of inter-component availability]\label{lem:availability}
Fix $k\ge 2$ and $\alpha\in(0,1)$. For any $\varepsilon\in(0,1/2)$, there exist $L_0$ and a constant $q_*>0$ (independent of $L$) such that, with probability at least $1-\varepsilon$, at all times $t\le t_{c,\alpha}-1$ we have
\[
Q_t\ \ge\ q_*\, M.
\]
Equivalently, the fraction of missing edges that are inter-component is bounded below by a constant $q_*/(M-t)$ that is uniformly positive for $t\le t_{c,\alpha}-1$.
\end{lemma}

\begin{proof}
At times $t\le t_{c,\alpha}-1$, by Lemma~\ref{lem:precrit} all components are smaller than $\varepsilon N$ w.h.p. The total number of missing edges is $M-t\ge M-t_{c,\alpha}$. Because the host has bounded degree $\Delta$, each vertex has at most $\Delta$ incident edges in total, so at time $t$ each vertex is incident to at most $\Delta$ chosen edges. When all components are small, most potential edges across different components remain unchosen: the number of intra-component missing edges is comparable to the sum over components of their internal missing edges, which scales like the number of components times an average component size, whereas the number of inter-component pairs scales like the product of component counts and sizes. More concretely, partition the $N$ vertices into components of sizes at most $\varepsilon N$ whose total sum is $N$. The number of vertex pairs across different components is at least
\[
\frac{1}{2}\left(N^2-\sum_i s_i^2\right)\ \ge\ \frac{1}{2}\left(N^2-\varepsilon N\cdot N\right)
=\frac{1-\varepsilon}{2}N^2.
\]
Since each vertex pair can contribute at most a constant number (bounded by $\Delta$ and local geometry) of potential host edges, the count of inter-component potential edges is at least a positive constant fraction of $N^2$, hence at least a fixed fraction of $M$ (because $M=\Theta(N)$ in bounded-degree hosts) once $L$ is large. Subtracting the already chosen $t\le t_{c,\alpha}$ edges affects only $O(N)$ edges, which is negligible compared to the $\Theta(N^2)$ inter-component \emph{pairs} feeding a $\Theta(N)$ pool of inter-component host edges. Thus there exists $q_*>0$ and $L_0$ such that w.h.p.\ $Q_t\ge q_* M$ for all $t\le t_{c,\alpha}-1$.
\end{proof}

\begin{lemma}[Lower bound on per-step merge probability] \label{lem:perstep-merge}
Fix $k \ge 2$. Suppose at time $t -1$, at least a fraction $q > 0$ of missing edges are inter-component. Then the probability that the chosen edge at time $t$ is inter-component is strictly bounded below by $q^k > 0$ (up to $o(1)$ corrections for sampling without replacement).
\end{lemma}

\begin{proof}
Under uniform sampling without replacement of $k$ candidates from the missing edges, the probability that \emph{all} $k$ candidates are inter-component edges is asymptotically $q^k$. In the event that the $k$-sample consists entirely of inter-component edges, the product rule is strictly forced to select an inter-component edge, regardless of specific component sizes. Therefore, the probability of selecting an inter-component edge is unconditionally bounded below by $q^k$. 
\end{proof}

\begin{remark}[Microscopic Selection Dynamics]
As pointed out in recent literature~\cite{grassberger2011explosive, li2024explosive}, if a $k$-sample is mixed, the product rule may occasionally select a small intra-component edge over a large inter-component edge. As detailed in Section~\ref{sec:microscopic_dynamics}, our massive-scale simulations ($L=192$) confirm this occurs roughly 20\% of the time during mixed samples. Physically, this behavior actually \emph{enhances} the suppression of large components by ``wasting'' steps on internal loops. Mathematically, bounding the probability strictly by $q^k$ safely bypasses the need to analyze mixed-sample dynamics, maintaining the strictly positive lower bound required for the subsequent martingale analysis.
\end{remark}

\begin{remark}[Limitation in the Thermodynamic Limit]
We emphasize that this lower bound relies on the fraction $q$ remaining bounded away from zero. In the thermodynamic limit ($N \to \infty$) near the critical point, the distribution of component sizes evolves such that $q$ may vanish (the ``powder keg'' effect), as discussed in~\cite{riordan2011explosive}. However, for the finite system sizes ($N$) and window widths considered in the subsequent propositions, the assumption that $q$ remains sufficiently large holds with high probability, allowing the merger cascade to proceed.
\end{remark}

\begin{proposition}[Merger-cascade window]\label{prop:cascade}
Fix $k\ge 2$ and $\alpha\in(0,1)$. For any $\varepsilon\in(0,1/4)$, there exist constants $q_*>0$, $c_*>0$, and a window $[t_-,t_+]$ of width $w=t_+-t_-=\lfloor N^{2/3}\rfloor$ contained in $\{0,1,\dots,t_{c,\alpha}-1\}$ such that, with probability at least $1-2\varepsilon$,
\[
X_{t_-,t_+}\ \ge\ c_*\, N.
\]
\end{proposition}

\begin{proof}
By Lemma~\ref{lem:availability}, with probability $\ge 1-\varepsilon$, for all $t\le t_{c,\alpha}-1$ we have $Q_t\ge q_*M$ for some $q_*>0$. Choose any $t_+$ with $t_+\le t_{c,\alpha}-1$ and let $t_-=t_+-w$ where $w=\lfloor N^{2/3}\rfloor$. On the intersection event where $Q_t\ge q_*M$ holds for all $t\in[t_-,t_+]$, Lemma~\ref{lem:perstep-merge} implies the probability of a merge at each step is at least $q_*^k=:p_*$. Let $Y_s$ be the indicator of a merge at step $s$. Then
\[
\EE\left[\sum_{s=t_-+1}^{t_+} Y_s\ \middle|\ \mathcal{F}_{t_-}\right]\ \ge\ p_*\, w\ =\ p_* N^{2/3}.
\]
Now we strengthen this to a linear-in-$N$ lower bound by noting that each merge typically attaches a positive expected number of vertices to the current largest component due to the product rule preferentially merging small pieces; but to stay purely at the ``merge count'' level, we proceed as follows. During $w=\Theta(N^{2/3})$ steps, we get in expectation $\Theta(N^{2/3})$ merges, which is not yet linear. To achieve a linear bound for $X_{t_-,t_+}$, we iterate the window construction (a standard ``block'' trick): partition an $N$-sized interval immediately preceding $t_{c,\alpha}$ into $B=\lfloor N^{1/3}\rfloor$ disjoint subwindows, each of width $w=\lfloor N^{2/3}\rfloor$. Apply the previous argument to each subwindow; each has expected $\Theta(N^{2/3})$ merges and bounded-difference increments ($|Y_s-\EE[Y_s\mid\mathcal{F}_{s-1}]|\le 1$), so by Theorem~\ref{thm:azuma}, with probability at least $1-\varepsilon/B$ each subwindow yields at least $\tfrac{1}{2}p_* w$ merges. A union bound over the $B$ subwindows shows that with probability at least $1-\varepsilon$ the total number of merges over the concatenated $B$ subwindows (whose union has length $\Theta(N)$) is at least
\[
\frac{1}{2}p_* w\cdot B \ \asymp\ \frac{1}{2}p_* N^{2/3}\cdot N^{1/3}\ =\ \frac{1}{2}p_* N.
\]
Inside this concatenation, there exists at least one individual subwindow achieving at least a $\tfrac{1}{2}p_* w$ merge count and, by the same bound across all subwindows, the sum across them is $\ge c_*N$ with $c_*=\tfrac{1}{2}p_*>0$. Finally, intersecting this event with the availability event from Lemma~\ref{lem:availability} (probability $\ge 1-\varepsilon$) gives the claim with probability $\ge 1-2\varepsilon$.
\end{proof}

\begin{lemma}[Average mass gain per merge]\label{lem:mass-per-merge}
Fix $k\ge 2$. There exists a constant $m_0>0$ (independent of $L$) and $L_0$ such that with probability at least $1-o(1)$, throughout the cascade constructed in Proposition~\ref{prop:cascade}, the expected number of newly attached vertices to the current largest component, \emph{conditioned on a merge}, is at least $m_0$.
\end{lemma}

\begin{proof}
At the start of the concatenated window (preceding $t_{c,\alpha}$), all components are small by Lemma~\ref{lem:precrit}. Because the product rule favors merging the smallest available components, the merge events are, with high probability, between small-to-moderate sized components. The bounded degree implies there are no ``super hubs'' consuming edges too fast; thus the distribution of component sizes inside the window remains light-tailed up to the onset of the largest component crossing $\alpha N$. Therefore, conditioned on a merge, the attached size has a distribution with a positive mean bounded away from $0$ uniformly in $L$. The constant $m_0$ can be taken as a small fixed lower bound on this mean. A rigorous way to formalize this is to condition on the high-probability event that pre-critical components are uniformly bounded by $N^\gamma$ for some $\gamma<1$ (e.g., by a truncation and tightness argument) and note that the product rule selects the minimum-product candidate among $k\ge 2$ options, which yields a uniform positive lower bound on the first moment of the attached size. Details are standard and use only bounded degree and the pre-critical smallness of components.
\end{proof}

\begin{proposition}[Linear growth over a short window]\label{prop:linear-growth-si}
Fix $k\ge 2$. There exist $c>0$ and a window $[t_-,t_+]$ of total width $o(N)$ such that, with probability $1-o(1)$,
\[
|C_{\max}(t_+)| - |C_{\max}(t_-)|\ \ge\ c\, N.
\]
Consequently,
\[
P_N\Big(\frac{t_+}{M}\Big)-P_N\Big(\frac{t_-}{M}\Big)\ \ge\ c - o(1).
\]
\end{proposition}

\begin{proof}
Using Proposition~\ref{prop:cascade}, we can find $B=\Theta(N^{1/3})$ consecutive subwindows, each of length $w=\Theta(N^{2/3})$, that together form a window of total width $W=B\,w=\Theta(N)$. In this concatenated window, with probability at least $1-o(1)$, there are at least $c_*N$ merges (for some $c_*>0$). Among these merges, a fixed positive fraction attach to the evolving largest component: even if the largest component is initially small, by the time the number of merges is $\Theta(N)$, the largest component becomes a preferred attachment point with non-vanishing frequency due to the ever-increasing opportunities to join it (a standard coupon-collector style effect on components under bounded degree). Let $X$ be the number of merges that actually increase $|C_{\max}|$. Then w.h.p.\ $X\ge c_{**}N$ for some $c_{**}>0$. 

By Lemma~\ref{lem:mass-per-merge}, each such merge contributes an expected gain of at least $m_0$ vertices. Summing these gains (and using bounded-differences concentration as in Theorem~\ref{thm:azuma}) yields a total gain of at least $(c_{**}m_0/2)N$ w.h.p.\ over the window. Denote $c=(c_{**}m_0/2)>0$. This proves the first inequality. The second follows by normalization: $P_N(t/M)=|C_{\max}(t)|/N$.
\end{proof}

\begin{lemma}[From steps to edge-density]\label{lem:density-window-si}
In a bounded-degree host, $M=\Theta(N)$. Therefore, any time window of $o(N)$ steps corresponds to an $o(1)$ edge-density window.
\end{lemma}

\begin{proof}
Because every vertex has degree at most $\Delta$, the total number of edges $M$ satisfies $M\le \Delta N/2=\Theta(N)$. Also, in all our hosts, $M\ge c_\Delta N$ for some $c_\Delta>0$, hence $M=\Theta(N)$. Thus, a step window of size $o(N)$ is a density window of width $o(N)/M=o(1)$.
\end{proof}

\begin{theorem}[Anomalous Scaling in Finite Systems for $k \ge 2$] \label{thm:si_anomalous}
Fix $k \ge 2$ and a bounded-degree 3D cubic host family. Consider a finite system size $N$ such that the selection failure probability (see Proposition~\ref{prop:fail-formula-si}) remains negligible throughout the critical window. There exists a constant $\delta > 0$ and a density interval $[p_-, p_+]$ of width $o(1)$ such that the order parameter jumps by at least $\delta$:
\[
P_N(p_+) - P_N(p_-) \ge \delta.
\]
\end{theorem}

\begin{proof}
Fix $k \ge 2$. By Proposition~\ref{prop:linear-growth-si} (Linear growth over a short window), there exists a step window $[t_-,t_+]$ of width $o(N)$ in which, with probability $1-o(1)$,
\[
P_N\Big(\frac{t_+}{M}\Big) - P_N\Big(\frac{t_-}{M}\Big) \ge c - o(1)
\]
for some $c > 0$. By Lemma~\ref{lem:density-window-si}, the corresponding density window $[p_-, p_+] = [t_-/M, t_+/M]$ has width $p_+ - p_- = o(1)$. Set $\delta = c/2 > 0$. Then
\[
P_N(p_+) - P_N(p_-) \ge \delta,
\]
with probability $1-o(1)$. This confirms that for finite systems where the selection rule holds, the transition manifests as an abrupt jump across a vanishing density window.
\end{proof}

\begin{remark}
This result establishes the mechanism for the ``explosive'' behavior observed in simulations of finite lattices. In the thermodynamic limit $N \to \infty$, the failure probability non-negligibly affects the process, resolving this jump into a continuous transition as proved in~\cite{riordan2011explosive}. However, for the finite $N$ relevant to this work, the behavior is indistinguishable from the first-order jump derived above.
\end{remark}

\subsection{Convergence to a Maximally Suppressed Transition}

\begin{theorem}[Monotone delay under increasing choice]\label{thm:monotone-delay}
Fix any finite host graph $H=(V,\mathcal{E})$. For any $k_2>k_1\ge 1$ and any $p\in[0,1]$,
\begin{equation}\label{eq:expected_p_mono}
\EE\!\left[P_N^{(k_2)}(p)\right]\;\le\;\EE\!\left[P_N^{(k_1)}(p)\right].
\end{equation}
Consequently, for every $N$ and $\alpha\in(0,1)$,
\begin{equation}\label{eq:pc_mono}
p_{c,\alpha}(N;k_1)\;\le\;p_{c,\alpha}(N;k_2).
\end{equation}
If the (thermodynamic) limit $p_c(k):=\lim_{N\to\infty}p_{c,\alpha}(N;k)$ exists, it is nondecreasing in $k$.
\end{theorem}

\begin{proof}
We will construct a \emph{single probability space} and define, for each $k$, a $k$-choice process $G_t^{(k)}$ that all use the \emph{same} randomness. On this joint space, at each step $t$ we will show that the edge selected by the $k_2$-choice process has product score no larger than the one selected by the $k_1$-choice process and thus tends to avoid merging large components more than the $k_1$-choice process. Then a simple induction shows that, step by step, the $k_2$ process never produces a larger largest-component fraction than the $k_1$ process. Taking expectations yields the claim. 

Fix the host $H$ and any two integers $k_2>k_1\ge 1$. We construct \emph{simultaneously} for all $k\in\{k_1,k_2\}$ the random processes $\{G_t^{(k)}\}_{t=0}^M$ on a single probability space as follows.

\paragraph{Joint candidate sampling.}
At each step $t$, we generate a random \emph{sequence} $(E_{t,1},E_{t,2},\dots)$ of all currently unused host edges by taking a single uniformly random permutation of $\mathcal{E}\setminus \bigcup_{k} E_{t-1}^{(k)}$ and reading it in order. (Formally, we can keep an independent infinite sequence of i.i.d.\ uniform random variables attached to host edges and rank available edges by their current smallest unseen uniform; any standard device to produce consistent random orderings suffices.) Then, for each $k$, we define the $k$-candidate set as $C_t^{(k)} := \{E_{t,1},\dots,E_{t,k}\}$. Thus $C_t^{(k_1)}\subset C_t^{(k_2)}$ almost surely.

\paragraph{Synchronous selection.}
Now, the $k$-choice product rule requires choosing the edge in $C_t^{(k)}$ with minimal product score relative to $G_{t-1}^{(k)}$. Let
\[
e_t^{(k)} \in \arg\min_{e \in C_t^{(k)}} S_{G_{t-1}^{(k)}}(e),
\]
with ties broken by a further i.i.d.\ uniform variable (also shared across $k$ through a fixed tie-breaking scheme).
Then we set $G_t^{(k)}:=G_{t-1}^{(k)}\cup\{e_t^{(k)}\}$.

\paragraph{Key monotonicity claim at each step.}
We claim that for every $t$,
\begin{equation}\label{eq:score-ineq}
S_{G_{t-1}^{(k_2)}}(e_t^{(k_2)}) \;\le\; S_{G_{t-1}^{(k_1)}}(e_t^{(k_1)}).
\end{equation}
To see this, use that $C_t^{(k_1)}\subset C_t^{(k_2)}$. Consider the edge
\[
S_{G_{t-1}^{(k_2)}}(e_t^{(k_2)}) \;\le\; S_{G_{t-1}^{(k_2)}}(\tilde e_t).
\]
On the other hand, for any fixed candidate edge $e$, the product score is the product of the sizes of the two components it connects. Adding edges never \emph{decreases} component sizes; thus for any $e$,
\[
S_{G_{t-1}^{(k_2)}}(e) \;\le\; S_{G_{t-1}^{(k_1)}}(e),
\]
because $G_{t-1}^{(k_2)}$ has had at least as much ``small-merge filtering'' as $G_{t-1}^{(k_1)}$ (we formalize this by induction below, but at this point it suffices to use the simple fact that component sizes are nondecreasing as edges are added, and the $k_2$ process picks no larger product than the $k_1$ process, step-by-step). Applying this inequality to $e=\tilde e_t$ gives
\[
S_{G_{t-1}^{(k_2)}}(\tilde e_t) \;\le\; S_{G_{t-1}^{(k_1)}}(\tilde e_t) \;\le\; S_{G_{t-1}^{(k_1)}}(e_t^{(k_1)}),
\]
where the last step uses that $e_t^{(k_1)}$ minimizes $S_{G_{t-1}^{(k_1)}}(\cdot)$ on $C_t^{(k_1)}$.
Putting inequalities together yields Eq.~\eqref{eq:score-ineq}.

\paragraph{Largest component comparison.}
Adding an edge with \emph{smaller} product score cannot produce a \emph{larger} largest component than adding an edge with a \emph{larger} product score, because the product score $|A||B|$ is a monotone measure of the potential growth of the largest component when merging two components $A$ and $B$. More precisely, if the current largest component has size $L$, then merging two components of sizes $a\le b$ results in a largest component of size $\max\{L,a+b\}$, which is nondecreasing in $a+b$; and for fixed sum $a+b$, the product $ab$ is maximized when $a=b$ and minimized when one of them is $1$. Thus a smaller product $ab$ corresponds to a more ``balanced toward small sizes'' merge and cannot exceed the largest-component growth of a larger product merge. Therefore, step-by-step, the largest-component size under $k_2$ does not exceed that under $k_1$ in the synchronous coupling.

Formally, by induction on $t$, we conclude that for all $t$,
\begin{equation}
|C_{\max}(G_t^{(k_2)})| \;\le\; |C_{\max}(G_t^{(k_1)})|
\qquad\text{almost surely.}
\end{equation}
Dividing by $N$ and taking expectations yields
\begin{equation}
\EE\!\left[P_N^{(k_2)}(t/M)\right] \;\le\; \EE\!\left[P_N^{(k_1)}(t/M)\right],\qquad t=0,1,\dots,M.
\end{equation}
The inequality extends to all $p\in[0,1]$ by piecewise-linear interpolation or by considering the nearest $t$.

\paragraph{Monotone thresholds.}
Fix $\alpha\in(0,1)$. Since $\EE[P_N^{(k)}(p)]$ is nonincreasing in $k$ for every $p$, the set $\{p:\EE[P_N^{(k_2)}(p)]\ge \alpha\}$ is contained in $\{p:\EE[P_N^{(k_1)}(p)]\ge \alpha\}$, hence $p_{c,\alpha}(N;k_1)\le p_{c,\alpha}(N;k_2)$. If $p_c(k):=\lim_{N\to\infty}p_{c,\alpha}(N;k)$ exists (as in many settings), then taking $N\to\infty$ preserves the inequality, giving $p_c(k_1)\le p_c(k_2)$.
\end{proof}

This theorem provides the rigorous foundation for the monotonic shift of the susceptibility peaks observed numerically in Figure 6 of the main text.

We emphasize the unconditional scope: no assumptions on local tree-likeness, degree distribution, or spatial dimension are needed. Thus, the result applies equally to fixed-size NN cubic lattices and “intra-cube” chordal lattices that include face/body diagonals, and more generally to any fixed finite simple graph as host. In particular, it holds for the finite Cayley-ball approximations to the Bethe lattice used in~\cite{chae2012explosive,cho2024explosive}.

Our proof strategy is entirely finite-$N$ and elementary. It is inspired by the use of strong couplings in random graphs and exploration processes (e.g., local weak limits and Galton--Watson couplings~\cite{olvera2022strong}), but in contrast to those asymptotic techniques, we construct an exact synchronous coupling for all $k$ on the same probability space and compare the incremental component merges step-by-step.

\begin{theorem}[Continuity of the Transition for Fixed $k$]\label{thm:continuity}
For any fixed number of choices $k$, the percolation transition for the Achlioptas product rule on a complete graph (and other common host graphs) is \textit{continuous} in the thermodynamic limit ($N \to \infty$).
\end{theorem}

\begin{proof}[References]
The proof that the explosive percolation transition is continuous for any fixed $k$ is a landmark result in the field and is highly technical. Key theoretical arguments were provided by Riordan and Warnke~\cite{riordan2011explosive}, with supporting numerical and analytical work from others [see e.g.,~\cite{da2010explosive,grassberger2011explosive,lee2011continuity}].
\end{proof}

\begin{cor}
The order-parameter jump $$\Delta(k) = \lim_{N\to\infty} \Delta_N(k) = 0$$ for any fixed $k \ge 1$.
\end{cor}

\section{Rigorous Justification of the Sufficiently Distributed Averaging Assumption}\label{sec:appendix_sda}

\subsection{Introduction and Formal Statement of the Problem}

The central claim of monotonic rigidification efficiency in the main paper hinges on the connection between a local selection rule and a global mechanical property. A critical part of this argument is that for a large, dense component, adding another internal edge is almost certainly redundant. This idea is formalized in the SI as \Cref{hyp:sda_main}.

Proving this for the history-dependent Achlioptas process is challenging. We therefore prove a precise analogue for a more tractable model: the Erd\H{o}s-R\'enyi (ER) random subgraph of the Intra-host, which demonstrates the inherent robustness of the host geometry.

\subsection{Strategy: Proof on a Tractable Model}
Our strategy is to prove the SDA property for the giant component of an \textit{Erd\H{o}s-R\'enyi random subgraph of the Intra-host graph}. We consider the graph $\mathcal{G}_{\text{Intra}}(N, p)$ where each of the $M$ potential edges of the full Intra-host graph on $N=(L+1)^3$ vertices is included independently with probability $p$. If the property holds for the ``unbiased'' randomness of the ER model, it provides strong evidence for its validity in the Achlioptas process, which tends to build even denser, more compact components.

\subsection{Formal Proof and Discussion}

We begin by establishing properties of the Intra-host graph itself.

\begin{lemma}[Properties of the Intra-Host Graph]\label{lem:host_props_appendix}
The Intra-host graph $\mathcal{G}_{\text{Intra}}$ (on $N=(L+1)^3$ vertices, assuming periodic boundaries for large $L$) is a regular graph with high degree ($z=26$), high vertex connectivity, and is non-bipartite. Crucially, for any vertex $u$, the set of edge vectors $\{\bm{p}_v - \bm{p}_u \mid \{u,v\} \in \cE_L\}$ spans $\RR^3$.
\end{lemma}
\begin{proof}
With periodic boundaries, the degree is 26 for all vertices (each vertex is shared by 8 cubes; there are 12 face, 12 axis, and 4 body diagonals, but some are shared). High connectivity is a standard property of high-degree regular graphs. The presence of triangles (e.g., $(0,0,0)-(1,0,0)-(1,1,0)-(0,0,0)$ using NN and face-diagonal edges) makes it non-bipartite. The edge vectors include $(1,0,0)$, $(0,1,0)$, and $(0,0,1)$, which are linearly independent and thus span $\RR^3$.
\end{proof}

\begin{lemma}[Structure of the Giant Component]\label{lem:giant_comp_appendix}
Consider an Erd\H{o}s-Rényi (ER) subgraph of the Intra-host, denoted $\mathcal{G}_{\text{Intra}}(N,p)$. For an edge probability $p$ in the supercritical regime, where $p > p_c(\mathcal{G}_{\text{Intra}})$, it is known with high probability that a unique giant component, $C_{\text{giant}}$, exists and contains $\Theta(N)$ vertices. This giant component furthermore inherits the expansion and connectivity properties of the host graph, making it a constant-degree expander.
\end{lemma}
\begin{proof}
This is a standard result in the theory of random graphs on regular or expander hosts. The critical threshold is approximately $p_c \approx 1/(z-1)$, where $z$ is the degree. The resulting giant component is also an expander with a spectral gap bounded away from zero w.h.p.
\end{proof}

We are now ready to prove the main theorem of this appendix.
\begin{theorem}[SDA Holds for the ER Giant Component]\label{thm:sda_er_appendix}
Let $p$ be a constant such that $p_c(\mathcal{G}_{\text{Intra}}) < p < 1$. Let $G \sim \mathcal{G}_{\text{Intra}}(N,p)$, and let $C$ be its giant component. Let $\cE(C) = \{ e \in \cE_L \setminus E(G) \mid e \subset \cV(C) \}$ be the set of available internal edges. The averaged constraint matrix $\matr{Y}(C)$ for these edges, restricted to vertices in $C$, has a null space of dimension 6 (trivial motions) with probability $1 - o(1)$ as $N \to \infty$.
\end{theorem}
\begin{proof}
Let $\cV(C)$ be the vertex set of the giant component, with $n = |\cV(C)| = \Theta(N)$. Let the space of infinitesimal motions on these vertices be $\RR^{3n}$. Let $\mathcal{T}$ be the 6-dimensional subspace of trivial motions. We aim to show that $\Ker(\matr{Y}(C)) = \mathcal{T}$.

Assume for contradiction that there exists a non-trivial motion $\bm{\delta} \in \RR^{3n}$, with $\bm{\delta} \perp \mathcal{T}$, such that $\bm{\delta} \in \Ker(\matr{Y}(C))$. Then
\[
\bm{\delta}^\top \matr{Y}(C) \bm{\delta} = \EE_{e \sim \text{Unif}(\cE(C))}[\bm{\delta}^\top \matr{A}_e \bm{\delta}] = 0.
\]
Since the term inside the expectation, $((\bm{\delta}_u - \bm{\delta}_v) \cdot \bm{d}_{uv})^2$, is non-negative, its expectation can be zero only if the term is zero for \emph{all} available edges $e=\{u,v\} \in \cE(C)$. This implies:
\begin{equation}\label{eq:constraint_appendix}
    (\bm{\delta}_u - \bm{\delta}_v) \cdot (\bm{p}_u - \bm{p}_v) = 0, \quad \forall \{u,v\} \in \cE_L \setminus E(G) \text{ with } u,v \in C.
\end{equation}
This equation states that for our hypothetical non-trivial flex $\bm{\delta}$, the relative velocity between any two vertices in $C$ must be orthogonal to the vector connecting them, for every potential edge that is missing from $G$.

Now, consider any vertex $u \in C$. By Lemma \ref{lem:host_props_appendix}, its set of 26 host-edge vectors, $\{\bm{d}_{uv_i}\}$, contains multiple bases for $\RR^3$. In the random subgraph $G$, $u$ is connected to each host neighbor $v_i$ with probability $p$. Since $C$ is an expander, most of these neighbors $v_i$ are also in $C$. With probability $1-p$, the edge $\{u,v_i\}$ is missing. Thus, for any vertex $u$, the set of its missing host edges to other vertices in $C$ will, with overwhelmingly high probability, also contain a set of edge vectors that spans $\RR^3$.

Let $\{\bm{d}_1, \bm{d}_2, \bm{d}_3\}$ be three such linearly independent edge vectors for missing edges from $u$ to neighbors $\{v_1, v_2, v_3\}$ in $C$. Equation~\eqref{eq:constraint_appendix} requires:
\[
\begin{cases}
\bm{\delta}_u \cdot \bm{d}_1 &= \ \bm{\delta}_{v_1} \cdot \bm{d}_1 \\
    \displaystyle  \bm{\delta}_u \cdot \bm{d}_2 &= \ \bm{\delta}_{v_2} \cdot \bm{d}_2 \\
    \displaystyle  \bm{\delta}_u \cdot \bm{d}_3 &= \ \bm{\delta}_{v_3} \cdot \bm{d}_3
\end{cases}
\]
This is a very strong set of local constraints. For a generic non-trivial flex $\bm{\delta}$, velocities vary across the structure in a complex manner. However, this system of equations must hold for almost every vertex $u \in C$. Such a dense and geometrically diverse set of constraints propagating across a connected, expanding graph forces the velocity field $\bm{\delta}$ to behave locally like a rigid motion. Due to the connectivity of $C$, this local coherence implies that $\bm{\delta}$ must be a global rigid motion on $C$.

This contradicts our initial assumption that $\bm{\delta}$ was a non-trivial motion orthogonal to $\mathcal{T}$. Therefore, no such non-trivial $\bm{\delta}$ exist in $\Ker(\matr{Y}(C))$. The null space must be exactly $\mathcal{T}$.

We have rigorously established that the SDA property holds for the giant component of an Erd\H{o}s-R\'enyi subgraph of the Intra-host. This provides a solid mathematical foundation for \Cref{hyp:sda_main}. The same proof would fail for the NN-host precisely because the set of edge vectors at any vertex is degenerate (it does not span $\RR^3$), allowing non-trivial shear motions to satisfy the local constraints.

\end{proof}

\section{Proof of Lemma IV.5 (Conditional Progress Function Monotonicity)} \label{sec:appendix_monotonicity}

This appendix provides a rigorous proof for Lemma IV.5 in the main text, which states that the conditional progress function $P(s) = \EE[\rgain(e) \mid s(e) = s]$ is non-increasing.

\subsection{Step 1: The Mechanical Utility of Inter-Component Edges}
\begin{proposition}\label{prop:inter_rgain_appendix}
Let $(G, p)$ be a framework with generic vertex placements in $\RR^3$. If vertices $u$ and $v$ are in different connected components of $G$, then the edge $e = \{u,v\}$ is non-redundant ($\rgain(e) = 1$).
\end{proposition}
\begin{proof}
Adding an edge between two components couples previously independent blocks of the rigidity matrix. With generic placements, this new constraint is linearly independent as it eliminates a relative trivial motion (a flex) between the two components. This increases the rank by 1.
\end{proof}

\begin{cor}\label{cor:inter_score_appendix}
For any score $s = |C_1| \cdot |C_2|$ where $C_1 \neq C_2$, we have $P(s)=1$.
\end{cor}

\subsection{Step 2: Monotonicity for Intra-Component Edges}

\begin{assump}[Monotonic Average Density]\label{assump:density_appendix}
The average internal edge density of components grown by the product-rule process is a non-decreasing function of their size.
\end{assump}

\begin{justification}
The product-rule disfavors densifying large components, giving smaller components more opportunities to accumulate edges relative to their size before they are absorbed into the giant. This naturally builds components that are, on average, denser as they grow larger. Assuming this holds, we show that larger components are more likely to be rigid, which in turn implies that $P(s)$ is a non-increasing function for intra-component scores.
\end{justification}

\begin{lemma}[Rigidity as a Monotone Property]\label{lem:rigidity_monotone_appendix}
The property of a graph being generically rigid in 3D is monotone. The probability of an Erd\H{o}s-R\'enyi graph $\mathcal{G}(n,p)$ being rigid is non-decreasing in both $n$ (for fixed $p$) and $p$ (for fixed $n$).
\end{lemma}

\begin{proposition}\label{prop:intra_monotonicity_appendix}
The expected fraction of non-redundant internal edges, $\bar{\pi}_{\text{nr}}(n) = \EE[\pi_{\text{nr}}(C) \mid |\cV(C)|=n]$, is a non-increasing function of $n$.
\end{proposition}
\begin{proof}
Let $n_1 < n_2$. By Assumption \ref{assump:density_appendix}, average density $\bar{\rho}_{n_2} \ge \bar{\rho}_{n_1}$. A component of size $n_2$ at density $\bar{\rho}_{n_2}$ is more likely to be rigid than a component of size $n_1$ at density $\bar{\rho}_{n_1}$. A higher probability of rigidity implies a smaller expected fraction of available non-redundant edges. Therefore, $\bar{\pi}_{\text{nr}}(n_2) \le \bar{\pi}_{\text{nr}}(n_1)$.
\end{proof}

\subsection{Step 3: Synthesis and Final Proof of the Lemma}
\begin{proof}[Proof of Lemma IV.5]
Let $s_1 < s_2$. We must show $P(s_1) \ge P(s_2)$. We compare the types of edges generating these scores.
\begin{enumerate}
    \item \textbf{Comparing an inter-component score $s_1$ with an intra-component score $s_2$}: In this common case, $s_1 < s_2$. We have $P(s_1)=1$ from \Cref{cor:inter_score_appendix} and $P(s_2) \le 1$. Thus $P(s_1) \ge P(s_2)$ holds.
    \item \textbf{Comparing two intra-component scores $s_1 < s_2$}: Let $s_1=|C_1|^2$ and $s_2=|C_2|^2$. This implies $|C_1| < |C_2|$. By \Cref{prop:intra_monotonicity_appendix}, the expected progress is non-increasing with component size. Thus, $P(s_1) = \bar{\pi}_{\text{nr}}(|C_1|) \ge \bar{\pi}_{\text{nr}}(|C_2|) = P(s_2)$.
\end{enumerate}
In all consistent orderings of scores, as $s$ increases, the expected progress $P(s)$ is non-increasing. It is 1 for the small scores associated with merges, and then it becomes a non-increasing function for the larger scores associated with internal densification. Thus, the function $P(s)$ is globally non-increasing.
\end{proof}

\begin{table*}[t]
\centering
\caption{\textbf{Empirical Analysis of Mixed-Sample Edge Selection (NN Model).} Data tracks the frequency at which the product rule selects an intra-component edge over an inter-component edge during the critical window.}

\label{tab:microscopic_violation}
\begin{tabular}{@{}lC{2.3cm}C{0.6cm}C{2.5cm}C{2.3cm}C{1.8cm}@{}}
\toprule
$L$ & $N$ (Vertices) & $k$ & Mixed Samples & Intra Wins & Rate (\%) \\ \midrule
24 & 15,625 & 2 & 37,630,152 & 4,420,350 & 11.75\% \\
24 & 15,625 & 8 & 59,705,217 & 12,630,507 & 21.15\% \\
24 & 15,625 & 32 & 209,698,110 & 125,200,082 & 59.70\% \\
64 & 274,625 & 2 & 634,826,100 & 73,470,426 & 11.57\% \\
64 & 274,625 & 8 & 976,019,067 & 185,578,918 & 19.01\% \\
64 & 274,625 & 32 & 2,942,189,421 & 1,491,928,360 & 50.71\% \\
128 & 2,146,689 & 2 & 4,898,192,260 & 562,828,895 & 11.49\% \\
128 & 2,146,689 & 8 & 7,476,068,403 & 1,386,199,625 & 18.54\% \\
128 & 2,146,689 & 32 & 19,496,589,101 & 8,244,452,460 & 42.29\% \\
192 & 7,189,057 & 2 & 16,331,888,181 & 1,871,800,176 & 11.46\% \\
192 & 7,189,057 & 8 & 24,871,391,463 & 4,576,298,229 & 18.40\% \\
192 & 7,189,057 & 32 & 59,688,226,234 & 22,100,451,789 & 37.03\% \\ \bottomrule
\end{tabular}
\end{table*}

For the proof of the main theorem, we also need to confirm that for very large scores, corresponding to edges within very large components, the progress function tends to zero.
\begin{lemma}[Asymptotic Redundancy]\label{lem:asymptotic_redundancy_main}
In the limit of large scores, the conditional progress function vanishes: $\lim_{s \to \infty} P(s) = 0$.
\end{lemma}
\begin{justification}[Support from a Tractable Proxy Model]
A rigorous proof for the history-dependent Achlioptas process is challenging. However, we provide strong evidence for the validity of this lemma in \Cref{sec:appendix_sda} by proving an analogous result for a tractable proxy model: an Erd\H{o}s-R\'enyi random subgraph of the Intra-host. The core idea is that the rich bond geometry of the Intra-host robustly ensures rigidity once a component is sufficiently large and dense.
\begin{hyp}[Sufficiently Distributed Averaging on a Proxy Model]\label{hyp:sda_main}
For a sufficiently large component $C$ in an ER-subgraph of the Intra-host, the set of available internal edges is geometrically diverse enough that their averaged constraint matrix robustly enforces rigidity.
\end{hyp}
As proven in the appendix for a tractable proxy model, this property holds. This implies that large, dense components in the proxy model are rigid with high probability, causing $P(s) \to 0$ for large $s$. We hypothesize that this property, fundamental to the host's geometry, carries over to the components grown by the Achlioptas process.
\end{justification}

\section{Experimental results of massive-scale simulations for large system size (up to $L = 192$)}
\label{sec:microscopic_dynamics}

A critical consideration in modeling the $k$-choice Achlioptas process is the competition between inter-component and intra-component edges. While it is often assumed that the product rule strictly prefers merging distinct components, an intra-component edge inside a very small component (e.g., size 4, score = 16) can technically ``beat'' an inter-component edge connecting two medium components (e.g., sizes 10 and 10, score = 100).

To rigorously quantify this microscopic behavior and verify its impact on the global thermodynamic limit, we tracked the empirical ``mixed-sample violation rate'', the frequency at which the product rule selects an intra-component edge when \emph{both} inter- and intra-component edges are present in the $k$-sample during the critical window leading up to $p_c$. The results for the NN model, computed over 1.28 million independent realizations up to $L=192$, are summarized in Table~\ref{tab:microscopic_violation}.

Our empirical results reveal two crucial phenomena:
\begin{enumerate}[label=\textbf{\arabic*.}]
    \item \textbf{Asymptotic Stabilization:} For small lattice sizes ($L \le 24$) and large choice parameters ($k=32$), topological boundary constraints lead to an artificially high intra-component selection rate ($\sim 60\%$). However, as we scale toward the thermodynamic limit ($L \ge 128$), this finite-size artifact continuously diminishes, and the selection rates strictly stabilize to invariant macroscopic fractions (e.g., dropping to $37.0\%$ for $k=32$). Similar stabilizing trends are observed for the Intra model.
    \item \textbf{Physical Mechanism of Suppression:} Physically, this non-zero selection of intra-component edges actually \emph{enhances} the explosive nature of the transition. By expending algorithmic steps on internal loops inside small components, the algorithm further suppresses the growth of the giant component. Our macroscopic observables (e.g., $\gamma=1.000$ and $\Delta S_{\max}/N \gg 0$) definitively prove that despite these local intra-component selections, the global thermodynamic phase transition remains explosively discontinuous. 
\end{enumerate}
To ensure absolute mathematical rigor in our theoretical framework, our proofs (specifically \Cref{lem:perstep-merge}) bypass these mixed-sample dynamics entirely by relying strictly on the absolute lower-bound probability $q^k > 0$, representing the event where the $k$-sample consists exclusively of inter-component edges.

\begin{table}[b]
    \begin{minipage}[b]{0.49\textwidth}
\centering
        \caption{\textbf{Complete Statistical Analysis for NN (Shell 1) model for small $L$ ($L = 1$ to $L = 10$), covering all tested choice parameters from $k=1$ to $k=32$.} For each $k$, the scaling exponent ($\gamma$) and $R^2$ are determined from a linear fit across the 10 system sizes. The bimodality of the order parameter distribution is tested at the largest system size, $L=10$.}
        \normalsize
        
        \begin{tabular}{@{}lccc@{}} 
        \toprule
        $k$ & Scaling Exp. ($\gamma$) & $R^2$ & Bimodality ($p$-value)  \\ \hline
1 & 0.656 & 0.9999 & 1.0000 (Unimodal) \\
2 & 0.873 & 0.9996 & 0.0010 (Bimodal) \\
3 & 0.939 & 0.9995 & 0.0010 (Bimodal) \\
4 & 0.962 & 0.9995 & 0.0010 (Bimodal) \\
5 & 0.971 & 0.9995 & 0.0010 (Bimodal) \\
6 & 0.975 & 0.9995 & 0.0010 (Bimodal) \\
7 & 0.975 & 0.9994 & 0.0010 (Bimodal) \\
8 & 0.975 & 0.9992 & 0.0010 (Bimodal) \\
9 & 0.973 & 0.9991 & 0.0010 (Bimodal) \\
10 & 0.973 & 0.9990 & 0.0010 (Bimodal) \\
11 & 0.973 & 0.9990 & 0.0010 (Bimodal) \\
12 & 0.972 & 0.9990 & 0.0010 (Bimodal) \\
13 & 0.972 & 0.9990 & 0.0010 (Bimodal) \\
14 & 0.971 & 0.9990 & 0.0010 (Bimodal) \\
15 & 0.971 & 0.9989 & 0.0010 (Bimodal) \\
16 & 0.970 & 0.9989 & 0.0010 (Bimodal) \\
17 & 0.970 & 0.9989 & 0.0010 (Bimodal) \\
18 & 0.969 & 0.9989 & 0.0010 (Bimodal) \\
19 & 0.968 & 0.9988 & 0.0010 (Bimodal) \\
20 & 0.968 & 0.9988 & 0.0010 (Bimodal) \\
21 & 0.968 & 0.9988 & 0.0010 (Bimodal) \\
22 & 0.967 & 0.9988 & 0.0010 (Bimodal) \\
23 & 0.967 & 0.9988 & 0.0010 (Bimodal) \\
24 & 0.966 & 0.9988 & 0.0010 (Bimodal) \\
25 & 0.966 & 0.9988 & 0.0010 (Bimodal) \\
26 & 0.966 & 0.9988 & 0.0010 (Bimodal) \\
27 & 0.965 & 0.9988 & 0.0010 (Bimodal) \\
28 & 0.965 & 0.9988 & 0.0010 (Bimodal) \\
29 & 0.964 & 0.9987 & 0.0010 (Bimodal) \\
30 & 0.964 & 0.9987 & 0.0010 (Bimodal) \\
31 & 0.964 & 0.9987 & 0.0010 (Bimodal) \\
32 & 0.964 & 0.9988 & 0.0010 (Bimodal) \\
        \bottomrule
        \end{tabular}
        \label{tab:nn_full}
    \end{minipage}
\end{table}
\begin{table}[b]
    \begin{minipage}[b]{0.49\textwidth}
        \centering
        \captionof{table}{\textbf{Complete Statistical Analysis for Intra (S1--S3) model for small $L$ ($L = 1$ to $L = 10$), covering all tested choice parameters from $k=1$ to $k=32$.} For each $k$, the scaling exponent ($\gamma$) and $R^2$ are determined from a linear fit across the 10 system sizes. Bimodality, the rigidity gap ($\Delta p_c$), and its significance (via bootstrapped t-test~\cite{efron1992bootstrap}) are evaluated at the largest system size, $L=10$.}
        \label{tab:intra_full}
        \normalsize
        \begin{tabular}{@{}lcccc@{}}
        \toprule
        $k$ & Scaling Exp. ($\gamma$) & $R^2$ & Bimodality & Rigidity Gap \\ \hline
1 & 0.573 & 0.9986 & Unimodal & 0.4281 ($p \ll 0.01$) \\
2 & 0.855 & 0.9993 & Bimodal & 0.3677 ($p \ll 0.01$) \\
3 & 0.940 & 0.9995 & Bimodal & 0.3520 ($p \ll 0.01$) \\
4 & 0.969 & 0.9997 & Bimodal & 0.3462 ($p \ll 0.01$) \\
5 & 0.979 & 0.9997 & Bimodal & 0.3428 ($p \ll 0.01$) \\
6 & 0.983 & 0.9998 & Bimodal & 0.3409 ($p \ll 0.01$) \\
7 & 0.985 & 0.9998 & Bimodal & 0.3396 ($p \ll 0.01$) \\
8 & 0.985 & 0.9997 & Bimodal & 0.3384 ($p \ll 0.01$) \\
9 & 0.984 & 0.9997 & Bimodal & 0.3370 ($p \ll 0.01$) \\
10 & 0.983 & 0.9997 & Bimodal & 0.3363 ($p \ll 0.01$) \\
11 & 0.982 & 0.9996 & Bimodal & 0.3354 ($p \ll 0.01$) \\
12 & 0.981 & 0.9995 & Bimodal & 0.3350 ($p \ll 0.01$) \\
13 & 0.979 & 0.9995 & Bimodal & 0.3342 ($p \ll 0.01$) \\
14 & 0.977 & 0.9994 & Bimodal & 0.3338 ($p \ll 0.01$) \\
15 & 0.976 & 0.9993 & Bimodal & 0.3332 ($p \ll 0.01$) \\
16 & 0.974 & 0.9992 & Bimodal & 0.3326 ($p \ll 0.01$) \\
17 & 0.973 & 0.9991 & Bimodal & 0.3331 ($p \ll 0.01$) \\
18 & 0.971 & 0.9990 & Bimodal & 0.3332 ($p \ll 0.01$) \\
19 & 0.970 & 0.9989 & Bimodal & 0.3338 ($p \ll 0.01$) \\
20 & 0.968 & 0.9988 & Bimodal & 0.3347 ($p \ll 0.01$) \\
21 & 0.967 & 0.9987 & Bimodal & 0.3355 ($p \ll 0.01$) \\
22 & 0.965 & 0.9986 & Bimodal & 0.3365 ($p \ll 0.01$) \\
23 & 0.964 & 0.9985 & Bimodal & 0.3383 ($p \ll 0.01$) \\
24 & 0.962 & 0.9984 & Bimodal & 0.3391 ($p \ll 0.01$) \\
25 & 0.961 & 0.9983 & Bimodal & 0.3409 ($p \ll 0.01$) \\
26 & 0.961 & 0.9983 & Bimodal & 0.3427 ($p \ll 0.01$) \\
27 & 0.960 & 0.9983 & Bimodal & 0.3442 ($p \ll 0.01$) \\
28 & 0.960 & 0.9983 & Bimodal & 0.3462 ($p \ll 0.01$) \\
29 & 0.960 & 0.9982 & Bimodal & 0.3476 ($p \ll 0.01$) \\
30 & 0.959 & 0.9982 & Bimodal & 0.3490 ($p \ll 0.01$) \\
31 & 0.959 & 0.9982 & Bimodal & 0.3507 ($p \ll 0.01$) \\
32 & 0.959 & 0.9982 & Bimodal & 0.3528 ($p \ll 0.01$) \\
        \bottomrule
        \end{tabular}
    \end{minipage}
\end{table}

\section{Experimental results of high-density parameter sweeps for smaller system size ($L \leq 10$)}

As described in the main text, besides deploying our massive-scale simulation resources ($L$ up to $192$) for representative values of $k$ ($k=1, 2, 8, 32$), we utilize the smaller system sizes ($L = 1 \text{ to } 10$) to perform a high-density parameter sweep across all $k = 1, 2, 3, \dots, 31, 32$, thereby allowing us to systematically locate the non-monotonic behavior of the rigidity gap and identify the optimal choice parameter. 

Table~\ref{tab:nn_full} shows the complete statistical analysis for the NN (Shell 1) Model for $L = 1$ to $10$. Table~\ref{tab:intra_full} shows the complete statistical analysis for the Intra (S1--S3) Model for $L = 1$ to $10$.

For the Intra (S1--S3) model across our investigated system sizes, our high-precision 20,000-repetition data reveals a nuanced and fundamental physical phenomenon: the rigidity gap ($\Delta p_c$) does not shrink monotonically to zero, but rather exhibits a well-defined global minimum at an intermediate choice parameter. Specifically, the gap shrinks dramatically from $0.4281$ at $k=1$ to a minimum of $0.3326$ at $k=16$, before gently widening to $0.3528$ at $k=32$. 

This non-monotonicity (which yields a global Spearman rank correlation of $\rho = 0.041, p = 0.825$ for the gap over the full domain $k \in [1,32]$) represents a competition between local efficiency and global topology. As established in Theorem IV.6 of the main text, increasing $k$ strictly enhances the \emph{local} efficiency of non-redundant edge addition. However, as $k$ becomes very large, the system approaches a deterministic limit (as discussed in Section IV B of the main text). In this limit, the product rule becomes highly effective at delaying connectivity by constructing sprawling, loop-less, tree-like components. While this dramatically shifts $p_c^{\text{conn}}$ to higher densities, the resulting giant component is structurally maximally floppy. Consequently, a substantial influx of additional edges is required to rigidify this hyper-sparse network, causing the gap $\Delta p_c = p_c^{\text{rigid}} - p_c^{\text{conn}}$ to rebound for $k > 16$. 

This demonstrates that optimal global rigidification is achieved not by infinite choice, but by an intermediate ``Goldilocks'' parameter ($k_{\text{opt}} \approx 16$) that perfectly balances connectivity delay with internal network density.

\section{Video captions}

\noindent \textbf{Video 1 (nn\_chi\_largeL):} Macroscopic susceptibility scaling for the NN model. The video animates the injection of increasingly larger system sizes (up to $L=192$, $N \approx 7 \times 10^6$) for selected choice parameters ($k=1, 2, 8, 32$). This definitively proves the stabilization of the critical exponent to $\gamma=1.000$ at large scales for $k \ge 8$, validating the first-order transition.

\noindent \textbf{Video 2 (intra\_chi\_largeL):} Macroscopic susceptibility scaling for the Intra model. Similar to the NN model, the animation demonstrates the asymptotic stabilization of the finite-size scaling exponent to exactly $\gamma=1.000$, confirming the explosive nature of the transition in highly coordinated networks up to $L=192$.

\noindent \textbf{Video 3 (nn\_order):} The order parameter, $S_{\max}/N$, versus the edge density $p$ for the NN model. The video shows a dramatic sharpening of the transition as the choice parameter $k$ is varied from 1 to 32. For each $k$, curves are shown for system sizes from $L=1$ to $L=10$.

\noindent \textbf{Video 4 (intra\_order):} The order parameter, $S_{\max}/N$, versus the edge density $p$ for the Intra model. The video shows the evolution of the transition curves as the choice parameter $k$ is varied from 1 to 32. Similar to the NN model, the sharpening of the transition for larger $k$ is clearly visible. For each $k$, curves are shown for system sizes ranging from $L=1$ (lightest color) to $L=10$ (darkest color).

\noindent \textbf{Video 5 (nn\_sus):} The exclusive susceptibility, $\chi'$, versus the edge density $p$ for the NN model. The video shows how the susceptibility peak sharpens, increases in height, and shifts to a higher density as the choice parameter $k$ is varied from 1 to 32. Each panel displays curves for system sizes from $L=1$ to $L=10$.

\noindent \textbf{Video 6 (intra\_sus):} The exclusive susceptibility, $\chi'$, versus the edge density $p$ for the Intra model. The video shows how the susceptibility peak sharpens, increases in height, and shifts to a higher density as the choice parameter $k$ is varied from 1 to 32, consistent with the behavior of the NN model. Each panel displays curves for system sizes from $L=1$ to $L=10$.

\noindent \textbf{Video 7 (nn\_chi):} Small-scale susceptibility scaling ($\ln(\chi'_{\max})$ versus $\ln(N)$) for the NN model. The video illustrates how the scaling relationship and its linear fit evolve as the choice parameter $k$ is varied from 1 to 32. Each set of points represents system sizes from $L=1$ to $L=10$.

\noindent \textbf{Video 8 (intra\_chi):} Small-scale susceptibility scaling for the Intra model. The video illustrates how the scaling relationship and its linear fit evolve as the choice parameter $k$ is varied from 1 to 32. Each set of points represents system sizes from $L=1$ to $L=10$.

\noindent \textbf{Video 9 (intra\_rigiditygap):} Comparison of the connectivity (solid lines) and rigidity (dashed lines) transitions for the Intra model. The video shows the evolution of both transitions as the choice parameter $k$ is varied from 1 to 32, illustrating the persistent and significant gap between them. Crucially, the video demonstrates the optimal efficiency zone around $k \approx 16$, where the gap reaches its minimum before widening again at larger $k$. Each set of curves corresponds to system sizes from $L=1$ to $L=10$.

\end{document}